\documentclass[10pt]{article}
\usepackage[utf8]{inputenc}
\usepackage{comment}
\usepackage{graphicx} % Required for inserting images
\usepackage{booktabs}
\usepackage[left=1in, bottom=1in, right=1in, top=1in]{geometry}
\usepackage{amsmath, bm}
\usepackage{longtable}
\usepackage{multirow}
\usepackage{algorithm}
\usepackage{mathrsfs}
\usepackage{algpseudocode}
\usepackage{amsfonts}

\usepackage[dvipsnames]{xcolor}
\usepackage{enumitem}
\usepackage{wrapfig}
\usepackage{algpseudocode}
\usepackage{pgfplots}
\pgfplotsset{compat=1.18}
\usepackage{amsthm}
\usepackage{caption}
\usepackage{csquotes}
\usepackage{hyperref}

\usepackage{tikz}
\usetikzlibrary{positioning}
\usepackage{setspace}
\newcommand{\barsim}{%
  \mathrel{%
    \tikz[baseline=-0.4ex]{
      \node (b) {$\lvert\mkern2mu\lvert$};
      \node[below=-1.5ex of b] {$\sim$};
    }%
  }%
}

\usepackage{rotating}

\newtheorem{theorem}{\noindent Theorem}

\newtheorem{corollary}{\noindent Corollary}

\newtheorem{lem}{\noindent Lemma}[section]
\newtheorem{cor}{\noindent Corollary}[section]

\newtheorem{assumpt}{\noindent Assumption}[section]

{\newtheorem{example}{\noindent Example}}
{}
{}
{}
{}
{\newtheorem{remark}{\noindent Remark}}
{}
{}
{}
{ }
{}
{}
{}
{}
{}
{}
{}
{}

\numberwithin{equation}{section}
\newcommand{\fst}[1]{\left(#1 \right)}
\newcommand{\secnd}[1]{\left\{#1 \right\}}
\newcommand{\thrd}[1]{\left[#1 \right]}
\newcommand{\E}{\mathbb{E}}

\newcommand{\vech}{\operatorname{vech}}

\hypersetup{
	colorlinks=true,
	linkcolor=blue,
	filecolor=magenta,      
	urlcolor=blue,
	citecolor=blue
}

\usepackage[citestyle=authoryear-comp,
    maxcitenames=2,
    giveninits=true,
    uniquename=init,
    uniquelist=false,
	ibidtracker=context,
    bibstyle=authoryear,
    dashed=false, % dashed: substitute rep. author with -
    date=year,
	sorting=nyt,
    minbibnames=3,
	hyperref=true,
	backref=true,
	citecounter=true,
	citetracker=true,
    labelnumber=true,
    natbib=true, % natbib compatibility mode (\citep and \citet still work)
	backend=biber, % Compile the bibliography with biber
]{biblatex}

\DeclareFieldFormat{citehyperref}{%
  \DeclareFieldAlias{bibhyperref}{noformat}% Avoid nested links
  \bibhyperref{#1}}

\DeclareFieldFormat{textcitehyperref}{%
  \DeclareFieldAlias{bibhyperref}{noformat}% Avoid nested links
  \bibhyperref{%
    #1%
    \ifbool{cbx:parens}
      {\bibcloseparen\global\boolfalse{cbx:parens}}
      {}}}

\savebibmacro{cite}
\savebibmacro{textcite}

\renewbibmacro*{cite}{%
  \printtext[citehyperref]{%
    \restorebibmacro{cite}%
    \usebibmacro{cite}}}

\renewbibmacro*{textcite}{%
  \ifboolexpr{
    ( not test {\iffieldundef{prenote}} and
      test {\ifnumequal{\value{citecount}}{1}} )
    or
    ( not test {\iffieldundef{postnote}} and
      test {\ifnumequal{\value{citecount}}{\value{citetotal}}} )
  }
    {\DeclareFieldAlias{textcitehyperref}{noformat}}
    {}%
  \printtext[textcitehyperref]{%
    \restorebibmacro{textcite}%
    \usebibmacro{textcite}}}
\DeclareAutoCiteCommand{inline}{\textcite}{\parencite}

\title{Robust Inferential Methodology for \\Multidimensional Diffusion Processes}
\author{Sourojyoti Barick \\ Interdisciplinary Statistical Research Unit, Indian Statistical Institute, Kolkata}       
\begin{document}

\maketitle
\begin{abstract}
We investigate robust parameter estimation and testing procedure for multivariate diffusion processes observed at high frequency via the minimum density power divergence estimator (MDPDE). Within a general diffusion framework and under standard regularity conditions, we establish consistency and asymptotic normality for the estimators of both drift and diffusion parameters. The drift estimator converges at the $\sqrt{n h_n}$ rate, whereas the diffusion estimator attains the standard $\sqrt{n}$ rate, and the two estimators are shown to be asymptotically independent. The proposed methodology constitutes a robust alternative to quasi-likelihood and ordinary least squares based approaches, offering resilience against outliers, local contamination, and mild model misspecification, while remaining asymptotically equivalent to classical methods in the absence of contamination. Simulation studies demonstrate that the MDPDE achieves reliable finite-sample performance and enhanced numerical stability relative to likelihood-based estimators. These results underscore the practical relevance of divergence-based estimation for high-frequency diffusion models and point to natural extensions to more complex continuous-time settings.

\end{abstract}

\section{Introduction}

Stochastic models form a cornerstone of  statistical inference and play a fundamental role in a wide range of applications, including finance, econometrics, queuing systems, epidemiology and the natural sciences. Inference for such models is traditionally based on likelihood or contrast-based methods, which typically assume a correct specification of the underlying data-generating mechanism; see, for instance, \cite{IbragimovHasminskii1981,Kessler1997,Uchida2010Contrast}. Under ideal conditions, these procedures enjoy desirable efficiency properties. However, their performance may deteriorate substantially in realistic settings where the observed data deviate from the assumed model, particularly due to the presence of outliers.

In such settings, outliers frequently arise as a consequence of sudden structural shocks or rare extreme events. Representative examples include abrupt market downturns triggered by large-scale disruptions such as the COVID--19 pandemic, sharp corrections following speculative bubbles in technology-driven markets, or sporadic high-magnitude jumps induced by unexpected policy interventions or systemic failures. Such atypical observations may appear as sudden level shifts, unusually large increments, or isolated extreme jumps in the observed trajectories. Although these events occur infrequently, their influence on inference can be disproportionate, often resulting in severe bias and instability of classical likelihood-based estimators, especially for drift and volatility parameters.

These considerations naturally motivate the development of robust inference procedures that remain stable under small departures from the assumed model while retaining high efficiency under the nominal specification. Robust statistical methodology is well established in the framework of independent and identically distributed (i.i.d.) observations. However, extending robustness principles to dependent data structures, and in particular to stochastic processes and diffusion models, presents substantial theoretical and methodological challenges. These difficulties arise from temporal dependence, continuous-time dynamics, and the frequent unavailability of closed-form transition densities; see, among others, \cite{GenonCatalotJacod1993,Kessler1997}.

Among modern robust estimation techniques, divergence-based methods have emerged as a particularly versatile and powerful class of procedures. A prominent example is the density power divergence introduced by \cite{Basu1998MDPDE}, which gives rise to the minimum density power divergence estimator (MDPDE). This estimator is indexed by a tuning parameter that provides a continuous trade-off between efficiency and robustness, recovering the maximum likelihood estimator as a limiting case. While the theoretical properties of the MDPDE are now well understood in the i.i.d.\ setting, their extension to dependent data and continuous-time stochastic models is a more recent and active area of research.

Motivated by these developments, the present paper investigates robust divergence-based inference for multidimensional diffusion processes observed at discrete time points. Our primary objective is to develop an MDPDE framework tailored to this setting, establish its large-sample properties, and demonstrate its robustness advantages relative to classical likelihood-based estimators in the presence of outliers and model deviations.

The rest of the paper is organized as follows. Section~\ref{sec:lit_survey} presents a review of the existing literature relevant
to the present study. In Section~\ref{sec:mdpd}, after a brief introduction to the minimum density power
divergence estimator (MDPDE) in subection~\ref{subsec:mdpd_intro}, we develop the estimation methodology and state the main theoretical results in subection~\ref{subsec:mdpd_estm_main_results}, including all principal theorems and
corollaries. The finite-sample performance of the proposed estimators is examined through simulation studies in section~\ref{sec:simulation}. Finally, proofs of all results are collected in the Appendix and relevant desiderata in a Supplementary section.

\section{Literature Survey}\label{sec:lit_survey}

Early developments in robust inference for stochastic processes were primarily driven
by quasi-likelihood and M-estimation methodologies. A foundational contribution was
made by \cite{KulkarniHeyde1987}, who introduced a general framework for constructing
robust quasi-likelihood estimating functions for discrete-time stochastic processes.
Subsequently, robust M-estimators for Poisson processes with periodic intensities were
proposed in \cite{YoshidaHayashi1990}, together with a detailed analysis of their
asymptotic properties. More general point process models were  explored
in \cite{AssuncaoGuttorp1995,AssuncaoGuttorp1999}, where robust estimation procedures were
developed for inhomogeneous Poisson and general point processes. In the context of
branching processes, \cite{Stoimenova2005} introduced robust and efficient estimators
under power series offspring distributions, with extensions provided in
\cite{StoimenovaAtanasov2011}. Robust estimation of the memory parameter for an infinite
source Poisson process, corresponding to a special case of the $M/G/\infty$ queue, was
studied in \cite{Fay2007}. 

A major advance in robust statistical methodology was the introduction of the density
power divergence (DPD) by \cite{Basu1998MDPDE}, which provides a flexible class of
divergence measures indexed by a tuning parameter controlling the trade-off between
efficiency and robustness. The resulting minimum density power divergence estimator
(MDPDE) has been shown to possess strong robustness properties while retaining high
efficiency under the assumed model. The MDPDE framework has since been extended beyond the i.i.d.\ setting. For discretely
observed diffusion processes, initial developments appeared in \cite{Song2007MDPDE},
with a rigorous asymptotic theory, including consistency and asymptotic normality,
established in \cite{LeeSong2013}. Further contributions include robust estimation of
dispersion parameters in univariate diffusion models \cite{Song2017RobustDiffusion} and
divergence-based inference for more general discretely observed univariate stochastic processes
\cite{HoreGhosh2022MDPDEstochastic}. A broad overview of divergence-based robust
procedures is provided in \cite{PardoMartin2021Divergence}. Collectively, these works
demonstrate the effectiveness of divergence-based methods as a tractable and robust
alternative to likelihood-based inference in the presence of contamination or model
misspecification.

Robust estimation has also received considerable
attention within the diffusion framework. Early work by \cite{Yoshida1988} investigated efficient robust estimation via
M-estimation techniques. Optimally robust asymptotic linear estimators for the
Ornstein--Uhlenbeck process were developed in \cite{Rieder2012}, while robust parameter
estimation for the Weibull process was considered in \cite{WangBebbington2013}. These
studies laid the groundwork for robust inference in continuous-time stochastic
models.

Statistical inference for multidimensional diffusion processes introduces additional
challenges due to high dimensionality and complex dependence structures. Early work on
estimation of diffusion coefficients in multivariate settings can be traced back to
\cite{GenonCatalotJacod1993}. Nonparametric approaches for estimating diffusion
coefficients in multidimensional time-varying models were proposed in
\cite{WangChen2020MultidimDiffusionCoeff}, while recent Bayesian nonparametric
developments for multidimensional diffusion models were studied in
\cite{HoffmannRay2024MultidimDiffusion}. 

% Standard references on multivariate analysis
% and matrix calculus relevant to multidimensional diffusion modeling include
% \cite{MardiaKentBibby1979,Anderson2003,MagnusNeudecker1988,PetersenPedersen2012}, with
% probabilistic foundations summarized in \cite{SchillingPartzsch2014}.

In parallel, a substantial body of work has contributed to robust estimation and hypothesis testing for multivariate
dependent data, providing important methodological foundations for inference in multidimensional
diffusion settings. Robust Wald-type tests and divergence-based testing procedures for
multivariate parametric models were developed in
\cite{BasuGhoshMandol2016,ghosh2013robust}, while extensions to composite hypotheses and
multivariate estimating equations were studied in \cite{GhoshBasuPardo2013}. Robust
testing and estimation for multivariate normality, covariance structures were further examined in
\cite{FujisawaEguchi2008,GhoshBasuPardo2013,BasuChakrabortyGhoshPardo2022}. Although these
contributions are primarily formulated for static or discrete-time multivariate
models, their methodological insights naturally extend to discretely observed
multidimensional diffusion processes through Gaussian or locally Gaussian
approximations of transition densities.

Despite the expanding literature on robust inference for diffusion processes,
divergence-based estimation for multidimensional diffusion models remains
comparatively underdeveloped. Existing results for the MDPDE are largely confined to
univariate diffusions or restrictive parametric structures. The present work aims to
bridge this gap by developing a unified MDPDE framework for multidimensional diffusion
processes, establishing its asymptotic properties and demonstrating its robustness
advantages over classical likelihood-based inference.

\section{MDPDE for Multidimensional Diffusion Process}\label{sec:mdpd}
\subsection{Minimum Density Power Divergence Estimator (MDPDE)}\label{subsec:mdpd_intro}
We begin by recalling the minimum density power divergence estimation (MDPDE)
framework, which provides the methodological foundation for the robust inference
procedures developed in this paper. This brief review is included both for
completeness and to motivate the subsequent extension to discretely observed
multidimensional diffusion processes.

% \subsection{Density Power Divergence}

The density power divergence (DPD), introduced by
\citet{Basu1998MDPDE}, defines a parametric family of discrepancy measures between
two probability density functions $g$ and $f$, indexed by a tuning parameter
$\alpha \ge 0$. It is given by
\begin{equation}\label{eq:dpd_def}
d_\alpha(g,f)
=
\begin{cases}
\displaystyle
\int
\left\{
f^{1+\alpha}(z)
-
\left(1+\frac{1}{\alpha}\right)
g(z) f^\alpha(z)
+
\frac{1}{\alpha} g^{1+\alpha}(z)
\right\}
\,dz,
& \alpha>0,
\\[1.2ex]
\displaystyle
\int g(z)\{\log g(z)-\log f(z)\}\,dz,
& \alpha = 0.
\end{cases}
\end{equation}
Although $d_\alpha(\cdot,\cdot)$ is not directly defined at $\alpha=0$, the limiting
case coincides with the Kullback--Leibler divergence. At the other extreme,
$\alpha=1$ yields the squared $L_2$ distance between densities. The tuning
parameter $\alpha$ thus governs the efficiency--robustness trade-off: larger
values of $\alpha$ downweight observations that are unlikely under the postulated
model, thereby enhancing robustness against outliers and model misspecification.

% \subsection{MDPDE for Parametric Models}

Let $\{F_{\bm\theta} : \bm\theta \in \Theta \subset \mathbb{R}^p\}$ be a parametric
family of distributions with corresponding densities
$\{f(\cdot;\bm\theta)\}$, and let $G$ denote a distribution with density $g$.
~\cite{Basu1998MDPDE} defines the \emph{minimum density power divergence
functional} $T_\alpha(\cdot)$ by
\begin{equation}\label{eq:T_alpha_def}
d_\alpha\!\left(g, f(\cdot;T_\alpha(g))\right)
=
\min_{\bm\theta \in \Theta}
d_\alpha\!\left(g, f(\cdot;\bm\theta)\right).
\end{equation}
If the true distribution $G$ belongs to the model family, then
$T_\alpha(g)=\bm\theta_0$ for the corresponding true parameter
$\bm\theta_0\in\Theta$.

Given an i.i.d.\ sample $Y_1,\ldots,Y_n$ from density $g$, the empirical version of $T_\alpha(g)$ leads to the MDPDE
\begin{equation}\label{eq:mdpde_iid}
\hat{\bm\theta}_{\alpha,n}
=
\arg\min_{\bm\theta \in \Theta}
H_{n,\alpha}(\bm\theta),
\end{equation}
where the objective function takes the form
\[
H_{n,\alpha}(\bm\theta)
=
\frac{1}{n}
\sum_{i=1}^n
\left[
\int f^{1+\alpha}(y;\bm\theta)\,dy
-
\left(1+\frac{1}{\alpha}\right)
f^\alpha(Y_i;\bm\theta)+\frac{1}{\alpha}
\right],
\qquad \alpha>0,
\]
with the limiting case $\alpha=0$ corresponding to the negative log-likelihood.
A notable feature of this formulation is that it avoids explicit nonparametric
density estimation: the term involving $g^{1+\alpha}$ in
\eqref{eq:dpd_def} does not depend on $\bm\theta$ and can be omitted, while the
remaining integral with respect to $g$ is evaluated via the empirical
distribution.

Under suitable regularity conditions, $\hat{\bm\theta}_{\alpha,n}$ is weakly
consistent and asymptotically normal for $T_\alpha(g)$
\citep{Basu1998MDPDE,Basu2011}. Moreover, for $\alpha>0$, the estimator exhibits
strong robustness properties with only a modest loss in asymptotic efficiency
relative to the maximum likelihood estimator.

% \subsection{Independent Non-Homogeneous and Regression-Type Extensions}

The MDPDE framework extends naturally to independent but non-homogeneous (INH)
settings \citep{ghosh2013robust}. In this case, the observations
$Y_1,\ldots,Y_n$ are independent with $Y_i \sim g_i$, and the model is specified
through families
\[
\mathscr F_i
=
\left\{
f_i(\cdot;\bm\theta) : \bm\theta \in \Theta
\right\},
\qquad i=1,\ldots,n,
\]
sharing a common parameter $\bm\theta$. The corresponding estimator minimizes the
average density power divergence, leading to the objective function
\begin{equation}\label{eq:mdpde_inh}
H_{n,\alpha}(\bm\theta)
=
\frac{1}{n}
\sum_{i=1}^n
\left[
\int f_i^{1+\alpha}(y;\bm\theta)\,dy
-
\left(1+\frac{1}{\alpha}\right)
f_i^\alpha(Y_i;\bm\theta)+\frac{1}{\alpha}
\right].
\end{equation}

A particularly important special case is that of regression-type models. Let
$\{f_\theta(y \mid x)\}$ be a parametric family of conditional densities for $Y$
given $X=x$, and let $g(y \mid x)$ denote the true conditional density. Replacing
$f$ and $g$ in \eqref{eq:dpd_def} by $f_\theta(\cdot\mid x)$ and
$g(\cdot\mid x)$ yields a conditional density power divergence. Given observations
$(X_i,Y_i)$, the resulting estimator is defined by
\[
\arg\min_{\bm\theta \in \Theta}
\begin{cases}
\displaystyle
\frac{1}{n}\sum_{i=1}^n
\int f_\theta^{1+\alpha}(y \mid X_i)\,dy
-
\left(1+\frac{1}{\alpha}\right)
\frac{1}{n}\sum_{i=1}^n
f_\theta^\alpha(Y_i \mid X_i)+\frac{1}{\alpha},
& \alpha>0,
\\[2ex]
\displaystyle
-\frac{1}{n}\sum_{i=1}^n
\log f_\theta(Y_i \mid X_i),
& \alpha=0.
\end{cases}
\]

Compared to other density-based divergence methods, such as those relying on the
Hellinger distance \citep{Beran1977,TamuraBoos1986,Simpson1987} or smoothed
likelihood-type divergences \citep{BasuLindsay1994,Cao1995}, the density power
divergence enjoys the practical advantage of not requiring any smoothing
procedures. In particular, it avoids bandwidth selection and other numerical
difficulties inherent in kernel-based approaches.

This regression-type and INH formulation provides the conceptual basis for the
robust estimation procedures developed in this paper for discretely observed
multidimensional diffusion processes.

% \section{Hhhhhhh}
\subsection{Estimation and Main Results}\label{subsec:mdpde-mv}\label{subsec:mdpd_estm_main_results}

We now extend the MDPDE framework
to multidimensional diffusion processes observed at discrete time points. The
primary objective is to develop a robust estimation procedure for both the drift
and diffusion components in a multivariate setting, allowing for high-frequency
sampling over an increasing time horizon.

Throughout this section, we work under the \emph{null hypothesis} that the data
are generated by a continuous diffusion model with constant diffusion matrix.
Specifically, let $(X_t)_{t \ge 0}$ be a $d$-dimensional diffusion process defined
on a filtered probability space
$(\Omega,\mathcal{F},(\mathcal{F}_t)_{t\ge0},\mathbb{P})$ and satisfying the
stochastic differential equation
\begin{equation}\label{eq:gen_multivariate_diffusion_gen}
dX_t
=
a(X_t,\bm\beta_0)\,dt
+
\Sigma_0^{1/2}\,dW_t,
\end{equation}
where $a:\mathbb{R}^d \times \Theta_{\bm\beta} \to \mathbb{R}^d$ is a measurable
drift function depending on an unknown parameter
$\bm\beta_0 \in \Theta_{\bm\beta} \subset \mathbb{R}^p$,
$\Sigma_0 \in \mathbb{R}^{d\times d}$ is a symmetric positive definite diffusion
matrix, and $(W_t)_{t \ge 0}$ is a $d$-dimensional standard Brownian motion.

In the statistical model, the diffusion matrix $\Sigma$ is treated as an
\emph{unknown parameter} to be estimated from the data, while $\Sigma_0$
denotes its true (but unobserved) value under the data-generating mechanism.
We write
\[
\Lambda = \Sigma^{1/2}
\quad\text{and}\quad
\Lambda_0 = \Sigma_0^{1/2}
\]
for the unique symmetric positive definite square roots of $\Sigma$ and
$\Sigma_0$, respectively. We further assume that $(X_t)_{t \ge 0}$ is ergodic
with invariant probability measure $\mu_0$ corresponding to the true parameter
value $\bm\theta_0$.

The full parameter vector is defined as
\[
\bm\theta
=
\bigl(
\bm\beta,
\operatorname{vech}(\Sigma)
\bigr)
\in \Theta,
\]
with true value
\[
\bm\theta_0
=
\bigl(
\bm\beta_0,
\operatorname{vech}(\Sigma_0)
\bigr).
\]
Hence, the total dimension of the parameter space is
$p + d^*$ where $d^* = \frac{d(d+1)}{2}$.

The process is observed at discrete time points
$t_i^n = i h_n$, $i = 0,1,\ldots,n$, where the sampling scheme satisfies the
high-frequency asymptotic regime
\[
h_n \to 0,
\qquad
n h_n \to \infty,
\qquad
n \to \infty.
\]

Under the null diffusion model \eqref{eq:gen_multivariate_diffusion_gen}, applying
the Euler--Maruyama approximation yields
\begin{equation}\label{eq:euler_multivariate}
X_{t_i^n}
=
X_{t_{i-1}^n}
+
a(X_{t_{i-1}^n},\bm\beta_0)\,h_n
+
\Lambda_0 \sqrt{h_n}\, Z_{n,i}
+
\Delta_{n,i},
\end{equation}
where
\[
Z_{n,i}
=
\frac{W_{t_i^n} - W_{t_{i-1}^n}}{\sqrt{h_n}}
\sim \mathcal{N}_d(0,I_d),
\]
and the remainder term is given by
\[
\Delta_{n,i}
=
\int_{t_{i-1}^n}^{t_i^n}
\bigl\{
a(X_s,\bm\beta_0)
-
a(X_{t_{i-1}^n},\bm\beta_0)
\bigr\}\,ds.
\]

Consequently, under suitable smoothness and regularity conditions on the drift
function $a$ 
% which will be 
(stated below), the remainder term $\Delta_{n,i}$ is
negligible uniformly in $\boldsymbol{\theta}$. As a result, the conditional
distribution of $X_{t_i^n}$ given $X_{t_{i-1}^n}$ admits a valid local Gaussian
approximation. In particular, $X_{t_i^n} \mid X_{t_{i-1}^n}$ is well approximated by
a multivariate normal distribution with mean
\[
X_{t_{i-1}^n} + a\!\left(X_{t_{i-1}^n}, \bm\beta_0\right) h_n
\]
and covariance matrix $h_n \Sigma_0$.

% Under standard regularity conditions on $a$,stated later, the remainder term $\Delta_{n, i}$ is negligible uniformly in $\boldsymbol{\theta}$, yielding a valid local Gaussian approximation for the conditional distribution of $X_{t_i^n} \mid X_{t_{i-1}^n}$ with mean $X_{t_{i-1}^n}+a\left(X_{t_{i-1}^n}, \beta_0\right) h_n$ and covariance $h_n \Sigma_0$.

Let $f_{\bm\theta}(\cdot\mid X_{t_{i-1}^n})$ denote the corresponding Gaussian
conditional density. Based on this local Gaussian approximation, the MDPDE for
multivariate diffusions is defined as
\begin{equation}\label{eq:mdpde_mv_min}
\hat{\bm\theta}_n^\alpha
=
\arg\min_{\bm\theta\in\Theta}
\frac{1}{n}
\sum_{i=1}^n
V_{n,i}^\alpha(\bm\theta),
\end{equation}
where, for $\alpha>0$,
\begin{align}\label{eq:mdpd_multivariate}
V_{n,i}^\alpha(\bm\theta)
&=
\frac{1}{(1+\alpha)^{d/2}}
|\Sigma|^{-\alpha/2}
\nonumber\\
&\quad
-
\left(1+\frac{1}{\alpha}\right)
|\Sigma|^{-\alpha/2}
\exp\!\left(
-\frac{\alpha}{2h_n}
R_i(\bm\beta)^\top
\Sigma^{-1}
R_i(\bm\beta)
\right),
\end{align}
with residuals
\[
R_i(\bm\beta)
=
X_{t_i^n}
-
X_{t_{i-1}^n}
-
a(X_{t_{i-1}^n},\bm\beta)h_n.
\]
Equivalently,
\begin{align}
R_i(\bm\beta)
=
\sqrt{h_n}\Lambda_0 Z_{n,i}
+
h_n D_i(\bm\beta)
+
\Delta_{n,i},
\qquad
D_i(\bm\beta)
=
a(X_{t_{i-1}^n},\bm\beta_0)
-
a(X_{t_{i-1}^n},\bm\beta).
\label{eqn:R_i_D_i}
\end{align}
Terms depending only on $(2\pi)$ and $h_n$ have been omitted, as they do not affect
the minimization problem in \eqref{eq:mdpde_mv_min}. The criterion
$V_{n,i}^\alpha(\bm\theta)$ therefore corresponds to the density power divergence
constructed from the local Gaussian approximation to the transition density.

For later use, define the class of functions with polynomial growth as
\[
    \mathscr P
    :=
    \Bigl\{
        f : \mathbb{R}^{d} \times \Theta \to \mathbb{R}^{d_1}
        :
        \exists\, C_{d_1} > 0 \text{ independent of } {\bm\theta}
        \text{ such that }
        \|f(x,{\bm\theta})\|
        \le
        C_{d_1}
        \left(1 + \|x\|\right)^{C_{d_1}},
        \ \forall\, {\bm\theta}
    \Bigr\}.
\]

\begin{assumpt}\label{assm:all_assumption}
Let $\{X_t\}_{t\ge0}$ be a $d$-dimensional diffusion process with true parameter
${\bm\theta}_0=({\bm\beta_0},\operatorname{vech}(\Sigma_0))$. Assume:

\begin{itemize}
\item[(A1)]
The drift $a(\cdot,{\bm\beta_0})$ is globally Lipschitz, i.e.,
\[
\|a(x,{\bm\beta_0})-a(y,{\bm\beta_0})\|\le C\|x-y\|,
\qquad x,y\in\mathbb{R}^d .
\]

\item[(A2)]
Under ${\bm\theta}_0$, the process is ergodic with invariant probability measure
$\mu_0$, and
\[
\int_{\mathbb{R}^d}\|x\|^k\,\mu_0(dx)<\infty,
\qquad \forall\,k\ge0 .
\]

\item[(A3)]
For all $k\ge0$,
\[
\sup_{t\ge0}\mathbb{E}_{{\bm\theta}_0}\!\left[\|X_t\|^k\right]<\infty .
\]

\item[(A4)]
For each $\beta\in\Theta$, the mapping $x\mapsto a(x,\bm\beta)$ is continuously
differentiable, and $a(\cdot,\bm\beta)$ together with its first-order derivatives
belongs to the class $\mathscr{P}$.

\item[(A5)]
For each $x\in\mathbb{R}^d$, the mapping $\beta\mapsto a(x,\bm\beta)$ and its derivative with respect to x are three times
continuously differentiable, and all partial derivatives up to order three belong
to $\mathscr{P}$.

\item[(A6)]
(Identifiability) If
\[
a(x,\bm\beta)=a(x,{\bm\beta_0})
\quad \text{for }\mu_0\text{-a.e. }x,
\]
then $\bm \beta={\bm\beta_0}$.

\item[(A7)]
The diffusion matrix $\Sigma_0$ is symmetric and positive definite, and the matrix
\[
\mathcal{B}
=
\int
\fst{\frac{\partial a(x,{\bm\beta_0})}{\partial{\bm\beta}} }^{\top}
\Sigma_0^{-1}
\fst{\frac{\partial a(x,{\bm\beta_0})}{\partial{\bm \beta}} }
\,\mu_0(dx)
\]
is positive definite.
\end{itemize}
\end{assumpt}

\begin{remark}
Assumptions \textup{(A1)}--\textup{(A3)} ensure the existence of a unique ergodic
solution possessing sufficiently high-order moments, while
Assumptions \textup{(A4)}--\textup{(A5)} provide the smoothness and regularity
conditions required for uniform laws of large numbers. Identifiability of the
model parameters is guaranteed by Assumption \textup{(A6)}.
\end{remark}

\begin{example}
An important example of the model \eqref{eq:gen_multivariate_diffusion_gen}
is the multivariate Ornstein--Uhlenbeck process
\begin{align*}
    dX_t = A X_t\,dt + \Lambda_0\, dW_t ,
\end{align*}
where $A\in\mathbb{R}^{d\times d}$ and $\Lambda_0\Lambda_0^\top=\Sigma_0$.
If the eigenvalues of $A$ have strictly negative real parts, then the
process admits a unique stationary Gaussian distribution and is ergodic.
The drift $a(x,\beta)=Ax$ is globally Lipschitz and infinitely
differentiable in both $x$ and $\beta$. 
Moreover, identifiability reduces to the condition that the parameter
matrix $A$ is uniquely determined from $a(x,\beta)=Ax$,
which holds whenever the parameterization of $A$ has full column rank.
Hence Assumptions~\ref{assm:all_assumption} (A1)--(A7) are satisfied.
\end{example}

Under Assumptions~\ref{assm:all_assumption}, the diffusion process
\eqref{eq:gen_multivariate_diffusion_gen} is ergodic and the drift
function satisfies the required smoothness and identifiability
conditions. These properties allow us to establish the consistency
of the proposed estimator. The following theorem states that the
MDPDE estimator converges in probability to the true parameter.

\begin{theorem}[Consistency of parameters]\label{thm:mdpde-consistency}
Suppose that Assumptions~\ref{assm:all_assumption} \textup{(A1)}--\textup{(A6)}
hold. Let $\alpha>0$ be fixed. If $h_n \to 0$, $n h_n \to \infty$ and $nh_n^q\to0$ for some $q>1$, then
\[
    \hat{{\bm\theta}}_n^\alpha
    \xrightarrow{\mathbb P}
    {\bm\theta}_0.
\]
\end{theorem}

To establish the asymptotic normality of the proposed multivariate estimator,
we impose an additional condition on the sampling frequency of the discretely
observed data.

\begin{assumpt}\label{assm:h_n^2}
In addition to $h_n \to 0$ and $n h_n \to \infty$, we assume that
\[
n h_n^2 \to 0 .
\]
\end{assumpt}

Under this additional condition, the following asymptotic normality result holds.

\begin{theorem}[Asymptotic normality of the drift estimator]
\label{thm:beta-clt}
Suppose that the diffusion covariance matrix $\Sigma$ is known.
Under Assumptions~\ref{assm:all_assumption}~\textup{(A1)}--\textup{(A7)}
together with Assumption~\ref{assm:h_n^2}, we have
\begin{align}
\sqrt{n h_n}\bigl(\hat{\bm\beta}_n^\alpha - \bm\beta_0\bigr)
\;\Rightarrow\;
\mathcal{N}\!\left(0,\Sigma_{\bm\beta}\right),
\label{eqn:thm_2}
\end{align}
where
\[
\Sigma_{\bm\beta}
=
\frac{(1+\alpha)^{d+2}}{(1+2\alpha)^{\frac d2+1}}
\,\mathcal B^{-1}.
\]
\end{theorem}

We now turn to statistical inference for the drift parameter. 
The asymptotic normality established in Theorem~\ref{thm:beta-clt}
naturally leads to a Wald-type test statistic. 
Since the estimator $\hat{\bm\beta}_n^\alpha$ is consistent and the drift
function $a(x,\bm\beta)$ is continuously differentiable with respect to
$\bm\beta$, the following corollary follows directly from standard
Wald theory.

\begin{corollary}[Wald-type test for the drift parameter]
\label{cor:wald}
Under the assumptions of Theorem~\ref{thm:beta-clt},
\[
    n h_n\,(\hat{\bm\beta}_n^\alpha - {\bm\beta}_0)^{\top}
    \hat\Sigma_\beta^{-1}
    (\hat{\bm\beta}_n^\alpha - {\bm\beta}_0)
    \;\Rightarrow\;
    \chi^2_p,
\]
where $p=\dim(\bm\beta)$ and $\hat\Sigma_\beta$ is a consistent estimator of the
asymptotic covariance matrix $\Sigma_\beta$, given by
\[
    \hat\Sigma_\beta
    =
    \frac{(1+\alpha)^{d+2}}{(1+2\alpha)^{\frac d2+1}}
    \,\hat{\mathcal B}^{-1}.
\]
Here, the matrix $\hat{\mathcal B}=(\hat{\mathcal B}_{u,v})_{1\le u,v\le p}$ is defined by
\[
    \hat{\mathcal B}_{u,v}
    =
    \int
    \left(\frac{\partial a(x,\hat{\bm\beta}_n^\alpha)}{\partial \beta_u}\right)^{\!\top}
    \hat\Sigma^{-1}
    \left(\frac{\partial a(x,\hat{\bm\beta}_n^\alpha)}{\partial \beta_v}\right)
    \,\mu_0(dx),
\]
where $\hat\Sigma$ denotes a consistent estimator of the true covariance matrix
$\Sigma_0$.
\end{corollary}

We next investigate the weak convergence of the estimator of the diffusion
covariance matrix. To explicitly characterize the asymptotic
variance--covariance structure, we analyze the first-order derivatives of
$V_{n,i}^\alpha$. After establishing their joint asymptotic normality, the result
follows from an application of Slutsky's theorem. For notational convenience and
clarity, the estimator is vectorized via the $\operatorname{vech}$ operator, and
the limiting variance--covariance matrix is derived elementwise.

\begin{theorem}[Asymptotic normality of the covariance estimator]
\label{thm:asymp-sigma}
Assuming $\beta$ is known and let $\hat\Sigma_n^{\alpha}$ denote the MDPDE of the covariance matrix $\Sigma_0$.
Under the regularity conditions stated in
Assumptions~\ref{assm:all_assumption} \textup{(A1)}--\textup{(A7)}, together with Assumption~\ref{assm:h_n^2}, it holds that
\[
\sqrt{n}\Bigl(
\operatorname{vech}(\hat\Sigma_n^{\alpha})
-
\operatorname{vech}(\Sigma_0)
\Bigr)
\;\Rightarrow\;
\mathcal{N}\!\left(
0,\;
\mathscr{L}^{-\top}\,\Xi\,\mathscr{L}^{-1}
\right).
\]

where the matrices
$\Xi = (\Xi_{kl,rs})$ and $\mathscr L = (\mathscr L_{kl,rs})$ are indexed by
$1 \le k \le l \le d$ and $1 \le r \le s \le d$, and are given elementwise by
\begin{align*}
\Xi_{kl,rs}
&=
(1+\alpha)^{d+2}(1+2\alpha)^{-d/2-2}
\Bigl[
    \alpha^2 \operatorname{tr}(A_{kl}) \operatorname{tr}(A_{rs})
    +
    \frac{1}{2}\operatorname{tr}(A_{kl}A_{rs})
\Bigr]
-
\frac{\alpha^2}{4}
\operatorname{tr}(A_{kl}) \operatorname{tr}(A_{rs}),
\\
\mathscr L_{kl,rs}
&=
\frac{\alpha^2}{4(1+\alpha)}
\operatorname{tr}(A_{kl}) \operatorname{tr}(A_{rs})
+
\frac{1}{2(1+\alpha)}
\operatorname{tr}(A_{kl}A_{rs}).
\end{align*}
Here,
\[
    A_{kl} = \Sigma_0^{-1} S_{kl},
    \qquad
    S_{kl} =
    \begin{cases}
        E_{kk}, & k = l,\\[0.2cm]
        E_{kl} + E_{lk}, & k \neq l,
    \end{cases}
\]
where $E_{kl} = e_k e_l^\top$ and $e_k$ denotes the $k$th standard basis vector in
$\mathbb R^d$.
\end{theorem}

It has been seen in \cite{LeeSong2013}, in the univariate diffusion setting, that the estimator of the
drift parameter is asymptotically independent of the estimator of the diffusion
coefficient. This property continues to hold in the present multivariate
framework, as stated in the following result.
\begin{theorem}[Asymptotic independence]
\label{thm:asymp-indep}
Let $\hat{\bm\beta}_n^{\alpha}$ and $\hat{\Sigma}_n^{\alpha}$ denote the
MDPDEs of the drift parameter $\bm\beta_0$ and the covariance matrix
$\Sigma_0$, respectively.
Under Assumptions~\ref{assm:all_assumption}~\textup{(A1)--(A7)} together with
Assumption~\ref{assm:h_n^2}, we have
\[
\sqrt{n h_n}\bigl(\hat{\bm\beta}_n^{\alpha}-\bm\beta_0\bigr)
\;\barsim\;
\sqrt{n}\,
\operatorname{vech}\!\bigl(\hat{\Sigma}_n^{\alpha}-\Sigma_0\bigr),
\]
where $\barsim$ denotes asymptotic independence.
\end{theorem}

Combining the marginal asymptotic normality results for the drift and diffusion
estimators with their asymptotic independence yields the following joint
Gaussian limit.

\begin{theorem}[Joint asymptotic normality]
\label{thm:joint-clt}
Under Assumptions~\ref{assm:all_assumption}~\textup{(A1)--(A7)} and
Assumption~\ref{assm:h_n^2}, it holds that
\begin{align}
\begin{pmatrix}
\sqrt{n h_n}\bigl(\hat{\bm\beta}_n^\alpha - \bm\beta_0\bigr) \\
\sqrt{n}\Bigl(
\operatorname{vech}(\hat\Sigma_n^{\alpha})
-
\operatorname{vech}(\Sigma_0)
\Bigr)
\end{pmatrix}
\;\Rightarrow\;
\mathcal{N}\!\left(0,\mathscr{B}\right),
\end{align}
where the asymptotic covariance matrix $\mathscr{B}$ is block diagonal with
$\mathscr{B}_{11}^{p\times p}=\Sigma_{\bm\beta}$ and
\[
\mathscr{B}_{22}^{d^*\times d^*}
= \mathscr{L}^{-\top}\,\Xi\,\mathscr{L}^{-1}.
\]
The off-diagonal blocks vanish, reflecting the asymptotic independence of the
drift and diffusion estimators. The explicit expressions for
$\Sigma_{\bm\beta}$, $\mathscr{L}$, and $\Xi$ are provided in
Theorems~\ref{thm:beta-clt} and~\ref{thm:asymp-sigma}.
\end{theorem}

\begin{remark}
The block-diagonal structure of the asymptotic covariance matrix reflects the
fact that drift and diffusion parameters are estimated at different rates and
are asymptotically independent under high-frequency sampling.
\end{remark}

\section{Simulation Study}
\label{sec:simulation}

This section presents a comprehensive simulation study designed to evaluate the finite-sample performance and robustness of the proposed MDPDE for multivariate diffusion-type models under contamination. The simulation framework is deliberately constructed to exploit the close connection between discretely observed diffusion processes and multivariate linear regression or vector autoregressive (VAR) models.

Let $\{X_t\}_{t=0}^{n}$ be a $d$-dimensional stochastic process observed at equidistant time points, where we focus on the bivariate case $d=2$. The data-generating mechanism is based on the Euler--Maruyama discretization of a linear diffusion process,
\begin{equation}
\Delta X_t := X_t - X_{t-1}
= (B X_{t-1} + b)\, h + \varepsilon_t,
\qquad
\varepsilon_t \sim \mathcal{N}(0, \Sigma h),
\label{eq:euler_sim}
\end{equation}
for $t=1,\ldots,n$, where $B \in \mathbb{R}^{2\times2}$ is the drift matrix, $b \in \mathbb{R}^2$ is an intercept vector, $\Sigma$ is a positive-definite diffusion matrix, and $h$ denotes the discretization step size.

Dividing \eqref{eq:euler_sim} by $h$ yields the regression-type representation
\begin{equation}
Y_t := \frac{\Delta X_t}{h}
= B X_{t-1} + b + \eta_t,
\qquad
\eta_t \sim \mathcal{N}\!\left(0, \frac{\Sigma}{h}\right),
\label{eq:var_form}
\end{equation}
which shows that the discretely observed diffusion model admits a natural interpretation as a multivariate linear regression or, equivalently, a VAR(1)-type model with heteroskedastic innovations. Such representations are standard in the analysis of discretely observed diffusions and multivariate time series; see, for example, \citet{Hamilton1994}, \citet{Lutkepohl2005}, and \citet{UchidaYoshida2012}. In empirical work, VAR models of the form \eqref{eq:var_form} are routinely estimated using standard software implementations, such as those provided in the \texttt{vars} package in \textsf{R} \citep{varsR}.

Importantly, instead of estimating \eqref{eq:var_form} via classical Gaussian likelihood or least squares—as is customary in standard VAR analysis—we adopt a robust regression-based approach using the MDPDE. This allows us to retain the VAR structure induced by the diffusion dynamics while achieving robustness against deviations from Gaussianity and the presence of outliers.

The discretization step size is chosen as a decreasing function of the sample size,
$h_n = n^{-0.55}$, reflecting an increasingly fine temporal resolution as $n$
grows. In our setting, this choice satisfies Assumption~\ref{assm:h_n^2} and
maintains an appropriate balance between numerical stability and the asymptotic
scaling of the drift and diffusion components. The true parameter values used throughout the simulation study are
\[
B_{\text{true}} =
\begin{pmatrix}
-0.6 & -0.2 \\
0.1 & -0.4
\end{pmatrix},
\qquad
b_{\text{true}} =
\begin{pmatrix}
2.0 \\
1.0
\end{pmatrix},
\qquad
\Sigma_{\text{true}} =
\begin{pmatrix}
1.0 & 0.5 \\
0.5 & 0.7
\end{pmatrix}.
\]
These values correspond to a stable multivariate autoregressive structure with nontrivial cross-dependence and correlated innovations.

To assess robustness, contamination is introduced additively into the generated trajectories. Let $X_t^{\text{clean}}$ denote the uncontaminated process generated from \eqref{eq:euler_sim}. The contaminated observations are defined as
\begin{equation}
X_t^{\text{contaminated}}
=
X_t^{\text{clean}} + \kappa Z_t,
\qquad
Z_t \sim \mathcal{N}(\bm{0}, I_d),
\label{eq:contamination}
\end{equation}
for a randomly selected fraction $\epsilon$ of time points, where
\[
\epsilon \in \{0, 0.05, 0.10, 0.20\},
\]
and the contamination magnitude is fixed at $\kappa = 5$.

For each simulated dataset, the response and design matrices are constructed as
\[
Y_t = {X_t - X_{t-1}},
\qquad
Z_t = \begin{pmatrix} X_{t-1}^\top & 1 \end{pmatrix},
\quad t=1,\ldots,n.
\]
The resulting multivariate regression model is then estimated using the MDPDE with tuning parameters
\[
\alpha \in \{0, 0.1, 0.3, 0.5\}.
\]
The case $\alpha=0$ corresponds to the classical Gaussian maximum likelihood estimator, while $\alpha>0$ yields increasing robustness against contamination.

\begin{table}
\centering
\resizebox{1\textwidth}{!}{
\begin{tabular}{crrrrrrrrrrr}
  \toprule
$\epsilon$ & alpha & B11 & B12 & B21 & B22 & b1 & b2 & S11 & S12 & S21 & S22 \\ 
  \midrule
0\% & 0.00 & -0.7559 & -0.0735 & 0.1564 & -0.4574 & 1.3993 & 1.3891 & 1.0339 & 0.2131 & 0.2131 & 0.2962 \\ 
  0\% & 0.10 & -0.7046 & -0.1246 & 0.1622 & -0.4631 & 1.6155 & 1.4273 & 1.0388 & 0.2145 & 0.2145 & 0.3015 \\ 
  0\% & 0.30 & -0.6296 & -0.1992 & 0.1708 & -0.4720 & 1.9359 & 1.4996 & 1.0456 & 0.2120 & 0.2120 & 0.3041 \\ 
  0\% & 0.50 & -0.5771 & -0.2537 & 0.1749 & -0.4772 & 2.1940 & 1.5621 & 1.0567 & 0.2034 & 0.2034 & 0.2967 \\ \midrule
  5\% & 0.00 & -2.2574 & 1.1019 & 3.3789 & -4.0304 & -3.0000 & 19.8483 & 18.8673 & 7.0632 & 7.0632 & 37.6210 \\ 
  5\% & 0.10 & -0.6862 & -0.1391 & 0.2194 & -0.5230 & 1.6244 & 1.6735 & 1.0284 & 0.1920 & 0.1920 & 0.3000 \\ 
  5\% & 0.30 & -0.6128 & -0.2103 & 0.2285 & -0.5330 & 1.9090 & 1.7387 & 1.0364 & 0.1886 & 0.1886 & 0.3092 \\ 
  5\% & 0.50 & -0.5684 & -0.2531 & 0.2321 & -0.5377 & 2.0922 & 1.7842 & 1.0504 & 0.1819 & 0.1819 & 0.3119 \\ \midrule
  10\% & 0.00 & -5.2368 & 3.1727 & 2.6575 & -3.5577 & -9.0508 & 18.4406 & 56.2796 & 15.2271 & 15.2271 & 53.1060 \\ 
  10\% & 0.10 & -0.7586 & -0.1472 & 0.1556 & -0.4861 & 2.0517 & 1.6211 & 0.9669 & 0.1682 & 0.1682 & 0.3021 \\ 
  10\% & 0.30 & -0.7050 & -0.1958 & 0.1569 & -0.4880 & 2.2714 & 1.6653 & 1.0179 & 0.1774 & 0.1774 & 0.3222 \\ 
  10\% & 0.50 & -0.6802 & -0.2234 & 0.1570 & -0.4887 & 2.4186 & 1.7039 & 1.1086 & 0.1815 & 0.1815 & 0.3348 \\ \midrule
  20\% & 0.00 & -4.9382 & 2.7166 & 3.7332 & -5.0909 & -7.4022 & 28.8417 & 85.3125 & 7.4743 & 7.4743 & 115.4907 \\ 
  20\% & 0.10 & -5.2480 & 3.3204 & 4.5024 & -5.5042 & -9.2552 & 28.1360 & 58.1067 & 12.1681 & 12.1681 & 72.6244 \\ 
  20\% & 0.30 & -0.6565 & -0.1190 & 0.3440 & -0.6207 & 1.4426 & 1.8727 & 1.1155 & 0.2877 & 0.2877 & 0.4527 \\ 
  20\% & 0.50 & -0.6262 & -0.1409 & 0.3517 & -0.6276 & 1.5058 & 1.9099 & 1.2827 & 0.3382 & 0.3382 & 0.5271 \\ 
   \bottomrule
\end{tabular}}
\caption{MDPDE estimates of drift matrix $B$, intercept $b$, and diffusion $Sigma$ ($n=100$)} 
\label{tab:mdpde_all_n100}
\end{table}

\begin{table}
\centering
\resizebox{1\textwidth}{!}{
\begin{tabular}{crrrrrrrrrrr}
  \toprule
$\epsilon$ & alpha & B11 & B12 & B21 & B22 & b1 & b2 & S11 & S12 & S21 & S22 \\ 
  \midrule
0\% & 0.00 & -0.7051 & 0.0250 & 0.0793 & -0.3383 & 1.0675 & 0.7754 & 1.4103 & 0.4447 & 0.4447 & 0.4707 \\ 
  0\% & 0.10 & -0.7118 & 0.0439 & 0.0921 & -0.3424 & 0.8521 & 0.7337 & 1.4218 & 0.4484 & 0.4484 & 0.4673 \\ 
  0\% & 0.30 & -0.7058 & 0.0521 & 0.1278 & -0.3656 & 0.6197 & 0.7441 & 1.4513 & 0.4550 & 0.4550 & 0.4608 \\ 
  0\% & 0.50 & -0.7041 & 0.0601 & 0.1512 & -0.3786 & 0.4239 & 0.7233 & 1.4825 & 0.4594 & 0.4594 & 0.4557 \\ \midrule
  5\% & 0.00 & -4.1536 & 2.4890 & 3.2993 & -4.0289 & -4.5583 & 16.7692 & 42.4082 & 18.5921 & 18.5921 & 51.6905 \\ 
  5\% & 0.10 & -0.6745 & 0.0374 & 0.1567 & -0.4015 & 0.6975 & 0.8605 & 1.4334 & 0.4467 & 0.4467 & 0.4609 \\ 
  5\% & 0.30 & -0.6500 & 0.0284 & 0.1835 & -0.4199 & 0.5407 & 0.8898 & 1.5083 & 0.4612 & 0.4612 & 0.4662 \\ 
  5\% & 0.50 & -0.6276 & 0.0162 & 0.2027 & -0.4330 & 0.4507 & 0.9108 & 1.5870 & 0.4759 & 0.4759 & 0.4739 \\ \midrule
  10\% & 0.00 & -7.8755 & 5.4985 & 4.3082 & -5.2371 & -12.2016 & 21.7219 & 94.4123 & 20.3333 & 20.3333 & 72.4413 \\ 
  10\% & 0.10 & -0.6431 & 0.0104 & 0.0547 & -0.3218 & 0.6144 & 0.5970 & 1.5191 & 0.4619 & 0.4619 & 0.5062 \\ 
  10\% & 0.30 & -0.6190 & 0.0125 & 0.0848 & -0.3367 & 0.3307 & 0.5451 & 1.5904 & 0.4817 & 0.4817 & 0.5302 \\ 
  10\% & 0.50 & -0.6082 & 0.0148 & 0.1050 & -0.3466 & 0.1548 & 0.5079 & 1.7039 & 0.5094 & 0.5094 & 0.5569 \\ \midrule
  20\% & 0.00 & -8.6224 & 5.6377 & 5.4207 & -7.0913 & -8.6557 & 30.0396 & 125.0633 & 27.9376 & 27.9376 & 133.3203 \\ 
  20\% & 0.10 & -0.7554 & 0.1067 & 0.0949 & -0.3126 & 0.2123 & 0.2343 & 1.5787 & 0.5997 & 0.5997 & 0.6613 \\ 
  20\% & 0.30 & -0.7017 & 0.0780 & 0.1280 & -0.3351 & 0.1142 & 0.2543 & 1.7199 & 0.5822 & 0.5822 & 0.6057 \\ 
  20\% & 0.50 & -0.6718 & 0.0643 & 0.1462 & -0.3465 & 0.0085 & 0.2347 & 1.9330 & 0.6422 & 0.6422 & 0.6718 \\ 
   \bottomrule
\end{tabular}}
\caption{MDPDE estimates of drift matrix $B$, intercept $b$, and diffusion $Sigma$ ($n= 200$)} 
\label{tab:mdpde_all_n200}
\end{table}

\begin{table}
\centering
\resizebox{1\textwidth}{!}{
\begin{tabular}{crrrrrrrrrrr}
  \toprule
$\epsilon$ & alpha & B11 & B12 & B21 & B22 & b1 & b2 & S11 & S12 & S21 & S22 \\ 
  \midrule
0\% & 0.00 & -0.5123 & -0.2830 & 0.1825 & -0.4409 & 2.1266 & 1.0641 & 1.2485 & 0.3669 & 0.3669 & 0.4797 \\ 
  0\% & 0.10 & -0.4814 & -0.3051 & 0.1923 & -0.4517 & 2.1531 & 1.1085 & 1.2219 & 0.3578 & 0.3578 & 0.4722 \\ 
  0\% & 0.30 & -0.4334 & -0.3367 & 0.2099 & -0.4704 & 2.2025 & 1.1846 & 1.1961 & 0.3526 & 0.3526 & 0.4582 \\ 
  0\% & 0.50 & -0.3960 & -0.3589 & 0.2247 & -0.4856 & 2.2472 & 1.2439 & 1.1911 & 0.3559 & 0.3559 & 0.4457 \\ \midrule
  5\% & 0.00 & -6.1694 & 2.8960 & 3.7834 & -4.5292 & 0.3436 & 14.7355 & 58.8943 & 13.8092 & 13.8092 & 74.2389 \\ 
  5\% & 0.10 & -0.5959 & -0.2343 & 0.1768 & -0.4526 & 2.1174 & 1.1631 & 1.2555 & 0.3727 & 0.3727 & 0.4873 \\ 
  5\% & 0.30 & -0.5314 & -0.2789 & 0.1930 & -0.4668 & 2.2088 & 1.1896 & 1.2748 & 0.3793 & 0.3793 & 0.4723 \\ 
  5\% & 0.50 & -0.4880 & -0.3075 & 0.2050 & -0.4774 & 2.2986 & 1.2155 & 1.3049 & 0.3956 & 0.3956 & 0.4726 \\ \midrule
  10\% & 0.00 & -9.7854 & 5.2484 & 5.0744 & -5.7864 & -3.5724 & 19.0504 & 112.1633 & 1.6783 & 1.6783 & 97.4943 \\ 
  10\% & 0.10 & -0.5451 & -0.2507 & 0.1662 & -0.4301 & 2.1068 & 1.1428 & 1.2780 & 0.3871 & 0.3871 & 0.5185 \\ 
  10\% & 0.30 & -0.4974 & -0.2787 & 0.1830 & -0.4464 & 2.1009 & 1.2087 & 1.3083 & 0.3987 & 0.3987 & 0.5138 \\ 
  10\% & 0.50 & -0.4693 & -0.2919 & 0.1965 & -0.4585 & 2.0879 & 1.2504 & 1.3797 & 0.4272 & 0.4272 & 0.5311 \\ \midrule
  20\% & 0.00 & -16.0131 & 8.2888 & 6.6667 & -9.7005 & -3.2650 & 37.9938 & 243.3430 & 33.3199 & 33.3199 & 209.0705 \\ 
  20\% & 0.10 & -0.4522 & -0.3286 & 0.1760 & -0.4836 & 2.3428 & 1.4389 & 1.4151 & 0.3415 & 0.3415 & 0.5109 \\ 
  20\% & 0.30 & -0.4024 & -0.3702 & 0.1604 & -0.4796 & 2.4067 & 1.4839 & 1.5206 & 0.3754 & 0.3754 & 0.5410 \\ 
  20\% & 0.50 & -0.3957 & -0.3788 & 0.1595 & -0.4819 & 2.4347 & 1.5081 & 1.7105 & 0.4261 & 0.4261 & 0.5964 \\ 
   \bottomrule
\end{tabular}}
\caption{MDPDE estimates of drift matrix $B$, intercept $b$, and diffusion $Sigma$ ($n=500$)} 
\label{tab:mdpde_all_n500}
\end{table}

\begin{table}
\centering
\resizebox{1\textwidth}{!}{
\begin{tabular}{crrrrrrrrrrr}
  \toprule
$\epsilon$ & alpha & B11 & B12 & B21 & B22 & b1 & b2 & S11 & S12 & S21 & S22 \\ 
  \midrule
0\% & 0.00 & -0.4550 & -0.2438 & 0.2546 & -0.4835 & 1.8413 & 0.8322 & 1.2340 & 0.3494 & 0.3494 & 0.4221 \\ 
  0\% & 0.10 & -0.4539 & -0.2414 & 0.2372 & -0.4666 & 1.8464 & 0.7919 & 1.2157 & 0.3474 & 0.3474 & 0.4238 \\ 
  0\% & 0.30 & -0.4524 & -0.2375 & 0.2082 & -0.4395 & 1.8624 & 0.7365 & 1.1805 & 0.3443 & 0.3443 & 0.4261 \\ 
  0\% & 0.50 & -0.4558 & -0.2324 & 0.1851 & -0.4191 & 1.8799 & 0.7021 & 1.1539 & 0.3426 & 0.3426 & 0.4281 \\ \midrule
  5\% & 0.00 & -13.1782 & 7.5944 & 10.8178 & -9.9883 & 3.5234 & 14.2734 & 84.8650 & -7.8879 & -7.8879 & 107.9517 \\ 
  5\% & 0.10 & -0.4283 & -0.2197 & 0.2201 & -0.4430 & 1.6085 & 0.7380 & 1.2399 & 0.3555 & 0.3555 & 0.4317 \\ 
  5\% & 0.30 & -0.4372 & -0.2121 & 0.1894 & -0.4129 & 1.6402 & 0.6549 & 1.2332 & 0.3611 & 0.3611 & 0.4396 \\ 
  5\% & 0.50 & -0.4500 & -0.2026 & 0.1656 & -0.3904 & 1.6628 & 0.5946 & 1.2353 & 0.3685 & 0.3685 & 0.4519 \\ \midrule
  10\% & 0.00 & -17.8379 & 9.3182 & 10.3345 & -11.7449 & 11.2862 & 23.8553 & 169.1094 & 25.1268 & 25.1268 & 191.2546 \\ 
  10\% & 0.10 & -0.4585 & -0.2446 & 0.2039 & -0.4701 & 1.7840 & 0.8952 & 1.2317 & 0.3384 & 0.3384 & 0.4338 \\ 
  10\% & 0.30 & -0.4703 & -0.2339 & 0.1748 & -0.4409 & 1.7918 & 0.8009 & 1.2431 & 0.3473 & 0.3473 & 0.4483 \\ 
  10\% & 0.50 & -0.4798 & -0.2272 & 0.1545 & -0.4219 & 1.8034 & 0.7457 & 1.2830 & 0.3622 & 0.3622 & 0.4728 \\ \midrule
  20\% & 0.00 & -25.9607 & 11.8961 & 12.2063 & -17.6399 & 21.2474 & 48.8700 & 364.0903 & 41.5058 & 41.5058 & 404.2682 \\ 
  20\% & 0.10 & -0.5561 & -0.2070 & 0.1782 & -0.4609 & 1.7540 & 0.7925 & 1.2866 & 0.3760 & 0.3760 & 0.4604 \\ 
  20\% & 0.30 & -0.5738 & -0.1963 & 0.1419 & -0.4285 & 1.7691 & 0.7133 & 1.3972 & 0.4303 & 0.4303 & 0.5250 \\ 
  20\% & 0.50 & -0.5957 & -0.1812 & 0.1155 & -0.4050 & 1.7599 & 0.6634 & 1.5590 & 0.4889 & 0.4889 & 0.5972 \\ 
   \bottomrule
\end{tabular}}
\caption{MDPDE estimates of drift matrix $B$, intercept $b$, and diffusion $Sigma$ ($n= 1000$)} 
\label{tab:mdpde_all_n1000}
\end{table}

\begin{table}
\centering
\resizebox{1\textwidth}{!}{
\begin{tabular}{crrrrrrrrrrr}

  \toprule
$\epsilon$ & alpha & B11 & B12 & B21 & B22 & b1 & b2 & S11 & S12 & S21 & S22 \\ 
  \midrule
0\% & 0.00 & -0.6430 & -0.2488 & 0.0245 & -0.3805 & 2.1538 & 1.1627 & 1.3182 & 0.3671 & 0.3671 & 0.4569 \\ 
  0\% & 0.10 & -0.6522 & -0.2385 & 0.0227 & -0.3791 & 2.0995 & 1.1648 & 1.3328 & 0.3700 & 0.3700 & 0.4605 \\ 
  0\% & 0.30 & -0.6736 & -0.2227 & 0.0171 & -0.3773 & 2.0425 & 1.1756 & 1.3613 & 0.3751 & 0.3751 & 0.4663 \\ 
  0\% & 0.50 & -0.6877 & -0.2121 & 0.0117 & -0.3793 & 1.9936 & 1.2060 & 1.3850 & 0.3793 & 0.3793 & 0.4706 \\ \midrule
  5\% & 0.00 & -20.2670 & 10.5295 & 11.0288 & -12.2400 & -0.9734 & 29.8966 & 135.9442 & -1.0510 & -1.0510 & 135.9482 \\ 
  5\% & 0.10 & -0.7037 & -0.1959 & 0.0016 & -0.3589 & 2.0109 & 1.1133 & 1.3413 & 0.3691 & 0.3691 & 0.4606 \\ 
  5\% & 0.30 & -0.7122 & -0.1853 & -0.0030 & -0.3548 & 1.9266 & 1.1027 & 1.4029 & 0.3821 & 0.3821 & 0.4763 \\ 
  5\% & 0.50 & -0.7157 & -0.1811 & -0.0080 & -0.3530 & 1.8709 & 1.1022 & 1.4649 & 0.3964 & 0.3964 & 0.4916 \\ \midrule
  10\% & 0.00 & -29.3462 & 13.2961 & 11.8691 & -17.5687 & 6.8052 & 53.5519 & 292.9587 & 38.4244 & 38.4244 & 278.8074 \\ 
  10\% & 0.10 & -0.6842 & -0.2173 & 0.0230 & -0.3864 & 2.1081 & 1.2672 & 1.3756 & 0.3942 & 0.3942 & 0.4815 \\ 
  10\% & 0.30 & -0.7228 & -0.1819 & 0.0057 & -0.3733 & 2.0029 & 1.2482 & 1.4785 & 0.4223 & 0.4223 & 0.5087 \\ 
  10\% & 0.50 & -0.7508 & -0.1571 & -0.0076 & -0.3645 & 1.9287 & 1.2381 & 1.5833 & 0.4509 & 0.4509 & 0.5394 \\ \midrule
  20\% & 0.00 & -37.4225 & 14.9905 & 17.3959 & -29.5742 & 19.0351 & 97.7580 & 472.7275 & 57.8714 & 57.8714 & 556.7383 \\ 
  20\% & 0.10 & -0.6138 & -0.2176 & 0.0432 & -0.3753 & 1.5691 & 1.0632 & 1.4728 & 0.4414 & 0.4414 & 0.5248 \\ 
  20\% & 0.30 & -0.6216 & -0.2081 & 0.0577 & -0.3862 & 1.4789 & 1.0864 & 1.6472 & 0.4875 & 0.4875 & 0.5868 \\ 
  20\% & 0.50 & -0.6243 & -0.2039 & 0.0668 & -0.3958 & 1.4257 & 1.1141 & 1.8744 & 0.5542 & 0.5542 & 0.6676 \\ 
   \bottomrule
\end{tabular}}
\caption{MDPDE estimates of drift matrix $B$, intercept $b$, and diffusion $Sigma$ ($n= 2000$)} 
\label{tab:mdpde_all_n2000}
\end{table}
Simulations\footnote{All simulations are implemented in \textsf{R}, using the \texttt{parallel} and \texttt{doParallel} packages to accelerate computation. Each sample size is processed independently, and for fixed $n$, all contamination levels and tuning parameters are evaluated in parallel.} are conducted for sample sizes $n \in \{100, 200, 500, 1000, 2000\}$. For each $(n,\epsilon,\alpha)$ combination, parameter estimates of $(B,b,\Sigma)$ are obtained and summarized across repeated runs. Estimation accuracy for different n are shown in tables \ref{tab:mdpde_all_n100}, \ref{tab:mdpde_all_n200}, \ref{tab:mdpde_all_n500}, \ref{tab:mdpde_all_n1000} and \ref{tab:mdpde_all_n2000}; these results illustrate the progressive robustness gains of the MDPDE as $\alpha$ increases, particularly under moderate and severe contamination, while maintaining efficiency in the uncontaminated setting.

\section{Conclusion}
As discussed earlier, multivariate diffusion processes play an important role in diverse application settings. Hence, a unified inferential treatment for such processes is an important task.
This paper addressed robust parameter estimation and hypothesis testing procedure for multivariate diffusion
processes observed at high frequency by employing the minimum density power
divergence estimator (MDPDE).
Under a general diffusion framework and suitable regularity conditions, we
established consistency and asymptotic normality for both drift and diffusion
parameter estimators. Asymptotic test statistics are also developed for these parameters.
The drift estimator converges at the $\sqrt{n h_n}$ rate, while the diffusion
estimator achieves the standard $\sqrt{n}$ rate, and the two estimators are shown to be 
asymptotically independent.

From an applied perspective, the proposed methodology offers a stable and
reliable alternative to quasi-likelihood–based estimation in settings where
high-frequency data are affected by outliers, local contamination, or mild
model misspecification.
Simulation results demonstrate that the MDPDE maintains good finite-sample
accuracy while substantially improving robustness relative to likelihood-based
procedures.
This feature is particularly relevant in empirical applications, where
numerical optimization of quasi-likelihoods may lead to boundary solutions or
highly unstable estimates, even when the underlying continuous-time model is
correctly specified.
By downweighting extreme observations, the MDPDE-based quasi-likelihood approach
yields a more regular and reliable solution space for finite-sample inference.
These results underscore the practical relevance of divergence-based estimation for high-frequency diffusion models. It also points to natural extensions where we have more complex continuous-time settings. We aim to pursue this line of work in future.

\section*{Acknowledgements}
The author would like to express sincere gratitude to Prof.~Diganta Mukherjee\footnotemark~for his valuable guidance and supervision throughout the course of this research. 
The author is also grateful to Prof.~Abhik Ghosh\footnotemark[\value{footnote}]
for carefully reading the manuscript and providing insightful comments and suggestions 
that helped improve the quality of this work.

\footnotetext{Indian Statistical Institute, India}

% \paragraph{Future directions.}
% Several extensions of the present work are of practical and theoretical
% interest.
% First, the proposed framework can be extended to diffusion models with jumps or
% microstructure noise, which frequently arise in high-frequency financial and
% economic data.
% Second, data-driven selection of the divergence tuning parameter remains an
% important open problem, particularly in multivariate settings.
% Third, robust inference for functionals of the diffusion, such as integrated
% volatility or covariation, may be developed within the same divergence-based
% framework.
% Finally, applications to large-scale empirical studies, including financial
% markets and other high-frequency observational systems, provide a promising
% direction for assessing the practical impact of robust divergence-based
% estimation in continuous-time models.

% \section{Hi}

\printbibliography

\clearpage\newpage
\appendix

 \section{Appendix}\label{sec:app}

\setcounter{subsection}{0}
\renewcommand{\thesubsection}{A\arabic{subsection}}
\renewcommand{\theequation}{\thesection.\arabic{subsection}.\arabic{equation}}
\numberwithin{equation}{subsection}

\subsection{Proof of Theorems}
\subsubsection{Proof of Theorem \ref{thm:mdpde-consistency} (Consistency)}

 \begin{proof}
  We begin by recalling the expression of $V_{n,i}^\alpha({\bm\theta})$ from~\eqref{eq:mdpd_multivariate}. The proof is divided into the following steps.

\begin{enumerate}

\item[\textbf{Step 1.}]
We first establish the uniform convergence of the empirical objective function.
Specifically, we show that
\[
\sup_{\bm\theta \in \Theta}
\left|
\frac{1}{n}
\sum_{i=1}^n
V_{n,i}^\alpha(\bm\theta)
-
\Psi(\Sigma)
\right|
\xrightarrow{\mathbb{P}} 0,
\]
where the deterministic limit $\Psi(\Sigma)$ is given by
\begin{align}
\Psi(\Sigma)
=
\frac{1}{(1+\alpha)^{d/2}}
|\Sigma|^{-\alpha/2}
-
\left(1+\frac{1}{\alpha}\right)
|\Sigma|^{-\alpha/2}
\Bigl|
I_d+\alpha\,\Lambda_0 \Sigma^{-1} \Lambda_0
\Bigr|^{-1/2}.\label{eqn:Psi_Sigma}
\end{align}
This convergence is obtained by applying a uniform law of large numbers together
with suitable moment bounds. Lemmas
\ref{lem:poly-negligible-mv} and
\ref{lem:gaussian_tilt_multivariate}
are used repeatedly in this step.

\item[\textbf{Step 2.}]
We next analyze the contribution of the drift parameter $\bm\beta$ at the
$h_n^{-1}$ scale. We show that
\begin{align*}
&\frac{1}{n h_n}
\sum_{i=1}^n
V_{n,i}^\alpha(\bm\beta,\Sigma)
-
\frac{1}{n h_n}
\sum_{i=1}^n
V_{n,i}^\alpha(\bm\beta_0,\Sigma)
\\
&\hspace{2cm}
\xrightarrow{\mathbb{P}}
\frac{1+\alpha}{2}
|\Sigma|^{-\alpha/2}
\Bigl|
I_d + \alpha \Lambda_0^{\top} \Sigma^{-1} \Lambda_0
\Bigr|^{-1/2}
\\
&\hspace{3cm}
\times
\int
\Bigl[
D(\bm\beta)^{\top} \Sigma^{-1} D(\bm\beta)
-
\alpha D(\bm\beta)^{\top} \Sigma^{-1} \Lambda_0
\Bigl(
I_d + \alpha \Lambda_0^{\top} \Sigma^{-1} \Lambda_0
\Bigr)^{-1}
\Lambda_0^{\top} \Sigma^{-1} D(\bm\beta)
\Bigr]
\,\pi(dx),
\end{align*}
where $D(\bm\beta)=a(x,\bm\beta)-a(x,\bm\beta_0)$ and $\pi$ denotes the invariant
probability measure of the diffusion.

The above convergence follows from a second-order expansion of the drift term and
uniform control of the remainder. Lemmas
\ref{lem:delta-moment-mv},
\ref{lem:poly-negligible-mv} and \ref{lem:linear_form_uniform}
are used throughout this step.

\item[\textbf{Step 3.}]
Following the argument of \cite{Kessler1997}, we first establish the
consistency of the diffusion parameter.
From Step~1 we have
\[
\sup_{\bm\theta\in\Theta}
\left|
\frac{1}{n}\sum_{i=1}^n V_{n,i}^{\alpha}(\bm\theta)
-
\Psi(\Sigma)
\right|
\xrightarrow{\mathbb P}0 .
\]
The function $\Psi(\Sigma)$ is continuous in $\Sigma$, since it is
composed of determinants, matrix inverses, and matrix products,
which are continuous on the set of positive definite matrices.
By the identifiability assumption, $\Psi(\Sigma)$ is uniquely minimized
at $\Sigma_0$.

Since $\Theta$ is compact, any sequence of estimators
$(\hat{\bm\beta}_n,\hat{\Sigma}_n)$ admits a convergent subsequence.
Let $(\hat{\bm\beta}_{n_k},\hat{\Sigma}_{n_k}) \to
(\bm\beta_\infty,\Sigma_\infty)$.
By the uniform convergence above,
\[
\frac{1}{n_k}\sum_{i=1}^{n_k}
V_{n_k,i}^{\alpha}(\hat{\bm\beta}_{n_k},\hat{\Sigma}_{n_k})
\longrightarrow
\Psi(\Sigma_\infty).
\]
By the minimizing property of the estimator,
\[
\sum_{i=1}^n V_{n,i}^{\alpha}(\hat{\bm\beta}_n,\hat{\Sigma}_n)
\le
\sum_{i=1}^n V_{n,i}^{\alpha}(\bm\beta,\Sigma_0)
\]
for any $\bm\beta$. Dividing by $n$ and passing to the limit yields
\(
\Psi(\Sigma_\infty)\le \Psi(\Sigma_0).
\)
Since $\Psi$ has the unique minimizer $\Sigma_0$, we conclude that
$\Sigma_\infty=\Sigma_0$. Hence
\(
\hat{\Sigma}_n \xrightarrow{\mathbb P} \Sigma_0 .
\)

\item[\textbf{Step 4.}]
We now establish the consistency of the drift parameter.
From Step~3 we already have $\hat{\Sigma}_n \xrightarrow{\mathbb P} \Sigma_0$.
Consider a subsequence $(\hat{\bm\beta}_{n_k},\hat{\Sigma}_{n_k})$
such that
\(
(\hat{\bm\beta}_{n_k},\hat{\Sigma}_{n_k})
\rightarrow
(\bm\beta_\infty,\Sigma_0).
\) Using the convergence result in Step~2,
\begin{align*}
&\frac{1}{n_k h_{n_k}}
\sum_{i=1}^{n_k}
V_{n_k,i}^{\alpha}(\hat{\bm\beta}_{n_k},\hat{\Sigma}_{n_k})
-
\frac{1}{n_k h_{n_k}}
\sum_{i=1}^{n_k}
V_{n_k,i}^{\alpha}(\bm\beta_0,\hat{\Sigma}_{n_k})
\\
&\hspace{1cm}
\xrightarrow{\mathbb P}
\frac{1}{2}(1+\alpha)^{-d/2}|\Sigma_0|^{-\alpha/2}
\int
D(\bm\beta_\infty)^{\top}\Sigma_0^{-1}D(\bm\beta_\infty)
\,\pi(dx).
\end{align*}

By the minimizing property of the estimator, the left-hand side is
nonpositive, and hence the limit must also be nonpositive.
Since the integrand is nonnegative and vanishes only when
$D(\bm\beta_\infty)=0$, the identifiability condition implies
$\bm\beta_\infty=\bm\beta_0$.
Therefore,
\(
\hat{\bm\beta}_n \xrightarrow{\mathbb P} \bm\beta_0 .
\) Combining the two results yields
\(
(\hat{\bm\beta}_n,\hat{\Sigma}_n)
\xrightarrow{\mathbb P}
(\bm\beta_0,\Sigma_0),
\)
which establishes the consistency of the minimum density power divergence estimator.
\end{enumerate}

Hence, it suffices to verify Steps~1 and~2.
\paragraph{Step 1:}
 We recall that
\[
R_i(\bm\beta)
=
X_{t_i^n}-X_{t_{i-1}^n}-a(X_{t_{i-1}^n},\bm\beta)h_n .
\]
Using the Euler expansion at the true parameter value,
\[
X_{t_i^n}-X_{t_{i-1}^n}
=
a(X_{t_{i-1}^n},{\bm\beta_0})h_n
+
\Lambda_0\sqrt{h_n}\,Z_{n,i}
+
\Delta_{n,i},
\]
where
\[
\Delta_{n,i}
=
\int_{t_{i-1}^n}^{t_i^n}
\Big\{
a(X_s,{\bm\beta_0})-a(X_{t_{i-1}^n},{\bm\beta_0})
\Big\}\,ds .
\]

Hence, owing to equation~\eqref{eqn:R_i_D_i}, we get,
\begin{align}
Q_i(\bm\beta) := \frac{1}{h_n}
R_i(\bm\beta)^\top \Sigma^{-1} R_i(\bm\beta)
&=
Z_{n,i}^\top
\Lambda_0 \Sigma^{-1} \Lambda_0
Z_{n,i}
+2\sqrt{h_n}\,
Z_{n,i}^\top
\Lambda_0 \Sigma^{-1}
D_i(\bm\beta)
+\frac{2}{\sqrt{h_n}}\,
Z_{n,i}^\top
\Lambda_0 \Sigma^{-1}
\Delta_{n,i}
\nonumber\\
&\quad
+h_n\,
D_i(\bm\beta)^\top
\Sigma^{-1}
D_i(\bm\beta)
+2\,
D_i(\bm\beta)^\top
\Sigma^{-1}
\Delta_{n,i}
+\frac{1}{h_n}\,
\Delta_{n,i}^\top
\Sigma^{-1}
\Delta_{n,i}.
\label{eq:quad_expansion}
\end{align}
Define, \begin{align*}
    K_i({\bm\theta}):=K_i(\bm\beta,\Sigma) &= \frac{\alpha}{2}
Q_i(\bm\beta)-\frac{\alpha}{2}Z_{n,i}^\top
\Lambda_0 \Sigma^{-1} \Lambda_0Z_{n,i}\\
&= \alpha\sqrt{h_n}\,
Z_{n,i}^\top
\Lambda_0 \Sigma^{-1}
D_i(\bm\beta) +\frac{\alpha}{\sqrt{h_n}}\,
Z_{n,i}^\top
\Lambda_0 \Sigma^{-1}
\Delta_{n,i}
+\frac{\alpha}{2}h_n\,
D_i(\bm\beta)^\top
\Sigma^{-1}
D_i(\bm\beta)
\\&\qquad\qquad+\alpha\,
D_i(\bm\beta)^\top
\Sigma^{-1}
\Delta_{n,i}
+\frac{\alpha}{2h_n}\,
\Delta_{n,i}^\top
\Sigma^{-1}
\Delta_{n,i}.
\end{align*} Now note that from Lemma \ref{lem:poly-negligible-mv} each of the part of $K_i({\bm\theta})$ has probability convergence to 0 and hence $$\sup_{{\bm\theta}}\max_{1\le i \le n}{|K_i({\bm\theta})|} = o_P(h_n^{r})\;,0<r<\frac{1}{2}$$ then \begin{align*}
    &\Bigg|\frac{1}{n}\sum_{i=1}^nV_{n,i}^\alpha({\bm\theta})- \frac{1}{(1+\alpha)^{d/2}}
    |\Sigma|^{-\alpha/2}
    +\frac{1}{n}\sum_{i=1}^n
    \left(1+\frac{1}{\alpha}\right)
    |\Sigma|^{-\alpha/2}
    \exp\!\left(
    -\frac{\alpha}{2}
    Z_{n,i}^\top\Lambda_0
    \Sigma^{-1}\Lambda_0
    Z_{n,i}
    \right)\Bigg|\\
    &\le \fst{1+\frac{1}{\alpha}} |\Sigma|^{-\alpha/2}\fst{\exp\fst{\sup_{{\bm\theta}}\max_{1\le i\le n}|K_i({\bm\theta})|}-1}\\
    & \le C_{\alpha,1}\fst{\exp\fst{\sup_{{\bm\theta}}\max_{1\le i\le n}|K_i({\bm\theta})|}-1}\xrightarrow{P}0.
\end{align*}

Hence using compactness of $\Theta$ and Lemma \ref{lem:gaussian_tilt_multivariate} Step 1 follows.

\paragraph{Step 2:} 
 Using taylor expansion in both the series we get,\begin{align*}
     \begin{aligned}
\exp\fst{-K_i\left({\bm\beta_0}
,\Sigma\right)}-\exp\fst{-K_i(\bm\beta,\Sigma)}
=&\alpha \sqrt{h_n} Z_{n, i}^{\top} \Lambda_0^{\top} \Sigma^{-1} D_i(\bm\beta)+ h_n\frac{\alpha}{2}{D_i(\bm\beta)^\top\Sigma^{-1}D_i(\bm\beta)} + \alpha D_i(\bm\beta)^{\top} \Sigma^{-1} \Delta_{n, i}\\
& -\frac{1}{2}\fst{\alpha^2 h_n\left(Z_{n, i}^{\top} \Lambda_0^{\top} \Sigma^{-1} D_i(\bm\beta)\right)^2+K_i(\bm\beta,\Sigma)^2-\alpha^2 h_n\left(Z_{n, i}^{\top} \Lambda_0^{\top} \Sigma^{-1} D_i(\bm\beta)\right)^2}\\
& +\frac{K_i\left({\bm\beta}_0, \Sigma\right)^2}{2!} \exp\fst{\zeta_{1, i}}+\frac{K_i({\bm\beta}, \Sigma)^3}{3!} \exp\fst{\zeta_2, i}
\end{aligned}
 \end{align*}
where, $\left|\zeta_{1, i}\right| \leq\left|K_i\left({\bm\beta_0}, \Sigma\right)\right|,\left|\zeta_{2, i}\right| \leq\left|K_i(\bm\beta, \Sigma)\right|$ and, let,\begin{align*}
     H_{n,i}(\bm\beta,\Sigma) = \alpha D_i(\bm\beta)^{\top} \Sigma^{-1} \Delta_{n, i} - \frac{1}{2}\secnd{K_i(\bm\beta,\Sigma)^2-\alpha^2 h_n\left(Z_{n, i}^{\top} \Lambda_0^{\top} \Sigma^{-1} D_i(\bm\beta)\right)^2}
 \end{align*}
The first term in $H_{n,i}(\bm\beta,\Sigma)$ is of order $o_{\mathbb{P}}(h_n^{r_1})$ for any $0 \le r_1 < 1.5$, by Lemma~\ref{lem:delta-moment-mv}. 
For the second term, note that upon squaring $K_i(\bm\beta,\Sigma)$, the leading $h_n$ terms cancel out, leaving a remainder of order $h_n^{3/2}$. 
Consequently, this term is also of order $o_{\mathbb{P}}(h_n^{r_2})$ for any $0 \le r_2 < 1.5$.

Furthermore, by the polynomial growth condition and the fact that
$K_i({\bm\beta_0},\Sigma) = o_{\mathbb{P}}(h_n^{r})$ for some $0 \le r < 1$, 
an application of Lemma~\ref{lem:poly-negligible-mv} yields

\begin{align*}
    \begin{aligned}
\sup _\eta \max _{1 \leq i \leq n}\left|H_{n,i}({\bm\beta}, \Sigma)\right| & =o_P\left(h_n^{r_1}\right), \quad 0\le r_1<1.5,\\
\sup _\eta \max _{1 \leq i \leq n}\left|K_i({\bm\beta}, \Sigma)^3\right| & =o_P\left(h_n^{r_2}\right),\quad 0\le r_2<1.5, \\
\sup _\eta \max _{1 \leq i \leq n}\left|K_i\left({\bm\beta}_0, \Sigma\right)^2\right| & =o_P\left(h_n^{r_3}\right), \quad 0\le r_3<2 
\end{aligned}.
\end{align*}

Consequently using Lemma \ref{lem:gaussian_tilt_multivariate} and \ref{lem:linear_form_uniform}, 
\begin{align*}
&\frac{1}{n h_n} \sum_{i=1}^n V_{n, i}^\alpha({\bm\beta}, \Sigma)-\frac{1}{n h_n} \sum_{i=1}^n V_{n, i}^\alpha\left({\bm\beta}_0, \Sigma\right)\\=
&
\frac{1+\alpha}{2}|\Sigma|^{-\frac{\alpha}{2}} \frac{1}{n}  \sum_{i=1}^n D_i(\bm\beta)^{\top}\left[\Sigma^{-1}-\alpha \Sigma^{-1} \Sigma_0 Z_{n, i} Z_{n, i}^{\top} \Sigma_0 \Sigma^{-1}\right] D_i(\bm\beta) 
% \\
% & 
\times \exp \left(-\frac{\alpha}{2} Z_{n, i}^{\top} \Sigma_0^{\top} \Sigma^{-1} \Sigma_0 Z_{n, i}\right) + o(1)\\
\xrightarrow{P} & \frac{1+\alpha}{2}|\Sigma|^{-\frac{\alpha}{2}}\left|I+\alpha \Sigma_0^{\top} \Sigma^{-1} \Sigma_0\right|^{-1 / 2} \int\fst{D(\bm\beta)^{\top}\Sigma^{-1}D(\bm\beta)- \alpha D(\bm\beta)^{\top} \Sigma^{-1} \Sigma_0\left(I+\alpha \Sigma_0^{\top} \Sigma^{-1} \Sigma_0\right)^{-1} \Sigma_0^{\top} \Sigma^{-1} D(\bm\beta) }\pi(dx).
\end{align*}
Thus Step~2 is proved. 

Combining the results of Steps~1--4 completes the proof.
\end{proof}

\subsubsection{Proof of Normality}
In this part we establish the asymptotic normality of the proposed estimators.
The argument is notationally involved; therefore, we decompose the proof into
manageable steps. First of all define $$a_{i-1}(\bm \beta_0):= a(X_{t_{i-1}^n},{\bm\beta_0}) $$ and recall that the estimator $\hat{{\bm\theta}}_n^\alpha$ satisfies the estimating equation
\[
\frac{\partial l_n^\alpha(\hat{{\bm\theta}}_n^\alpha)}{\partial{{\bm\theta}}} =0,
\]
where
\[
l_n^\alpha({\bm\theta})=\sum_{i=1}^n V_{n,i}^\alpha({\bm\theta}),
\]
and $V_{n,i}^\alpha({\bm\theta})$ is defined in~\eqref{eq:mdpd_multivariate}.
When both the drift and diffusion parameters are unknown, the parameter space is
$\Theta$. If the diffusion matrix $\Sigma$ is known, the parameter space reduces to
$\Theta_{\bm\beta}$.
Throughout, when considering parameters of $\Sigma$, we work with the vector
$\vech(\Sigma)$, which contains the $d(d+1)/2$ distinct elements of the symmetric
matrix $\Sigma$.

To derive the asymptotic distribution of $\hat{\bm\theta}_n^\alpha$, we consider a
Taylor expansion of the estimating equation around the true parameter
${\bm\theta}_0^\top=({\bm\beta_0}^\top,\vech(\Sigma_0)^\top)$.
This yields
\begin{align}
0
=
-L_n^\alpha({\bm\theta}_0)
+
\int_0^1
C_n^\alpha\!\left(
{\bm\theta}_0+u(\hat{\bm\theta}_n^\alpha-{\bm\theta}_0)
\right)\,du
\binom{
\sqrt{nh_n}(\hat{\bm\beta}_n^\alpha-{\bm\beta_0})
}{
\sqrt{n}\left(\vech(\hat\Sigma_n^\alpha)-\vech(\Sigma_0)\right)
}.
\label{eq:taylor_expansion}
\end{align}

Here,
\begin{align*}
L_n^\alpha({\bm\theta}_1)
=
\begin{pmatrix}
-\dfrac{1}{\sqrt{n h_n}}
\dfrac{\partial\, l_n^\alpha({\bm\theta}_1)}
{\partial {\bm\beta}}\\[8pt]
-\dfrac{1}{\sqrt{n}}
\dfrac{\partial\, l_n^\alpha({\bm\theta}_1)}
{\partial \vech(\Sigma)}
\end{pmatrix},\qquad
% \\[8pt]
C_n^\alpha({\bm\theta}_1)
=
\begin{pmatrix}
\dfrac{1}{n h_n}
\dfrac{\partial^2\, l_n^\alpha({\bm\theta}_1)}
{\partial {\bm\beta}\,\partial {\bm\beta}^{\top}}
&
\dfrac{1}{n\sqrt{h_n}}
\dfrac{\partial^2\, l_n^\alpha({\bm\theta}_1)}
{\partial {\bm\beta}\,\partial \vech(\Sigma)^{\top}}
\\[10pt]
\dfrac{1}{n\sqrt{h_n}}
\dfrac{\partial^2\, l_n^\alpha({\bm\theta}_1)}
{\partial \vech(\Sigma)\,\partial {\bm\beta}^{\top}}
&
\dfrac{1}{n}
\dfrac{\partial^2\, l_n^\alpha({\bm\theta}_1)}
{\partial \vech(\Sigma)\,\partial \vech(\Sigma)^{\top}}
\end{pmatrix}.
\end{align*}

% \begin{align*}
% L_n^\alpha({\bm\theta}_1)
% &=
% \binom{
% -\frac{1}{\sqrt{nh_n}}\frac{\partial l_n^\alpha({\bm\theta}_1)}{\partial{\bm\beta}} 
% }{
% -\frac{1}{\sqrt{n}}\frac{\partial l_n^\alpha({\bm\theta}_1)}{\partial{\vech(\Sigma)} }
% },
% \\[6pt]
% C_n^\alpha({\bm\theta}_1)
% &=
% \begin{pmatrix}
% \frac{1}{nh_n}\partial^2_{{\bm\beta}{\bm\beta}^\top}l_n^\alpha({\bm\theta}_1)
% &
% \frac{1}{n\sqrt{h_n}}\partial^2_{{\bm\beta}\,\vech(\Sigma)^\top}l_n^\alpha({\bm\theta}_1)
% \\[6pt]
% \frac{1}{n\sqrt{h_n}}\partial^2_{\vech(\Sigma)\,{\bm\beta}^\top}l_n^\alpha({\bm\theta}_1)
% &
% \frac{1}{n}\partial^2_{\vech(\Sigma)\,\vech(\Sigma)^\top}l_n^\alpha({\bm\theta}_1)
% \end{pmatrix}.
% \end{align*}

As in the standard theory of $M$-estimators, the proof reduces to the analysis of
first- and second-order derivatives of the objective contributions
$V_{n,i}^\alpha$. In particular, we require control of
\[
\frac{\partial}{\partial\beta_j}V_{n,i}^\alpha,
\quad
\frac{\partial^2}{\partial\beta_j\partial\beta_k}V_{n,i}^\alpha,
\quad
\frac{\partial}{\partial\sigma_{rs}}V_{n,i}^\alpha,
\quad
\frac{\partial^2}{\partial\sigma_{rs}\partial\sigma_{kl}}V_{n,i}^\alpha,
\quad
\frac{\partial^2}{\partial\beta_{j}\partial\sigma_{kl}}V_{n,i}^\alpha.
\]
Explicit expressions for these derivatives are provided in
Supplementary~\ref{App:Deriv_Expect}; see
\eqref{eqn:1st_deriv_beta}, \eqref{eqn:1st_deriv_sigma},
\eqref{eqn:2nd_deriv_beta}, \eqref{eqn:2nd_deriv_sigma},
and \eqref{eqn:1st_deriv_cross}.

For technical convenience, we introduce the auxiliary quantity
\begin{align*}
U_{n,i}^\alpha
=&\,
\frac{1}{(1+\alpha)^{d/2}}|\Sigma|^{-\alpha/2}
\\
&\quad
-
\left(1+\frac{1}{\alpha}\right)
|\Sigma|^{-\alpha/2}
\exp\!\left(
-\frac{\alpha}{2h_n}
\left(
D_i(\bm\beta)+\sqrt{h_n}\Lambda_0Z_{n,i}
\right)^\top
\Sigma^{-1}
\left(
D_i(\bm\beta)+\sqrt{h_n}\Lambda_0Z_{n,i}
\right)
\right).
\end{align*}

The essential distinction between $V_{n,i}^\alpha$ and $U_{n,i}^\alpha$ is that
$V_{n,i}^\alpha$ contains the discretization error term $\Delta_{n,i}$, whereas
$U_{n,i}^\alpha$ does not.
Consequently, $U_{n,i}^\alpha$ behaves as if it were generated from a discrete-time
Gaussian model without discretization bias.
A key step in establishing asymptotic normality is to show that the first
derivative of $U_{n,i}^\alpha$ evaluated at ${\bm\theta}_0$ admits the same limiting
distribution as the corresponding derivative of $V_{n,i}^\alpha$ evaluated at
${\bm\theta}_0$.
The proof proceeds through the following steps.

% \begin{itemize}

% \item[\textbf{Step 1.}]
% Using Lemma~\ref{lem:u_v_deriv_conver_p}, we show that the first derivative of
% $U_{n,i}^\alpha$ evaluated at ${\bm\theta}_0$ converges in distribution.
% As a consequence, the score-type term $L_n^\alpha({\bm\theta}_0)$ admits a
% distributional limit for which Lemma~\ref{lem:poly-negligible-mv},~\ref{lem:gaussian_tilt_multivariate},~\ref{lem:quad_form_mean} and~\ref{lem:u_v_deriv_conver_p}.

% \item[\textbf{Step 2.}]
% We next prove that
% \begin{align}
% \int_0^1
% \left\{
% C_n^\alpha\!\left(
% {\bm\theta}_0+u(\hat{\bm\theta}_n^\alpha-{\bm\theta}_0)
% \right)
% -
% C_n^\alpha({\bm\theta}_0)
% \right\}
% \,du\le \sup_{|{\bm\theta}_1|\le \epsilon_n}|C_n^\alpha({\bm\theta}_0 +{\bm\theta}_1)-C_n^\alpha({\bm\theta}_0)|
% \;\xrightarrow{\mathbb{P}}\;0,\label{eqn:C_n^alpha}
% \end{align}
%  where $\{\varepsilon_n\}$  is any positive real sequence decaying to  0.
% This allows us to establish the convergence in probability of
% $C_n^\alpha({\bm\theta}_0)$ by using Lemma~\ref{lem:poly-negligible-mv} and \ref{lem:linear_form_uniform}.

% \item[\textbf{Step 3.}]
% Combining the results of Steps~1 and~2 with the expansion
% \eqref{eq:taylor_expansion}, we obtain the joint distributional convergence
% \[
% \binom{
% \sqrt{nh_n}(\hat{\bm\beta}_n^\alpha-{\bm\beta_0})
% }{
% \sqrt{n}\left(\vech(\hat\Sigma_n^\alpha)-\vech(\Sigma_0)\right)
% }
% \;\xrightarrow{\mathcal{L}}\;
% \mathcal{N}(0,\;\cdot).
% \]

% \end{itemize}

\begin{enumerate}

\item[\textbf{Step 1.}]
We first establish the asymptotic behavior of the score-type term evaluated at the
true parameter value. Using
Lemmas~\ref{lem:poly-negligible-mv},
\ref{lem:gaussian_tilt_multivariate},
\ref{lem:quad_form_mean}, and
\ref{lem:u_v_deriv_conver_p},
we show that the first derivative of $U_{n,i}^\alpha$ at ${\bm\theta}_0$ converges
in distribution. Consequently, the normalized score
$L_n^\alpha({\bm\theta}_0)$ admits a well-defined distributional limit.

\item[\textbf{Step 2.}]
We next control the remainder term arising from the Taylor expansion of the
estimating equation. Specifically, for any positive sequence
$\{\varepsilon_n\}$ such that $\varepsilon_n \to 0$, we show that
\begin{align}
\int_0^1
\Bigl\{
C_n^\alpha\!\left(
{\bm\theta}_0
+
u(\hat{\bm\theta}_n^\alpha-{\bm\theta}_0)
\right)
-
C_n^\alpha({\bm\theta}_0)
\Bigr\}
\,du
&\le
\sup_{\lvert {\bm\theta}_1 \rvert \le \varepsilon_n}
\bigl|
C_n^\alpha({\bm\theta}_0+{\bm\theta}_1)
-
C_n^\alpha({\bm\theta}_0)
\bigr|\xrightarrow{\mathbb{P}} 0.
\label{eqn:C_n^alpha}
\end{align}
This implies that $C_n^\alpha({\bm\theta}_0)$ converges in probability, again
using Lemmas~\ref{lem:poly-negligible-mv} and
\ref{lem:linear_form_uniform}.

\item[\textbf{Step 3.}]
Combining the results of Steps~1 and~2 with the Taylor expansion
\eqref{eq:taylor_expansion}, we obtain the joint asymptotic distribution
\[
\binom{
\sqrt{n h_n}\bigl(\hat{\bm\beta}_n^\alpha-{\bm\beta}_0\bigr)
}{
\sqrt{n}
\Bigl(
\vech(\hat{\Sigma}_n^\alpha)
-
\vech(\Sigma_0)
\Bigr)
}
\;\xrightarrow{\mathcal{L}}\;
\mathcal{N}\!\left(0,\;\cdot\right),
\]
where the limiting covariance matrix is given by the inverse of the corresponding
information-type matrix.

\end{enumerate}

The first step involves convergence in distribution and therefore requires a
vector-valued formulation. Once this is established, the remaining steps rely
on convergence in probability. These can be verified componentwise and hence
imply matrix convergence with respect to the Frobenius norm. Accordingly, we
first establish Steps~1 and~2 below; Step~3 then follows directly.

% \subsubsection{Proof of Theorem~\ref{thm:beta-clt}}

\begin{proof}[{Proof of Theorem~\ref{thm:beta-clt}}] 
~\\
\noindent\textbf{Step 1.}
From equation~\eqref{eqn:1st_deriv_true_u_beta}, we may write
\begin{align}
    \frac{1}{\sqrt{n h_n}}
\sum_{i=1}^n
\frac{\partial U_{n,i}^\alpha}{\partial \bm\beta}
=
\frac{1}{\sqrt{n h_n}}
\sum_{i=1}^n
\left.
\Bigl(
\frac{\partial U_{n,i}^\alpha}{\partial \beta_1},
\ldots,
\frac{\partial U_{n,i}^\alpha}{\partial \beta_p}
\Bigr)
\right|_{({\bm\beta_0},\Sigma_0)}^{\!\top}.\label{eqn:thm_2_deriv_beta}
\end{align}

By construction, each summand has conditional mean zero given
$\mathcal F_{t_{i-1}^n}$, and therefore the entire sum has expectation zero.

The covariance matrix of this vector is given by
\[
\frac{1}{n h_n}
\sum_{i=1}^n
\mathbb E_0
\left[
\frac{\partial U_{n,i}^\alpha}{\partial \bm\beta}
\frac{\partial U_{n,i}^\alpha}{\partial \bm\beta^\top}
\right].
\]
Using Lemmas~\ref{lem:gaussian_tilt_multivariate}
and~\ref{lem:quad_form_mean}, this matrix converges in probability to
\[
\frac{(1+\alpha)^2}{(1+2\alpha)^{d/2+1}}
|\Sigma_0|^{-\alpha}\,
\mathcal B,
\]
where $\mathcal B$ is defined in Assumption~\ref{assm:all_assumption}~(A7).

Indeed, for $u,v \in \{1,\ldots,p\}$, conditioning on
$\mathcal F_{t_{i-1}^n}$ yields
\begin{align}
\mathbb E_0\!\left[
\left.
\frac{\partial U_{n,i}^\alpha}{\partial \beta_u}
\frac{\partial U_{n,i}^\alpha}{\partial \beta_v}
\,\right|\,
\mathcal F_{t_{i-1}^n}
\right]
&=
(1+\alpha)^2 h_n
|\Sigma_0|^{-\alpha}
\left({\frac{\partial a_{i-1}({\bm\beta_0})}{\partial\beta_u}} \right)^\top
\Lambda_0^{-\top}\nonumber
\\
&\quad\times
\mathbb E_0\!\left[
Z_{n,i} Z_{n,i}^\top
\exp\!\left(-\alpha Z_{n,i}^\top Z_{n,i}\right)
\right]
\Lambda_0^{-1}
\left({\frac{a_{i-1}({\bm\beta_0})}{\partial\beta_v}} \right),\label{eqn:thm_2_expectation}
\end{align}
where $Z_{n,i}$ denotes the standardized Gaussian increment. Application of
Lemma~\ref{lem:quad_form_mean} yields the stated limit.

Moreover, the Lindeberg condition holds since
\[
\frac{1}{n^2 h_n^2}
\sum_{i=1}^n
\mathbb E_0
\!\left[
\left.
\biggl\|
\frac{\partial U_{n,i}^\alpha}{\partial \bm\beta}
\biggr\|^4
\,\right|\,
\mathcal F_{t_{i-1}^n}
\right]
\;\longrightarrow\; 0.
\]
Consequently, the asymptotic normality of $\frac{1}{\sqrt{nh_n}}\sum_{i=1}^n \frac{\partial U_{n,i}^\alpha}{\partial \bm\beta}$ is established.

\medskip
\noindent\textbf{Step 2.}
We next analyze the second-order derivatives of $V_{n,i}^\alpha$ and establish
their convergence in probability. Since convergence in probability can be verified
componentwise, it suffices to study elementwise limits, which in turn imply matrix
convergence in the Frobenius norm.

From~\eqref{eqn:2nd_deriv_beta}, we can write
\begin{align*}
\frac{\partial^2 V_{n,i}^\alpha}{\partial \beta_u \partial \beta_v}
=
\left(1+\frac{1}{\alpha}\right)
|\Sigma|^{-\alpha/2}
\exp\!\left(-K_i(\bm{\theta})\right)
\exp\!\left(
-\frac{\alpha}{2}
Z_{n,i}^\top \Lambda_0 \Sigma^{-1} \Lambda_0 Z_{n,i}
\right)
\left(J_{0,i}+J_{1,i}+J_{2,i}\right),
\end{align*}
where $J_{0,i}$, $J_{1,i}$, and $J_{2,i}$ are defined in
\eqref{eqn:J_0}, \eqref{eqn:J_1}, and \eqref{eqn:J_2}, respectively.
Applying a Taylor expansion of $\exp\{-K_i(\bm\theta)\}$ around zero and
retaining terms up to order $\sqrt{h_n}$ yields
\begin{align}
\frac{\partial^2 V_{n,i}^\alpha}{\partial \beta_u \partial \beta_v}
&=
\left(1+\frac{1}{\alpha}\right)
|\Sigma|^{-\alpha/2}
\exp\!\left(
-\frac{\alpha}{2}
Z_{n,i}^{\top}\Lambda_0 \Sigma^{-1}\Lambda_0 Z_{n,i}
\right)
\nonumber\\
&\quad\times
\Bigl[
1
- K_i(\bm\theta)
+\frac{1}{2}K_i(\bm\theta)^2 e^{\Upsilon_i}
\Bigr]
\left(J_{0,i}+J_{1,i}+J_{2,i}\right),
\label{eq:taylor_second_beta}
\end{align}
where $|\Upsilon_i|\le |K_i(\bm\theta)|$.

Consequently, the leading contribution can be written as
\begin{align}
\frac{\partial^2 V_{n,i}^\alpha}{\partial \beta_u \partial \beta_v}
&=
\left(1+\frac{1}{\alpha}\right)
|\Sigma|^{-\alpha/2}
\exp\!\left(
-\frac{\alpha}{2}
Z_{n,i}^{\top}\Lambda_0 \Sigma^{-1}\Lambda_0 Z_{n,i}
\right)
\nonumber\\
&\quad\times
\Bigl\{
J_{1,i}
+ J_{0,i}
-\alpha \sqrt{h_n}\,
J_{0,i}
Z_{n,i}^{\top}
\Lambda_0 \Sigma^{-1}
D_i(\bm\beta)
\Bigr\}
+
\mathscr R_{\ref{thm:beta-clt},i},
\label{eq:leading_second_beta}
\end{align}
where the remainder term is given by
\begin{align*}
\mathscr R_{\ref{thm:beta-clt},i}
&=
\left(1+\frac{1}{\alpha}\right)
|\Sigma|^{-\alpha/2}
\exp\!\left(
-\frac{\alpha}{2}
Z_{n,i}^{\top}\Lambda_0 \Sigma^{-1}\Lambda_0 Z_{n,i}
\right)
\\
&\quad\times
\Bigl(
\alpha \sqrt{h_n}\,
Z_{n,i}^{\top}\Lambda_0 \Sigma^{-1} D_i(\bm\beta)
- K_i(\bm\theta)
+\tfrac12 K_i(\bm\theta)^2 e^{\Upsilon_i}
\Bigr)
\left(J_{0,i}+J_{1,i}\right)
\\
&\quad+
\left(1+\frac{1}{\alpha}\right)
|\Sigma|^{-\alpha/2}
\exp\!\left(
-\frac{\alpha}{2}
Z_{n,i}^{\top}\Lambda_0 \Sigma^{-1}\Lambda_0 Z_{n,i}
\right)
e^{-K_i(\bm\theta)} J_{2,i}.
\end{align*}

Using the inequality $e^{-x}\le 1+x$ for $x\ge0$, together with the order
properties of $J_{0,i}$ and $K_i(\bm\theta)$, and applying
Lemmas~\ref{lem:poly-negligible-mv} and
\ref{lem:linear_form_uniform}, we obtain
\[
\sup_{\bm\theta\in\Theta}
\max_{1\le i\le n}
\bigl|\mathscr R_{\ref{thm:beta-clt},i}\bigr|
=
o_P(h_n^{\,r}),
\qquad 0\le r<\tfrac32.
\]
Consequently,
\[
\frac{1}{n h_n}\sum_{i=1}^n
\mathscr R_{\ref{thm:beta-clt},i}
\xrightarrow{P}0,
\qquad
\frac{1}{n h_n}\sum_{i=1}^n J_{0,i}
\xrightarrow{P}0,
\]
uniformly in $\bm\theta\in\Theta$. Hence,
\begin{align*}
&\frac{1}{n h_n}
\sum_{i=1}^n
\frac{\partial^2 V_{n,i}^\alpha}{\partial \beta_u \partial \beta_v}
\\\xrightarrow{P}&
\left(1+\frac{1}{\alpha}\right)
|\Sigma|^{-\alpha/2}
\left|
I + \alpha \Lambda_0 \Sigma^{-1} \Lambda_0
\right|^{-1/2}
% \\
% &
\int
\Bigg[
-\alpha
\left(\frac{\partial^2 a(x,\bm{\beta}_0)}{\partial{\beta_u}\partial \beta_v} \right)^{\top}
\Sigma^{-1}
D(\bm{\beta})
% \\[0.3em]
% &
+\alpha
\left({\frac{\partial a(x,\bm{\beta}_0)}{\partial\beta_u}} \right)^{\top}
\Sigma^{-1}
\left({\frac{\partial a(x,\bm{\beta}_0)}{\partial\beta_v}} \right)
\\[0.3em]
&\qquad
-\alpha^2
\left({\frac{\partial a(x,\bm{\beta}_0)}{\partial\beta_u}} \right)^{\top}
\Sigma^{-1}
\Lambda_0
\left(
I + \alpha \Lambda_0 \Sigma^{-1} \Lambda_0
\right)^{-1}
\Lambda_0^{\top}
\Sigma^{-1}
\left({\frac{\partial a(x,\bm{\beta}_0)}{\partial\beta_v}} \right)
\\[0.3em]
&\qquad
+\alpha^2
\left(\frac{\partial^2 a(x,\bm{\beta}_0)}{\partial{\beta_u}\partial \beta_v} \right)^{\top}
\Sigma^{-1}
\Lambda_0
\left(
I + \alpha \Lambda_0 \Sigma^{-1} \Lambda_0
\right)^{-1}
\Lambda_0
\Sigma^{-1}
D(\bm{\beta})
\Bigg]
\,\pi(dx),
\end{align*}
uniformly in $\bm{\theta}$. Evaluating the above limit at
$\bm{\theta}=\bm{\theta}_0$, it simplifies to
\[
(1+\alpha)^{-d/2}
|\Sigma_0|^{-\alpha/2}
\int
\left({\frac{\partial a(x,\bm{\beta}_0)}{\partial\beta_u}} \right)^{\top}
\Sigma_0^{-1}
\left({\frac{\partial a(x,\bm{\beta}_0)}{\partial\beta_v}} \right)
\,\pi(dx).
\]

Since the limiting expression is continuous in the parameters, condition
\eqref{eqn:C_n^alpha} follows. This completes the proof of the asymptotic
normality of $\hat{\bm{\beta}}_n^\alpha$, with asymptotic variance given in
\eqref{eqn:thm_2}.

\end{proof}

% \subsubsection
\begin{proof}[{Proof of Theorem \ref{thm:asymp-sigma}}]
~\\
\textbf{Step 1.}
Following the same strategy as in the previous proofs, we proceed analogously.
First, from the first derivative of $U_{n,i}^\alpha$, we obtain
\begin{align}\label{eqn:vech_deriv_u}
\frac{\partial U_{n,i}^\alpha}{\partial \vech(\Sigma)}
=
\Big(
\frac{\partial U_{n,i}^\alpha}{\partial \sigma_{11}},
\frac{\partial U_{n,i}^\alpha}{\partial \sigma_{21}},
\frac{\partial U_{n,i}^\alpha}{\partial \sigma_{31}},
\ldots,
\frac{\partial U_{n,i}^\alpha}{\partial \sigma_{d1}},
\frac{\partial U_{n,i}^\alpha}{\partial \sigma_{22}},
\ldots,
\frac{\partial U_{n,i}^\alpha}{\partial \sigma_{d2}},
\ldots,
\frac{\partial U_{n,i}^\alpha}{\partial \sigma_{(d-1)(d-1)}},
\frac{\partial U_{n,i}^\alpha}{\partial \sigma_{(d-1)d}},
\frac{\partial U_{n,i}^\alpha}{\partial \sigma_{dd}}
\Big)^\top .
\end{align}

At ${\bm\theta}_0 = c({\bm\beta_0},\vech(\Sigma_0))$, each component in
\eqref{eqn:vech_deriv_u} is explicitly given by
\eqref{eqn:1st_deriv_true_u_sigma}. A direct calculation shows that the
expectation of each component is zero. Consequently, the matrix
\begin{equation}\label{eqn:step_2_u_sigma_clt}
\frac{1}{n}\sum_{i=1}^n
\frac{\partial U_{n,i}^\alpha}{\partial \vech(\Sigma)}
\frac{\partial U_{n,i}^\alpha}{\partial \vech(\Sigma)^\top}
\end{equation}
converges in probability to a deterministic matrix $\Xi$, whose elements
are obtained by taking expectations in
\eqref{eqn:1st_deriv_true_u_sigma_sqr} using arguments identical to those
employed previously.

More precisely, for $(kl,rs)$ indexing the entries of $\Xi$, we have
\begin{align*}
((\Xi))_{kl,rs}
&=
\frac{(1+\alpha)^2}{2(1+2\alpha)^{d/2+2}}
|\Sigma_0|^{-\alpha}
\operatorname{tr}\!\left(
\Sigma_0^{-1}S_{kl}\Sigma_0^{-1}S_{rs}
\right)
% \\
% &\quad
+
\frac{\alpha^2(1+\alpha)^2}{(1+2\alpha)^{d/2+2}}
|\Sigma_0|^{-\alpha}
\operatorname{tr}(\Sigma_0^{-1}S_{kl})
\operatorname{tr}(\Sigma_0^{-1}S_{rs})
\\
&\qquad-
\frac{\alpha^2}{4(1+\alpha)^d}
|\Sigma_0|^{-\alpha}
\operatorname{tr}(\Sigma_0^{-1}S_{kl})
\operatorname{tr}(\Sigma_0^{-1}S_{rs}).
\end{align*}

Therefore, by the multivariate central limit theorem, the quantity in
\eqref{eqn:step_2_u_sigma_clt} is asymptotically normal, which establishes
Step~1 of the proof.

\textbf{Step 2:} As done for CLT of $\bm \beta$ step 2, here also to prove step 2 for $\Sigma$ we will show the convergence of the matrix element wise and the matrix convergence in probability by frobenious norm. Hence from the equation~\eqref{eqn:2nd_deriv_sigma} can be rewritten as 
\begin{align}\begin{aligned}\label{eqn:V_double_thm_2_step_2}
 \frac{\partial^2}{\partial\sigma_{rs}\partial\sigma_{kl}}V_{n,i}^\alpha
 =
 \mathcal{J}_{0,i}'+ \mathcal{J}_{1,i}'+\mathcal{J}_{2,i}'+\mathcal{J}_{3,i}' +\mathcal{J}_{4,i}',
\end{aligned}
\end{align}
where $\mathcal{J}_{m,i}'$ where $m\in \{0,1,2,3,4\}$ is defined in~\eqref{eqn:all_J_sigma} to \eqref{eqn:J4i-2}. Using, taylor series expansion of $\mathcal{C}_i$, of the following form, 
\begin{align*}
    &\left(1+\frac{1}{\alpha}\right)
|\Sigma|^{-\alpha/2}
\exp\!\left(-K_i({\bm\theta})\right)
\exp\!\left(-\frac{\alpha}{2}Z_{n,i}^\top \Lambda_0\Sigma^{-1}\Lambda_0Z_{n,i}\right)
\\
=&
\left(1+\frac{1}{\alpha}\right)
|\Sigma|^{-\alpha/2}
\fst{1 - K_i({\bm\theta}) \exp\fst{\upsilon_i}}
\exp\!\left(-\frac{\alpha}{2}Z_{n,i}^\top \Lambda_0\Sigma^{-1}\Lambda_0Z_{n,i}\right),
\end{align*}
where, $|\upsilon_i|\le |K_i({\bm\theta})|$. Hence using above expansion, we can write~\eqref{eqn:V_double_thm_2_step_2} structure as, 
\begin{align}
    \begin{aligned}
        &\mathcal{J}_{0,i}'+ \mathcal{J}_{1,i}'+\mathcal{J}_{2,i}'+\mathcal{J}_{3,i}' +\mathcal{J}_{4,i}'\\
        =&
        \mathcal{J}_{0,i}^0
        - \frac{\alpha^2}{4}\left(1+\frac{1}{\alpha}\right)
        |\Sigma|^{-\alpha/2}
        \exp\!\left(-\frac{\alpha}{2}Z_{n,i}^\top \Lambda_0\Sigma^{-1}\Lambda_0Z_{n,i}\right)
        \mathcal{J}_{0,i}^1
        \\
        &\qquad\qquad
        + \frac{\alpha}{2}\left(1+\frac{1}{\alpha}\right)
        |\Sigma|^{-\alpha/2}
        \exp\!\left(-\frac{\alpha}{2}Z_{n,i}^\top \Lambda_0\Sigma^{-1}\Lambda_0Z_{n,i}\right)
        \mathcal{J}_{0,i}^2
        \\
        &\qquad\qquad
        + \mathscr{R}_{\ref{thm:asymp-sigma},i}
\end{aligned}\label{eqn:expansion_sigma_prob_conv}
\end{align}
where $\mathscr{R}_{\ref{thm:asymp-sigma},i}$ is defined by the equation~\eqref{eqn:remainder_thm_sigma_asymp}. Using the fact that,
\[
\sup_{{\bm\theta}}\max_{1\le i \le n}{|K_i({\bm\theta})|} = o_P(h_n^{r})\;, \quad 0\le r<\frac{1}{2},
\]
and using Lemma~\ref{lem:poly-negligible-mv}, we obtain the corresponding range of $r$ for 
$\sup_{\bm{\theta}} \max_{1 \le i \le n} \bigl|\mathscr{R}_{\ref{thm:asymp-sigma},i}\bigr|$ 
as $0 \le r < \tfrac{1}{2}$; consequently, the remainder term satisfies,
\[
\frac{1}{n}\sum_{i=1}^n\mathscr{R}_{\ref{thm:asymp-sigma},i}
\xrightarrow{P}0,
\]
and then we get
\begin{align}
\begin{aligned}
    &\frac{1}{n}\sum_{i=1}^n
    \frac{\partial^2V_{n,i}^\alpha}{\partial\sigma_{rs}\partial\sigma_{kl}}\\
    & \xrightarrow[]{P}
    \frac{\alpha}{2(1+\alpha)^{d / 2}}
    |\Sigma|^{-\alpha / 2}
    \left[
    \frac{\alpha}{2}
    \operatorname{tr}\left(A_{kl}\right)
    \operatorname{tr}\left(A_{rs}\right)
    +
    \operatorname{tr}\left(A_{k l} A_{r s}\right)
    \right]
    % \\
    % &\quad
    -\frac{\alpha^2}{4}\left(1+\frac{1}{\alpha}\right)
    |\Sigma|^{-\alpha/2}
    |I_d+\alpha \mathcal{M}|^{-1/2}
    \\
    &\quad\times
    \Bigg[
    \operatorname{tr}(A_{kl})\operatorname{tr}(A_{rs})
    -
    \operatorname{tr}(A_{kl})
    \operatorname{tr}\!\big((I_d+\alpha \mathcal{M})^{-1}C_{rs}\big)
    -
    \operatorname{tr}(A_{rs})
    \operatorname{tr}\!\big((I_d+\alpha \mathcal{M})^{-1}C_{kl}\big)
    \\
    &\qquad
    +
    \operatorname{tr}\!\big((I_d+\alpha \mathcal{M})^{-1}C_{kl}\big)
    \operatorname{tr}\!\big((I_d+\alpha \mathcal{M})^{-1}C_{rs}\big)
    % \\
    % &\qquad\qquad
    +2\operatorname{tr}\!\big(
    (I_d+\alpha \mathcal{M})^{-1}
    C_{kl}
    (I_d+\alpha \mathcal{M})^{-1}
    C_{rs}
    \big)
    \Bigg]
    \\
    &\quad
    +\frac{\alpha}{2}\left(1+\frac{1}{\alpha}\right)
    |\Sigma|^{-\alpha/2}
    |I_d+\alpha \mathcal{M}|^{-1/2}
    % \\
    % &\qquad\qquad
    \times
    \Big[
    -\operatorname{tr}(A_{kl}A_{rs})
    +
    \operatorname{tr}\!\big(
    (I_d+\alpha \mathcal{M})^{-1}
    \Lambda_0 B_{rs,kl}\Lambda_0
    \big)
    \Big],
\end{aligned}
\end{align}
where,
\begin{equation*}
C_{k l}=\Lambda_0^{\top} \Sigma^{-1} S_{k l} \Sigma^{-1} \Lambda_0, \quad
C_{r s}=\Lambda_0^{\top} \Sigma^{-1} S_{r s} \Sigma^{-1} \Lambda_0, \quad
\mathcal{M } = \Lambda_0^T\Sigma^{-1}\Lambda_0.
\end{equation*}

Hence at $\Sigma=\Sigma_0$, the above converges boils down to,
\begin{align*}
  &\frac{\alpha}{2(1+\alpha)^{d / 2}}|\Sigma_0|^{-\alpha / 2}
  \left[
  \frac{\alpha}{2} \operatorname{tr}\left(A_{kl}\right)
  \operatorname{tr}\left(A_{rs}\right)
  +
  \operatorname{tr}\left(A_{k l} A_{r s}\right)
  \right]
  \\
&\quad
-\frac{\alpha^2}{4}\left(1+\frac{1}{\alpha}\right)
|\Sigma_0|^{-\alpha/2}
(1+\alpha)^{-d/2}
\left[
\operatorname{tr}(A_{kl})\operatorname{tr}(A_{rs})
\left\{1 -\frac{2}{1+\alpha}+\frac{1}{(1+\alpha)^2}\right\}
+\frac{2}{(1+\alpha)^2}\operatorname{tr}(A_{kl}A_{rs})
\right]
\\
&\qquad
+\frac{\alpha}{2}\left(1+\frac{1}{\alpha}\right)
|\Sigma_0|^{-\alpha/2}
(1+\alpha )^{-d/2}
\operatorname{tr}(A_{kl}A_{rs})
\times\thrd{\frac{2}{(1+\alpha)}-1}
\\
=&
\left|\Sigma_0\right|^{-\alpha / 2}
(1+\alpha)^{-d / 2}
\left[
\frac{\alpha^2}{4(1+\alpha)}
\operatorname{tr}\left(A_{k l}\right)
\operatorname{tr}\left(A_{r s}\right)
+
\frac{1}{2(1+\alpha)}
\operatorname{tr}\left(A_{k l} A_{r s}\right)
\right].
\end{align*}
Finally owing to Lemma~\ref{lem:PD_M}, we get invertibility and the theorem is proved.

\end{proof}

\subsubsection{Proof of Theorem~\ref{thm:asymp-indep} (Asymptotic Independence)}

\begin{proof}
We show that all cross-product terms arising from the joint expansion of the score functions vanish in probability. This establishes asymptotic independence. The proof proceeds in two steps.

\medskip
\noindent
\textbf{Step 1 (First-order cross terms).}
Arguing as in Step~1 of Theorems~\ref{thm:beta-clt} and~\ref{thm:asymp-sigma}, consider the joint score vector
\begin{equation}\label{eqn:joint_normal_all_params_u_rewrite}
\left(
\frac{1}{\sqrt{n h_n}}
\sum_{i=1}^n
\frac{\partial U_{n,i}^\alpha}{\partial \bm\beta},
\;
\frac{1}{\sqrt{n}}
\sum_{i=1}^n
\frac{\partial U_{n,i}^\alpha}{\partial \vech(\Sigma)}
\right)^\top
\Bigg|_{({\bm\beta_0},\Sigma_0)} .
\end{equation}
Using the same moment bounds and Lindeberg-type conditions as before, we compute the conditional expectation of the cross-product term~\eqref{eqn:1st_deriv_true_u_cross} and obtain
\begin{align}
    \mathbb{E}\!\left[
\left.
\frac{\partial U_{n,i}^\alpha}{\partial \beta_u}
\right|_{({\bm\beta_0},\Sigma_0)}
\left.
\frac{\partial U_{n,i}^\alpha}{\partial \sigma_{kl}}
\right|_{({\bm\beta_0},\Sigma_0)}
\Bigg|\mathcal{F}_{t_{i-1}^n}
\right]
=0 \label{eqn:thm_4_u_expc}
\end{align}

Consequently, the limiting covariance between the $\bm\beta$- and $\vech(\Sigma)$-components are zero, and the vector in \eqref{eqn:joint_normal_all_params_u_rewrite} is asymptotically normal with vanishing first-order cross terms.

\medskip
\noindent
\textbf{Step 2 (Second-order cross terms).}
We next consider the cross-derivative term in \eqref{eqn:1st_deriv_cross}. Expanding the second derivative as in \eqref{eqn:expansion_cross_deriv}, we obtain
\begin{align*}
\frac{1}{n\sqrt{h_n}}\sum_{i=1}^n
\left.
\frac{\partial^2 V_{n,i}^\alpha}{\partial \beta_u\,\partial \sigma_{rs}}
\right.
&=
\left(1+\frac{1}{\alpha}\right)
|\Sigma|^{-\alpha/2}
\frac{1}{n}\sum_{i=1}^n
\exp\!\left(
-\frac{\alpha}{2}
Z_{n,i}^\top \Lambda_0\Sigma^{-1}\Lambda_0 Z_{n,i}
\right)
\\
&\quad\times
\alpha
\Bigg[
\fst{\frac{\partial a_{i-1}}{\partial\beta_u}}^\top
\Sigma^{-1}\Lambda_0 Z_{n,i}
\Bigg\{
\frac{\alpha}{2}
\operatorname{tr}(\Sigma^{-1}S_{rs})
-
\frac{\alpha}{2}
Z_{n,i}^\top
\Lambda_0^\top
\Sigma^{-1}S_{rs}\Sigma^{-1}
\Lambda_0 Z_{n,i}
\Bigg\}
\\
&\qquad\qquad
+
\fst{\frac{\partial a_{i-1}}{\partial\beta_u}}^\top
\Sigma^{-1}S_{rs}\Sigma^{-1}\Lambda_0 Z_{n,i}
\Bigg]
+
\frac{1}{n\sqrt{h_n}}
\sum_{i=1}^n
\mathscr{R}_{\ref{thm:asymp-indep},i},
\end{align*}
where the remainder term $\mathscr{R}_{\ref{thm:asymp-indep},i}$ is defined in \eqref{eqn:asymp_indep_rem}.  
Using the uniform bound
\[
\sup_{{\bm\theta}}\max_{1\le i\le n}|K_i({\bm\theta})|
=o_P(h_n^{r}),
\qquad 0<r<\tfrac12,
\]
together with polynomial growth of the derivatives and Lemma~\ref{lem:poly-negligible-mv}, we obtain
\[
\frac{1}{n\sqrt{h_n}}
\sum_{i=1}^n
\mathscr{R}_{\ref{thm:asymp-indep},i}
\xrightarrow{P}0 .
\]

Moreover, conditioning on $\mathcal{F}_{t_{i-1}^n}$ and applying Lemmas~\ref{lem:gaussian_tilt_multivariate} and~\ref{lem:linear_form_uniform}, we have
\begin{align*}
\mathbb{E}\Bigg[
&\exp\!\left(
-\frac{\alpha}{2}
Z_{n,i}^\top \Lambda_0^\top \Sigma^{-1}\Lambda_0 Z_{n,i}
\right)
\alpha
\Bigg\{
\fst{\frac{\partial a_{i-1}}{\partial\beta_u}}^\top
\Sigma^{-1}\Lambda_0 Z_{n,i}
\Bigg(
\frac{\alpha}{2}
\operatorname{tr}(\Sigma^{-1}S_{rs})
-
\frac{\alpha}{2}
Z_{n,i}^\top
\Lambda_0^\top
\Sigma^{-1}S_{rs}\Sigma^{-1}
\Lambda_0 Z_{n,i}
\Bigg)
\\
&\qquad\qquad
+
\fst{\frac{\partial a_{i-1}}{\partial\beta_u}}^\top
\Sigma^{-1}S_{rs}\Sigma^{-1}\Lambda_0 Z_{n,i}
\Bigg\}
\Bigg|\mathcal{F}_{t_{i-1}^n}
\Bigg]
=0 .
\end{align*}
Therefore,
\[
\frac{1}{n\sqrt{h_n}}
\sum_{i=1}^n
\frac{\partial^2 V_{n,i}^\alpha}{\partial \beta_u\,\partial \sigma_{rs}}
\xrightarrow{P}0 .
\]

Combining Steps~1 and~2 completes the proof of asymptotic independence.
\end{proof}

\subsubsection{Proof of Theorem \ref{thm:joint-clt} (Joint Normality)}

\begin{proof}
By Theorem~\ref{thm:beta-clt},
\[
\sqrt{n h_n}\bigl(\hat{\bm\beta}_n^\alpha - \bm\beta_0\bigr)
\;\Rightarrow\;
\mathcal{N}\!\left(0,\Sigma_\beta\right),
\]
and by Theorem~\ref{thm:asymp-sigma},
\[
\sqrt{n}\Bigl(
\operatorname{vech}(\hat\Sigma_n^{\alpha})
-
\operatorname{vech}(\Sigma_0)
\Bigr)
\;\Rightarrow\;
\mathcal{N}\!\left(0,\mathscr{L}^{-\top}\,\Xi\,\mathscr{L}^{-1}\right).
\]

Since the true parameters $\bm\beta_0$ and $\Sigma_0$ are non-random, conditioning
on them is immaterial and does not affect the marginal limiting distributions.
Moreover, asymptotic independence of the two estimators follows from
Theorem~\ref{thm:asymp-indep}.

Consequently, for any conformable vectors $t_1$ and $t_2$, the joint
characteristic function factorizes asymptotically as
\[
\mathbb{E}\exp\!\left\{
i t_1^\top \sqrt{n h_n}(\hat{\bm\beta}_n^\alpha - \bm\beta_0)
+
i t_2^\top \sqrt{n}\bigl(
\operatorname{vech}(\hat\Sigma_n^{\alpha})
-
\operatorname{vech}(\Sigma_0)
\bigr)
\right\}
=
\phi_\beta(t_1)\,\phi_\Sigma(t_2) + o(1),
\]
where $\phi_\beta$ and $\phi_\Sigma$ denote the characteristic functions of the
corresponding Gaussian limits.
By Lévy’s continuity theorem, the joint convergence follows.
\end{proof}

\subsection{Some Lemmas with Proofs}

The proofs of lemmas~\ref{lem:delta-moment-mv}, \ref{lem:poly-negligible-mv},
\ref{lem:ergodic_poly_mv}, \ref{lem:gaussian_tilt_multivariate},
\ref{lem:linear_form_uniform}, \ref{lem:quad_form_mean}, and
\ref{lem:u_v_deriv_conver_p} in our multivariate setting are extensions of the corresponding univariate
results established in \cite{LeeSong2013}.

\begin{lem}[Uniform moment bound]
\label{lem:moment-bound}
Under the Assumption~\ref{assm:all_assumption} (A1)--(A3) for the process define by the equation \ref{eq:gen_multivariate_diffusion_gen}, have following bound,
\[
\mathbb{E}
\left[
\sup_{0\le r\le t}\|X_r\|^p
\right]
\le
C_{p,T}(1+\E \|X_0\|^p).
\]
\end{lem}

\begin{proof}
Before proving this it is worth for mentioning the univariate proof of this lemma given in \cite{SchillingPartzsch2014} Chapter 19. We follow the similar steps. From the integral representation
\[
X_u
=
X_0
+
\int_0^u b(s,X_s)\,ds
+
\int_0^u \Lambda\,dW_s ,
\]
the triangle inequality in $\mathbb{R}^d$ together with the elementary bound
$(a+b+c)^p \le 3^{p-1}(a^p+b^p+c^p)$ for $a,b,c\ge0$ yields
\[
\|X_u\|^p
\le
3^{p-1}\|X_0\|^p
+
3^{p-1}
\Big\|
\int_0^u b(s,X_s)\,ds
\Big\|^p
+
3^{p-1}
\Big\|
\int_0^u \Lambda\,dW_s
\Big\|^p .
\]

By Jensen’s inequality applied to the probability measure $ds/u$ on $[0,u]$,
\[
\Big\|
\int_0^u b(s,X_s)\,ds
\Big\|^p
\le
u^{p-1}
\int_0^u \|b(s,X_s)\|^p\,ds .
\]
Using the polynomial growth condition $\|b(s,x)\|\le K(1+\|x\|)$, we obtain
\[
\Big\|
\int_0^u b(s,X_s)\,ds
\Big\|^p
\le
K^p u^{p-1}
\int_0^u (1+\|X_s\|)^p\,ds .
\]
Since $(1+\|x\|)^p\le 2^{p-1}(1+\|x\|^p)$, it follows that
\[
\Big\|
\int_0^u b(s,X_s)\,ds
\Big\|^p
\le
C u^{p-1}
\int_0^u
\left(1+\sup_{r\le s}\|X_r\|^p\right)\,ds .
\]

Let
\[
M_u := \int_0^u \Lambda\,dW_s .
\]
Then $M_u$ is a continuous $\mathbb{R}^d$-valued martingale with quadratic
variation
\[
\langle M\rangle_u = u\,\Lambda\Lambda^\top .
\]
By the Burkholder--Davis--Gundy inequality, for any $p\ge2$,
\[
\mathbb{E}
\left[
\sup_{0\le r\le u}
\|M_r\|^p
\right]
\le
C_p
\mathbb{E}
\left[
\left(\operatorname{tr}\langle M\rangle_u\right)^{p/2}
\right]
=
C_p\,u^{p/2}\,\|\Lambda\|_{\mathrm{F}}^{p},
\]
where $\|\cdot\|_{\mathrm{F}}$ denotes the Frobenius norm.

Combining the above bounds and taking expectations, we obtain
\[
\mathbb{E}
\left[
\sup_{r\le u}\|X_r\|^p
\right]
\le
C_1
+
C_2
\int_0^u
\left(
1+
\mathbb{E}
\left[
\sup_{r\le s}\|X_r\|^p
\right]
\right) ds .
\]
Setting
\[
w(u)
:=
\mathbb{E}
\left[
\sup_{r\le u}\|X_r\|^p
\right],
\]
Gronwall’s lemma implies that
\begin{align*}
    \E\thrd{\sup_{0\le u\le t }\|X_u\|^p}\le C_3 (1+\E\|X_0\|^p)
\end{align*}
which completes the proof.
\end{proof}
\begin{cor}
\label{cor:conditional-moment}
Under the assumptions of Lemma~\ref{lem:moment-bound}, let
$0\le t_1<t_2\le T$ and let $\mathcal F_t$ denote the natural filtration of
$(X_t)_{t\ge0}$.
Then, for any $p\ge2$, there exists a constant $C_p>0$ such that
\[
\mathbb{E}
\left(
\sup_{t_1\le u\le t_2}\|X_u\|^p
\;\middle|\;
\mathcal F_{t_1}
\right)
\le
C_p
\left(
1+\|X_{t_1}\|^p
\right)
\quad
\text{a.s.}
\]
\end{cor}

\begin{lem}%[Moment bound for discretization remainder (multivariate)]
\label{lem:delta-moment-mv}
Suppose that Assumptions~\ref{assm:all_assumption} (A1)  and $a(\cdot,\cdot)\in\mathscr{P}$, then, for any integer $k\ge1$, there exists a constant $C_k>0$ such that
\[
\mathbb{E}_{{\bm\theta}_0}
\!\left(
\|\Delta_{n,i}\|^k
\;\middle|\;
\mathcal F_{t_{i-1}^n}
\right)
\le
C_k\,h_n^{3k/2}
\left(1+\|X_{t_{i-1}^n}\|\right)^k ,
\]
where $\|\cdot\|$ denotes the Euclidean norm and
$\mathcal F_{t_{i-1}^n}$ is sigma-field generated by $\{X_{t_j^n}:j\le i-1\}$.
\end{lem}

\begin{proof}
By the global Lipschitz continuity of $a$, there exists $C>0$ such that
\[
\|a(X_s,{\bm\theta}_0)-a(X_{t_{i-1}^n},{\bm\theta}_0)\|
\le
C\,\|X_s-X_{t_{i-1}^n}\|,
\qquad s\in[t_{i-1}^n,t_i^n].
\]
Hence,
\[
\|\Delta_{n,i}\|
\le
C\int_{t_{i-1}^n}^{t_i^n}\|X_s-X_{t_{i-1}^n}\|\,ds .
\]
Applying Jensen’s inequality yields
\[
\|\Delta_{n,i}\|^k
\le
C\,h_n^{k-1}
\int_{t_{i-1}^n}^{t_i^n}
\|X_s-X_{t_{i-1}^n}\|^k\,ds .
\]
Taking conditional expectation with respect to $\mathcal F_{t_{i-1}^n}$ gives
\[
\mathbb{E}_{{\bm\theta}_0}
\!\left(
\|\Delta_{n,i}\|^k
\mid
\mathcal F_{t_{i-1}^n}
\right)
\le
C\,h_n^{k-1}
\int_{t_{i-1}^n}^{t_i^n}
\mathbb{E}_{{\bm\theta}_0}
\!\left(
\|X_s-X_{t_{i-1}^n}\|^k
\mid
\mathcal F_{t_{i-1}^n}
\right)
ds .
\]

Next, note that
\[
X_s-X_{t_{i-1}^n}
=
\int_{t_{i-1}^n}^s a(X_u,{\bm\theta}_0)\,du
+
\Lambda_0 (W_s-W_{t_{i-1}^n}).
\]
Using the polynomial growth of $a$ and standard moment bounds for diffusions using Lemma~\ref{lem:moment-bound} we obtain
\[
\mathbb{E}_{{\bm\theta}_0}
\!\left(
\|X_s-X_{t_{i-1}^n}\|^k
\mid
\mathcal F_{t_{i-1}^n}
\right)
\le
C\,h_n^{k/2}
\left(1+\|X_{t_{i-1}^n}\|\right)^k,
\qquad
s\in[t_{i-1}^n,t_i^n].
\]
Substituting this bound into the above expression yields
\[
\mathbb{E}_{{\bm\theta}_0}
\!\left(
\|\Delta_{n,i}\|^k
\mid
\mathcal F_{t_{i-1}^n}
\right)
\le
C\,h_n^{k-1}\cdot h_n\cdot h_n^{k/2}
\left(1+\|X_{t_{i-1}^n}\|\right)^k
=
C_k\,h_n^{3k/2}
\left(1+\|X_{t_{i-1}^n}\|\right)^k,
\]
which completes the proof.
\end{proof}

\begin{lem}
\label{lem:poly-negligible-mv}
Let $f:\mathbb{R}^d\times\Theta\to\mathbb{R}^p$ belong to $\mathscr P$.
Assume that Assumptions~\ref{assm:all_assumption} \emph{(A1)} and \emph{(A3)} hold,
and that $n h_n^{q}\to0$ for some $q>1$.
Let $P_j:\mathbb{R}^p\to\mathbb{R}$ be an arbitrary scalar polynomial function
of total degree $j$, that is,

\[
P_j(x)
=
\sum_{\nu\in\mathbb{N}_0^p:\,|\nu|\le j}
c_\nu\, x^\nu,
\qquad
x^\nu := \prod_{r=1}^p x_r^{\nu_r},
\quad
|\nu| := \sum_{r=1}^p \nu_r .
\]

Then, for any integers $j,k,l\ge0$ and any $m>-3l/2$,
\[
\sup_{{\bm\theta}\in\Theta}
\max_{1\le i\le n}
\left|
P_j\!\left(f(X_{t_{i-1}^n},{\bm\theta})\right)
\right|
\,
\left\|Z_{n,i}\right\|^{ k}
\,
\left\|\Delta_{n,i}\right\|^{l}
\,h_n^{m}
=
o_P(h_n^{r}),
\qquad
0\le r<\tfrac{3}{2}l+m .
\]
\end{lem}

\begin{proof}
Let $P_j:\mathbb{R}^p\to\mathbb{R}$ be a scalar polynomial of total degree $j$.
Since $f\in\mathscr P$, there exist constants $C>0$ and $\kappa\ge0$ such that
\[
\sup_{{\bm\theta}\in\Theta}
\left|
P_j\!\left(f(x,{\bm\theta})\right)
\right|
\le
C(1+\|x\|)^{\kappa},
\qquad x\in\mathbb{R}^d .
\]
By the polynomial growth of $P_j\circ f$ ,
%and the identity $\|u^{\otimes k}\|=\|u\|^{k}$, 
we obtain
\[
\left|
P_j\!\left(f(X_{t_{i-1}^n},{\bm\theta})\right)
\right|
\,
\left\|Z_{n,i}\right\|^{k}
\,
\left\|\Delta_{n,i}\right\|^{l} 
\le
C(1+\|X_{t_{i-1}^n}\|)^{\kappa}
\|Z_{n,i}\|^{k}
\|\Delta_{n,i}\|^{l}.
\]
Fix $p_1>1$ such that $p_1(m+3l/2-r)>q$.
Since $\sup_i\mathbb E\|Z_{n,i}\|^{p_1}<\infty$ and, by
Lemma~\ref{lem:delta-moment-mv},
\[
\mathbb E\!\left(
\|\Delta_{n,i}\|^{p_1}
\mid
\mathcal F_{t_{i-1}^n}
\right)
\le
C h_n^{3p_1/2}
(1+\|X_{t_{i-1}^n}\|)^{p_1},
\]
it follows that
\[
\mathbb E
\!\left(
\left|
P_j\!\left(f(X_{t_{i-1}^n},{\bm\theta})\right)
\right|^{p_1}
\,
\|Z_{n,i}\|^{k p_1}
\,
\|\Delta_{n,i}\|^{l p_1}
\right)
\le
C h_n^{3l p_1/2}
\mathbb E(1+\|X_{t_{i-1}^n}\|)^{p_1(\kappa+l)}.
\]
Under Assumption~\emph{(A3)}, the expectation on the right-hand side is
uniformly bounded in $i$. Consequently,
\[
n
\max_{1\le i\le n}
\mathbb E
\!\left(
\left|
P_j\!\left(f(X_{t_{i-1}^n},{\bm\theta})\right)
\right|^{p_1}
\,
\|Z_{n,i}\|^{k p_1}
\,
\|\Delta_{n,i}\|^{l p_1}
\right)
\le
C n h_n^{3l p_1/2}.
\]
Applying Markov’s inequality together with the union bound, we obtain
\[
\mathbb P
\!\left(
\sup_{{\bm\theta}\in\Theta}
\max_{1\le i\le n}
\left|
P_j\!\left(f(X_{t_{i-1}^n},{\bm\theta})\right)
\right|
\,
\|Z_{n,i}\|^{k}
\,
\|\Delta_{n,i}\|^{l}
\,h_n^{m}
>
h_n^{r}
\right)
\le
C n h_n^{p_1(m+3l/2-r)}.
\]
Since $r<m+3l/2$ and $n h_n^{q}\to0$, the right-hand side converges to zero,
which completes the proof.
\end{proof}

\begin{lem}
\label{lem:ergodic_poly_mv}
Assume that Assumptions~\ref{assm:all_assumption} \textnormal{A1--A3} hold. Let
\[
f : \mathbb{R}^d \times \Theta \to \mathbb{R}
\]
be a scalar-valued function. Suppose that $h_n \to 0$ and $n h_n \to \infty$, and
that $f$ is continuously differentiable with respect to both $x$ and ${\bm\theta}$.
Assume further that $f$ and all its partial derivatives belong to the class
$\mathscr{P}$. Then,
\[
\sup_{{\bm\theta}\in \Theta}
\left|
\frac{1}{n}\sum_{i=1}^n f\!\left(X_{t_{i-1}^n},{\bm\theta}\right)
-
\int f(x,{\bm\theta})\,\mu_0(dx)
\right|
\;\xrightarrow[n\to\infty]{\mathbb P_{{\bm\theta}_0}}\; 0,
\]
uniformly in ${\bm\theta}$.
\end{lem}

\begin{proof}
Fix ${\bm\theta}\in\Theta$. By the multivariate mean value theorem,
\[
f(X_s,{\bm\theta})-f(X_{t_{i-1}^n},{\bm\theta})
=
\int_0^1
\Bigl(
\nabla_x f\!\left(
X_{t_{i-1}^n}+u\left(X_s-X_{t_{i-1}^n}\right),{\bm\theta}
\right)
\Bigr)^\top
\left(X_s-X_{t_{i-1}^n}\right)\,du .
\]
Consequently,
\[
\left|
f(X_s,{\bm\theta})-f(X_{t_{i-1}^n},{\bm\theta})
\right|
\le
\|X_s-X_{t_{i-1}^n}\|
\int_0^1
\left\|
\nabla_x f\!\left(
X_{t_{i-1}^n}+u\left(X_s-X_{t_{i-1}^n}\right),{\bm\theta}
\right)
\right\|\,du .
\]
Applying the Cauchy--Schwarz inequality yields
\begin{align}
\mathbb{E}_{{\bm\theta}_0}
\left|
f(X_s,{\bm\theta})-f(X_{t_{i-1}^n},{\bm\theta})
\right|
&\le
\Bigl(
\mathbb{E}_{{\bm\theta}_0}\|X_s-X_{t_{i-1}^n}\|^2
\Bigr)^{1/2}
\nonumber\\
&\quad\times
\Biggl(
\mathbb{E}_{{\bm\theta}_0}\int_0^1
\left\|
\nabla_x f\!\left(
X_{t_{i-1}^n}+u\left(X_s-X_{t_{i-1}^n}\right),{\bm\theta}
\right)
\right\|^2\,du
\Biggr)^{1/2}.
\label{eq:kessler_trick_mv}
\end{align}

The right-hand side is uniformly bounded in ${\bm\theta}$ by Assumptions~\ref{assm:all_assumption} \textnormal{(A1--A3)} and the polynomial growth condition on $f$ and its derivatives. The conclusion then follows by arguments analogous to those in Lemma~8 of \citet{Kessler1997} and Lemma~1 of \citet{Uchida2010Contrast}.
\end{proof}

\begin{lem}\label{lem:gaussian_tilt_multivariate}
Suppose that Assumptions~\ref{assm:all_assumption} (A1)--(A3) hold.
Let $f:\mathbb{R}^d\times\Theta\to\mathbb{R}\in\mathscr{P}$ be differentiable
with respect to $x$ and ${\bm\theta}$, with derivatives belonging to $\mathscr{P}$.
Let $A({\bm\theta})$ and $B({\bm\theta})$ be symmetric positive definite matrices,
independent of $\{X_t\}$.
Define
\[
\mathscr{Q}=I_d+\Lambda_0\Sigma^{-1}\Lambda_0 .
\]
Then, uniformly in ${\bm\theta}$,
\begin{align}
\frac{1}{n}\sum_{i=1}^n
f(X_{t_{i-1}^n},{\bm\theta})\,
\exp\!\left(
-\frac{\alpha}{2}Z_{n,i}^\top\Lambda_0\Sigma^{-1}\Lambda_0 Z_{n,i}
\right)
&\xrightarrow{P}
|\mathscr{Q}|^{-1/2}
\int f(x,{\bm\theta})\,d\mu_0(x),
\label{eq:lim0}\\[0.4em]
\frac{1}{n}\sum_{i=1}^n
f(X_{t_{i-1}^n},{\bm\theta})\,
Z_{n,i}^\top A Z_{n,i}\,
\exp\!\left(
-\frac{\alpha}{2}Z_{n,i}^\top\Lambda_0\Sigma^{-1}\Lambda_0 Z_{n,i}
\right)
&\xrightarrow{P}
|\mathscr{Q}|^{-1/2}
\operatorname{tr}(A\mathscr{Q}^{-1})
\int f(x,{\bm\theta})\,d\mu_0(x),
\label{eq:lim1}
\end{align}
\begin{align}
\begin{aligned}
    &\quad\frac{1}{n}\sum_{i=1}^n
f(X_{t_{i-1}^n},{\bm\theta})\,
(Z_{n,i}^\top A Z_{n,i})(Z_{n,i}^\top B Z_{n,i})\,
\exp\!\left(
-\frac{\alpha}{2}Z_{n,i}^\top\Lambda_0\Sigma^{-1}\Lambda_0 Z_{n,i}
\right)
\\\xrightarrow{P}&|\mathscr{Q}|^{-1/2}\times
\Big[
2\operatorname{tr}(A\mathscr{Q}^{-1}B\mathscr{Q}^{-1})
+
\operatorname{tr}(A\mathscr{Q}^{-1})
\operatorname{tr}(B\mathscr{Q}^{-1})
\Big]
\times\int f(x,{\bm\theta})\,d\mu_0(x).
\end{aligned}\label{eq:lim2}
\end{align}
\end{lem}

\begin{proof}
We prove \eqref{eq:lim2}; %the proofs of \eqref{eq:lim0} and \eqref{eq:lim1}
 We need to calculate the conditional expectation give the sigma field $\mathcal{F}_{t_{i-1}^n}$. Since $Z_{n,i}$ is independent of $\mathcal{F}_{t_{i-1}^n}$,
\begin{align*}
\mathbb{E}_0\left(h_i({\bm\theta})\middle|\mathcal{F}_{t_{i-1}^n}\right) = &\mathbb{E}_0\!\left[
f(X_{t_{i-1}^n},{\bm\theta})
(Z_{n,i}^\top A Z_{n,i})(Z_{n,i}^\top B Z_{n,i})
\exp\!\left(
-\frac{\alpha}{2}Z_{n,i}^\top\Lambda_0\Sigma^{-1}\Lambda_0 Z_{n,i}
\right)
\;\middle|\;
\mathcal{F}_{t_{i-1}^n}
\right]
\\
&\qquad=
f(X_{t_{i-1}^n},{\bm\theta})\,
\mathbb{E}\!\left[
(Z^\top A Z)(Z^\top B Z)
\exp\!\left(
-\frac{\alpha}{2}Z^\top\Lambda_0\Sigma^{-1}\Lambda_0 Z
\right)
\right],
\end{align*}
where $Z\sim N(0,I_d)$.

Let $\mathscr{Q}=I_d+\Lambda_0\Sigma^{-1}\Lambda_0$.
Then the exponential tilting corresponds to a Gaussian distribution
$N(0,\mathscr{Q}^{-1})$, and
\[
\mathbb{E}\!\left[
\exp\!\left(
-\frac{\alpha}{2}Z^\top\Lambda_0\Sigma^{-1}\Lambda_0 Z
\right)
\right]
=|\mathscr{Q}|^{-1/2}.
\]
Using standard Gaussian moment identities, we obtain
\begin{align*}
&\mathbb{E}\!\left[
(Z^\top A Z)(Z^\top B Z)
\exp\!\left(
-\frac{\alpha}{2}Z^\top\Lambda_0\Sigma^{-1}\Lambda_0 Z
\right)
\right]
\\
&\qquad=
|\mathscr{Q}|^{-1/2}
\Big[
2\operatorname{tr}(A\mathscr{Q}^{-1}B\mathscr{Q}^{-1})
+
\operatorname{tr}(A\mathscr{Q}^{-1})
\operatorname{tr}(B\mathscr{Q}^{-1})
\Big].
\end{align*}
By Lemma \ref{lem:ergodic_poly_mv},
\[
\frac{1}{n}\sum_{i=1}^n f(X_{t_{i-1}^n},{\bm\theta})
\xrightarrow{P}
\int f(x,{\bm\theta})\,d\mu_0(x),
\]
Hence,
\begin{align*}
    \frac{1}{n}\sum_{i=1}^{n}h_i({\bm\theta})\xrightarrow{P}& |\mathscr{Q}|^{-1/2}\times
\Big[
2\operatorname{tr}(A\mathscr{Q}^{-1}B\mathscr{Q}^{-1})
+
\operatorname{tr}(A\mathscr{Q}^{-1})
\operatorname{tr}(B\mathscr{Q}^{-1})
\Big]
\nonumber
\times\int f(x,{\bm\theta})\,d\mu_0(x).
\end{align*}
Hence by using compactness of the parameter space and using Theorem 20 in \cite{IbragimovHasminskii1981}, we get required result exactly following the proof of Lemma 4 of \cite{LeeSong2013}. 
To proof the equations~\ref{eq:lim0} and \ref{eq:lim1} are similar to the prove of above, after deriving the expectations of quadratic forms.
\end{proof}

\begin{lem}\label{lem:linear_form_uniform}
Suppose that Assumptions~\ref{assm:all_assumption} \textnormal{(A1)--(A3)} hold.
Let $f:\mathbb{R}^d\times\Theta\to\mathbb{R}^d$ be a vector-valued function such that
$f(\cdot,{\bm\theta})$ and its derivatives with respect to $x$ and ${\bm\theta}$
belong to the polynomial growth class $\mathscr{P}$.
Then, as $n\to\infty$,
\begin{equation}\label{eq:linear_form_op}
\frac{1}{n\sqrt{h_n}}
\sum_{i=1}^n
f\!\left(X_{t_{i-1}^n},{\bm\theta}\right)^{\!\top}
Z_{n,i}
\exp\!\left(
-\frac{\alpha}{2}
Z_{n,i}^\top \Lambda_0\Sigma^{-1}\Lambda_0 Z_{n,i}
\right)
=
o_P(1),
\end{equation}
uniformly in ${\bm\theta}\in\Theta$.
\end{lem}

\begin{proof}
Define
\[
h_i({\bm\theta})
=
\frac{1}{n\sqrt{h_n}}\,
f\!\left(X_{t_{i-1}^n},{\bm\theta}\right)^{\!\top}
Z_{n,i}
\exp\!\left(
-\frac{\alpha}{2}
Z_{n,i}^\top \Lambda_0\Sigma^{-1}\Lambda_0 Z_{n,i}
\right).
\]
Let $\mathcal F_{t_{i}^n}$  sigma-field generated by $\{X_{t_j^n},\, j\le i\}$.
Since $Z_{n,i}$ is independent of $\mathcal F_{t_{i-1}^n}$ and centered, we have
\[
\mathbb E_0\!\left[h_i({\bm\theta})\mid\mathcal F_{t_{i-1}^n}\right]=0,
\qquad i=1,\dots,n,
\]
so that $\{h_i({\bm\theta}),\mathcal F_{t_{i}^n}\}$ is a martingale difference array.

Moreover,
\begin{align*}
\sum_{i=1}^n
\mathbb E_0\!\left[h_i^2({\bm\theta})\mid\mathcal F_{t_{i-1}^n}\right]
&=
\frac{1}{n^2 h_n}
\sum_{i=1}^n
\mathbb E_0\!\left[
Z_{n,i}^\top f\!\left(X_{t_{i-1}^n},{\bm\theta}\right) f\!\left(X_{t_{i-1}^n},{\bm\theta}\right)^\top Z_{n,i}
\exp\!\left(
-\alpha Z_{n,i}^\top\Lambda_0\Sigma^{-1}\Lambda_0 Z_{n,i}
\right)\Bigg|\mathcal{F}_{t_{i-1}^n}
\right] \\
&=
\frac{|\widetilde{\mathscr Q}|^{-1/2}}{n^2 h_n}
\sum_{i=1}^n
f^\top\!\left(X_{t_{i-1}^n},{\bm\theta}\right)
f\!\left(X_{t_{i-1}^n},{\bm\theta}\right),
\end{align*}
where
\[
\widetilde{\mathscr Q}
=
I_d+2\Lambda_0\Sigma^{-1}\Lambda_0 .
\]
By ergodicity and the polynomial growth of $f$, the right-hand side is
$o_P(1)$, implying \eqref{eq:linear_form_op} for each fixed ${\bm\theta}$.

To establish uniform convergence, we verify conditions (28)–(29) of
Theorem~20 in \cite{IbragimovHasminskii1981}.
First, differentiating $h_i({\bm\theta})$ with respect to ${\bm\theta}$ yields
\begin{align*}
\|\nabla_{\bm\theta} h_i({\bm\theta})\|
&\le
\frac{1}{n\sqrt{h_n}}
\left\|\nabla_{\bm\theta} f(X_{t_{i-1}^n},{\bm\theta})\right\|
\|Z_{n,i}\|
\exp\!\left(
-\frac{\alpha}{2}
Z_{n,i}^\top\Lambda_0\Sigma^{-1}\Lambda_0 Z_{n,i}
\right) \\
&\le
\frac{C_\alpha}{n\sqrt{h_n}}
\left(\|Z_{n,i}\|+\|Z_{n,i}\|^3\right)
\left(1+\|X_{t_{i-1}^n}\|\right)^C,
\end{align*}
where the last inequality follows from Assumptions~(A1)--(A3) and Gaussian
moment bounds. Hence,
\begin{equation}\label{eq:grad_bound_linear}
\sup_{{\bm\theta}\in\Theta}\|\nabla_{\bm\theta} h_i({\bm\theta})\|
\le
\frac{C_\alpha}{n\sqrt{h_n}}
\left(\|Z_{n,i}\|+\|Z_{n,i}\|^3\right)
\left(1+\|X_{t_{i-1}^n}\|\right)^C .
\end{equation}

Let $\mathscr D=\dim(\Theta)$.
By the mean value theorem, for each $i$ there exists a point
\[
{\bm\theta}^\ast
\in
\left\{
\lambda{\bm\theta}_1+(1-\lambda){\bm\theta}_2:\ \lambda\in(0,1)
\right\}
\subset\Theta
\]
such that
\[
h_i({\bm\theta}_1)-h_i({\bm\theta}_2)
=
\nabla_{\bm\theta} h_i({\bm\theta}^\ast)\,({\bm\theta}_1-{\bm\theta}_2).
\]
Hence, by Cauchy's inequality and Jensen's inequality,
\begin{align*}
\mathbb E_0
\left|
\sum_{i=1}^n\left(h_i({\bm\theta}_1)-h_i({\bm\theta}_2)\right)
\right|^{2\mathscr D}
&\le
|{\bm\theta}_1-{\bm\theta}_2|^{2\mathscr D}
\,
\mathbb E_0
\left|
\sum_{i=1}^n \nabla_{\bm\theta} h_i({\bm\theta}^\ast)
\right|^{2\mathscr D} \\[0.3em]
&\le
\frac{C_\alpha}{n^{\mathscr D} h_n^{\mathscr D}}
|{\bm\theta}_1-{\bm\theta}_2|^{2\mathscr D}
\,
\mathbb E_0
\left[
\frac{1}{n}\sum_{i=1}^n
\left(\|Z_{n,i}\|+\|Z_{n,i}\|^3\right)^2
\left(1+\|X_{t_{i-1}^n}\|\right)^{2C}
\right]^{\mathscr D} \\[0.3em]
&\le
\frac{C}{n^{\mathscr D} h_n^{\mathscr D}}
|{\bm\theta}_1-{\bm\theta}_2|^{2\mathscr D},
\end{align*}

where the last inequality follows from the existence of all Gaussian
moments and the polynomial moment bounds implied by
Assumptions~\textnormal{(A1)--(A3)}.
Thus, condition (29) of Theorem~20 in \cite{IbragimovHasminskii1981} holds,
and the proof is complete.
\end{proof}
\begin{lem}\label{lem:app_sigma_min}
Let $\alpha>0$, $\Sigma \in \mathbb S_{++}^d$ denote the space of
$d\times d$ symmetric positive definite matrices, and let
$\Lambda_0 := \Sigma_0^{1/2} \in \mathbb R^{d\times d}$ be the
symmetric positive definite square root of $\Sigma_0$.
Consider the function $\Psi(\Sigma)$ defined in
equation~\eqref{eqn:Psi_Sigma}.
Then $\Psi(\Sigma)$ admits a unique minimizer given by
\[
\Sigma^\star = \Lambda_0 \Lambda_0 = \Sigma_0 .
\]
\end{lem}

\begin{proof}
Define the change of variables
\[
M := \Lambda_0^{-\top}\Sigma\Lambda_0^{-1},
\qquad
\Sigma = \Lambda_0 M \Lambda_0 .
\]
Since $\Lambda_0$ is invertible, this is a bijection on $\mathbb S_{++}^d$.
Moreover,
\[
|\Sigma| = |\Lambda_0|^2 |M|,
\qquad
\Lambda_0\Sigma^{-1}\Lambda_0 = M^{-1}.
\]
Hence, minimizing $\Psi(\Sigma)$ is equivalent to minimizing
\[
\Phi(M)
=
\frac{1}{(1+\alpha)^{d/2}}|M|^{-\alpha/2}
-
\left(1+\frac{1}{\alpha}\right)
|M|^{-\alpha/2}
|I_d+\alpha M^{-1}|^{-1/2}.
\]

Let $m_1,\dots,m_d>0$ be the eigenvalues of $M$. Then
\[
|M|=\prod_{i=1}^d m_i,
\qquad
|I_d+\alpha M^{-1}|=\prod_{i=1}^d\left(1+\frac{\alpha}{m_i}\right),
\]
and
\[
\Phi(M)
=
\prod_{i=1}^d m_i^{-\alpha/2}
\left[
\frac{1}{(1+\alpha)^{d/2}}
-
\left(1+\frac{1}{\alpha}\right)
\prod_{i=1}^d\left(1+\frac{\alpha}{m_i}\right)^{-1/2}
\right].
\]

Define
\[
P=\prod_{j=1}^d m_j^{-\alpha/2},
\qquad
Q=\prod_{j=1}^d\left(1+\frac{\alpha}{m_j}\right)^{-1/2}.
\]
Then
\[
\Phi=\frac{P}{(1+\alpha)^{d/2}}-\left(1+\frac{1}{\alpha}\right)PQ.
\]

For each $i$,
\[
\frac{\partial P}{\partial m_i}=-\frac{\alpha}{2m_i}P,
\qquad
\frac{\partial Q}{\partial m_i}
=
\frac{\alpha}{2m_i^2(1+\alpha/m_i)}Q.
\]
Thus
\[
\frac{\partial\Phi}{\partial m_i}
=
-\frac{\alpha}{2m_i}\frac{P}{(1+\alpha)^{d/2}}
-
\left(1+\frac{1}{\alpha}\right)
\left[
-\frac{\alpha}{2m_i}PQ
+
\frac{\alpha}{2m_i^2(1+\alpha/m_i)}PQ
\right].
\]

Setting $\partial\Phi/\partial m_i=0$ and dividing by
$\frac{\alpha}{2m_i}P>0$, we obtain
\begin{align}
   -\frac{1}{(1+\alpha)^{d/2}}
+
\left(1+\frac{1}{\alpha}\right)
\left[
Q-\frac{Q}{m_i(1+\alpha/m_i)}
\right]
=0.\label{eqn:psi_inbetween} 
\end{align}

Equation \eqref{eqn:psi_inbetween} must hold for every $i$.
Since the left-hand side depends on $i$ only through $m_i$,
and the function
\[
m \longmapsto \frac{1}{m(1+\alpha/m)}=\frac{1}{m+\alpha}
\]
is strictly decreasing on $(0,\infty)$, the equality can hold for all $i$
only if
\[
m_1=m_2=\cdots=m_d=c.
\]

Hence $M=cI_d$. Substituting into $\Phi$,
\[
\Phi(c)
=
c^{-d\alpha/2}
\left[
(1+\alpha)^{-d/2}
-
\left(1+\frac{1}{\alpha}\right)
\left(1+\frac{\alpha}{c}\right)^{-d/2}
\right].
\]
A direct differentiation yields $\Phi'(c)=0$ if and only if $c=1$,
and $\Phi''(1)>0$. Therefore $M^\star=I_d$.

Finally,
\[
\Sigma^\star=\Lambda_0 M^\star\Lambda_0
=\Lambda_0\Lambda_0,
\]
which completes the proof.
\end{proof}
\begin{lem}\label{lem:quad_form_mean}
Suppose that Assumption~\ref{assm:all_assumption} (A1)--(A3) hold.
Let $A(x,{\bm\theta})$ be a $d\times d$ matrix-valued function whose entries,
as well as their derivatives with respect to $x$ and ${\bm\theta}$,
belong to the class $\mathscr{P}$.
Define
\[
\mathscr{Q} = I_d + \Lambda_0 \Sigma^{-1} \Lambda_0 .
\]
Then, uniformly in ${\bm\theta}$,
\begin{equation}\label{eq:lim4}
\frac{1}{n}\sum_{i=1}^n
Z_{n,i}^\top A(X_{t_{i-1}^n},{\bm\theta}) Z_{n,i}\,
\exp\!\left(
-\frac{\alpha}{2}Z_{n,i}^\top\Lambda_0\Sigma^{-1}\Lambda_0 Z_{n,i}
\right)
\xrightarrow{P}
|\mathscr{Q}|^{-1/2}
\int \operatorname{tr}\!\big(A(x,{\bm\theta})\mathscr{Q}^{-1}\big)\,d\mu_0(x).
\end{equation}
\end{lem}

\begin{proof}
Define
\[
h_{n,i}
=
\frac{1}{n}
Z_{n,i}^\top A(X_{t_{i-1}^n},{\bm\theta}) Z_{n,i}\,
\exp\!\left(
-\frac{\alpha}{2}Z_{n,i}^\top\Lambda_0\Sigma^{-1}\Lambda_0 Z_{n,i}
\right).
\]

Since $Z_{n,i}$ is independent of $\mathcal{F}_{t_{i-1}^n}$, it follows that
\[
\mathbb{E}_0\!\left[
h_{n,i}
\;\middle|\;
\mathcal{F}_{t_{i-1}^n}
\right]
=
|\mathscr{Q}|^{-1/2}
\operatorname{tr}\!\big(
A(X_{t_{i-1}^n},{\bm\theta})\mathscr{Q}^{-1}
\big).
\]

Moreover, letting
\[
\widetilde{\mathscr{Q}} = I_d + 2\Lambda_0\Sigma^{-1}\Lambda_0,
\]
a direct Gaussian calculation yields
\begin{align*}
\mathbb{E}_0\!\left[
\sum_{i=1}^n h_{n,i}^2
\;\middle|\;
\mathcal{F}_{t_{i-1}^n}
\right]
&=
\frac{1}{n^2}
\sum_{i=1}^n
|\widetilde{\mathscr{Q}}|^{-1/2}
\\
&\quad\times
\Big[
2\operatorname{tr}\!\big(
A(X_{t_{i-1}^n},{\bm\theta})
\widetilde{\mathscr{Q}}^{-1}
A(X_{t_{i-1}^n},{\bm\theta})
\widetilde{\mathscr{Q}}^{-1}
\big)
\\
&\qquad\quad
+
\operatorname{tr}\!\big(
A(X_{t_{i-1}^n},{\bm\theta})
\widetilde{\mathscr{Q}}^{-1}
\big)^2
\Big].
\end{align*}

By the polynomial growth of $A(\cdot,{\bm\theta})$ and the ergodicity of
$\{X_t\}$ under Assumptions~\ref{assm:all_assumption}~(A1)--(A3),
the above conditional second moment is of order $O(1)$.

Furthermore, since the derivatives of $A(x,{\bm\theta})$ belong to $\mathscr{P}$,
there exists a constant $C_\alpha>0$ such that
\begin{equation*}\label{eq:grad_bound}
\sup_{{\bm\theta}}
\big|
\nabla h_{n,i}({\bm\theta})
\big|
\le
\frac{C_\alpha}{n}
\big(1+\|Z_{n,i}\|^2\big)^2
\big(1+\|X_{t_{i-1}^n}\|\big)^C .
\end{equation*}

Therefore, using the same argument as in Lemma~4 of
\cite{LeeSong2013}, the moment conditions required for uniform convergence
in Theorem~20 of \cite{IbragimovHasminskii1981} are satisfied.
This completes the proof.
\end{proof}

% \begin{lem}
%     Under similar setup Assumptions~\ref{assm:all_assumption} (A1)--(A3), we have following  uniform convergence over ${\bm\theta}$, and let $f:R^d\times \Theta\to R^d $ and along with $\|f\|$ and also it's derivative also in the polynimial class $\mathscr{P}$ then, \begin{align}
%         \frac{1}{n\sqrt{h_n}} \sum_{i=1}^n f\left(X_{t_{i-1}^n}, \eta\right)^\top Z_{n,i} \mathrm{e}^{-\frac{\alpha}{2}  Z_{n,i}^\top \Lambda_0\Sigma^{-1}\Lambda_0}=o_P(1)
%     \end{align}
%     uniformly over ${\bm\theta}$.
% \end{lem}

\begin{lem}\label{lem:u_v_deriv_conver_p}
Suppose that Assumption~\ref{assm:all_assumption} and conditions \emph{(A1)--(A7)} are satisfied, and assume that
\[
n h_n^2 \to 0 .
\]
Then, for each index $s$ and each pair $(u,v)$, the following convergences in probability hold:
\begin{align}
\frac{1}{\sqrt{nh_n}}\sum_{i=1}^n
\left.\frac{\partial V_{n,i}^{\alpha}}{\partial \beta_s}\right|_{({\bm\beta_0},\Sigma_0)}
-
\frac{1}{\sqrt{nh_n}}\sum_{i=1}^n
\left.\frac{\partial U_{n,i}^{\alpha}}{\partial \beta_s}\right|_{({\bm\beta_0},\Sigma_0)}
&\xrightarrow{{P}} 0,
\label{eqn:lemm_U_V_P_conv_beta}
\\
\frac{1}{\sqrt{n}}\sum_{i=1}^n
\left.\frac{\partial V_{n,i}^{\alpha}}{\partial \sigma_{uv}}\right|_{({\bm\beta_0},\Sigma_0)}
-
\frac{1}{\sqrt{n}}\sum_{i=1}^n
\left.\frac{\partial U_{n,i}^{\alpha}}{\partial \sigma_{uv}}\right|_{({\bm\beta_0},\sigma_0)}
&\xrightarrow{{P}} 0 .
\label{eqn:lemm_U_V_P_conv_Sigma}
\end{align}
\end{lem}

\begin{proof}
The proof follows the same line of argument as Lemma~6 in \cite{LeeSong2013}.
We establish~\eqref{eqn:lemm_U_V_P_conv_beta}.

From the expressions of the first derivatives of $V_{n,i}^{\alpha}$ and $U_{n,i}^{\alpha}$ given in
\eqref{eqn:1st_deriv_beta} and \eqref{eqn:1st_deriv_u_beta}, respectively, define
\begin{align*}
V_{n,i}'(\beta_s;{\bm\beta_0},\Sigma_0)
&:=
\left.\frac{\partial V_{n,i}^{\alpha}}{\partial \beta_s}\right|_{({\bm\beta_0},\Sigma_0)}
\\
&=
-(1+\alpha)|\Sigma_0|^{-\alpha/2}
\exp\!\left(-\frac{\alpha}{2}Q_i^0\right)
\left(\frac{\partial a_{i-1}({\bm\beta_0})}{\partial{\beta_s}}\right)^{\top}
\Sigma_0^{-1}
\left(\sqrt{h_n}\Lambda_0 Z_{n,i}+\Delta_{n,i}\right),
\\
U_{n,i}'(\beta_s;{\bm\beta_0},\Sigma_0)
&:=
\left.\frac{\partial U_{n,i}^{\alpha}}{\partial \beta_s}\right|_{({\bm\beta_0},\Sigma_0)}
\\
&=
-(1+\alpha)\sqrt{h_n}
|\Sigma_0|^{-\alpha/2}
\exp\!\left(-\frac{\alpha}{2}Z_{n,i}^{\top}Z_{n,i}\right)
\left(\frac{\partial a_{i-1}({\bm\beta_0})}{\partial{\beta_s}}\right)^{\top}
\Lambda_0^{-\top} Z_{n,i},
\end{align*}
where
\[
Q_i^0
=
\frac{1}{h_n}
R_i({\bm\beta_0})^{\top}
\Sigma_0^{-1}
R_i({\bm\beta_0}).
\]

Set
\[
K_{i,0}
:=
\frac{\alpha}{2}Q_i^0
-
\frac{\alpha}{2}Z_{n,i}^{\top}Z_{n,i} .
\]
Then
\begin{align*}
&\left|
V_{n,i}'(\beta_s;{\bm\beta_0},\Sigma_0)
-
U_{n,i}'(\beta_s;{\bm\beta_0},\Sigma_0)
\right|\\
\leq\;&
(1+\alpha)|\Sigma_0|^{-\alpha/2}
\left|
\fst{\frac{\partial a_{i-1}({\bm\beta_0})}{\partial{\beta_s}}}^{\top}
\Sigma_0^{-1}
\Delta_{n,i}
\right|
\left|\exp\fst{K_{i,0}}\right|
\exp\!\left(-\frac{\alpha}{2}Z_{n,i}^{\top}Z_{n,i}\right)
\\
&+
\sqrt{h_n}\,
\frac{\alpha(1+\alpha)}{2}
|\Sigma_0|^{-\alpha/2}
\left|
\fst{\frac{\partial a_{i-1}({\bm\beta_0})}{\partial{\beta_s}}}^{\top}
\Lambda_0^{-1}Z_{n,i}
\right|
\left|
K_{i,0}\exp\fst{\xi_i}
\right|
\exp\!\left(-\frac{\alpha}{2}Z_{n,i}^{\top}Z_{n,i}\right),
\end{align*}
where $|\xi_i|<|K_{i,0}|$.

Using the polynomial growth of ${\frac{\partial a_{i-1}({\bm\beta_0})}{\partial{\beta_s}}}$ and the bound on $\xi_i$, we obtain
\begin{align*}
&\frac{1}{\sqrt{nh_n}}
\sum_{i=1}^n
\left|
V_{n,i}'(\beta_s;{\bm\beta_0},\Sigma_0)
-
U_{n,i}'(\beta_s;{\bm\beta_0},\Sigma_0)
\right|\\
\leq\;&
\exp\fst{\max_i|K_{i,0}|}
\frac{C_\alpha}{\sqrt{nh_n}}
\sum_{i=1}^n
\left(1+\|X_{t_{i-1}^n}\|\right)^C
\left(
\|Z_{n,i}\||K_{i,0}|\sqrt{h_n}
+
\|\Delta_{n,i}\|
\right).
\end{align*}

Taking conditional expectations given $\mathcal{F}_{t_{i-1}^n}$ and applying the Cauchy--Schwarz inequality yield
\begin{align*}
&\frac{1}{\sqrt{nh_n}}
\sum_{i=1}^n
\E\!\left[
\left|
V_{n,i}'(\beta_s;{\bm\beta_0},\Sigma_0)
-
U_{n,i}'(\beta_s;{\bm\beta_0},\Sigma_0)
\right|
\Big|
\mathcal{F}_{t_{i-1}^n}
\right]\\
\leq&
\frac{C_\alpha}{\sqrt{nh_n}}
\sum_{i=1}^n
\left(1+\|X_{t_{i-1}^n}\|\right)^C
\left(
\sqrt{\E\|Z_{n,i}\|^2
\E(|K_{i,0}|^2|\mathcal{F}_{t_{i-1}^n})}
\sqrt{h_n}
+
\E(\|\Delta_{n,i}\||\mathcal{F}_{t_{i-1}^n})
\right).
\end{align*}

Since
\[
\E\!\left(|K_{i,0}|^2\mid\mathcal{F}_{t_{i-1}^n}\right)
\leq
C\left(1+\|X_{t_{i-1}^n}\|\right)^C h_n^2,
\qquad
\exp\fst{\max_i|K_{i,0}|}=O_{\mathbb{P}}(1),
\]
it follows that the above term is of order $O_{\mathbb{P}}(\sqrt{n}h_n)$.

Similarly, using
\[
\E\!\left(|K_{i,0}|^4\mid\mathcal{F}_{t_{i-1}^n}\right)
\leq
C\left(1+\|X_{t_{i-1}^n}\|\right)^C h_n^4,
\]
we obtain
\begin{align*}
\frac{1}{nh_n}
\sum_{i=1}^n
\E\!\left[
\left|
V_{n,i}'(\beta_s;{\bm\beta_0},\Sigma_0)
-
U_{n,i}'(\beta_s;{\bm\beta_0},\Sigma_0)
\right|^2
\Big|
\mathcal{F}_{t_{i-1}^n}
\right]
\leq
\frac{1}{nh_n}
\sum_{i=1}^n
\left(1+\|X_{t_{i-1}^n}\|\right)^C
h_n^3
= O_{\mathbb{P}}(h_n^2).
\end{align*}

Therefore, by Lemma~9 of \cite{GenonCatalotJacod1993},
\eqref{eqn:lemm_U_V_P_conv_beta} follows.
The proof of~\eqref{eqn:lemm_U_V_P_conv_Sigma} proceeds along identical lines and is omitted.
\end{proof}

\begin{lem}
\label{lem:PD_M}
Let $\Sigma_0 \in \mathbb{R}^{d \times d}$ be a symmetric positive definite matrix and
set $\Omega = \Sigma_0^{-1}$. For $\alpha > -1$, define the matrix
\[
M = \bigl( M_{(kl),(rs)} \bigr)_{1 \le k \le l \le d,\; 1 \le r \le s \le d},
\]
of dimension $\frac{d(d+1)}{2} \times \frac{d(d+1)}{2}$, with entries
\begin{equation}
\label{eq:M_def}
M_{(kl),(rs)}
=
\frac{\alpha^2}{4(1+\alpha)}
\,\operatorname{tr}(A_{kl})\,\operatorname{tr}(A_{rs})
+
\frac{1}{2(1+\alpha)}
\,\operatorname{tr}(A_{kl}A_{rs}),
\end{equation}
where
\[
A_{kl} = \Omega S_{kl}, \qquad
S_{kl} =
\begin{cases}
E_{kk}, & \text{if } k = l, \\[0.3em]
E_{kl} + E_{lk}, & \text{if } k \neq l,
\end{cases}
\]
and $E_{ij}$ denotes the $d \times d$ matrix with a one in the $(i,j)$-th entry and
zeros elsewhere. Then the matrix $M$ is symmetric positive definite and, in
particular, invertible.
\end{lem}

\begin{proof}
Let $\mathcal S_d$ denote the vector space of $d\times d$ symmetric matrices.
The collection $\{S_{kl}:1\le k\le l\le d\}$ forms a basis of $\mathcal S_d$.
For any coefficient vector $x=(x_{kl})_{k\le l}$, define
\[
H:=\sum_{k\le l}x_{kl}S_{kl}\in\mathcal S_d.
\]

Using $A_{kl}=\Omega S_{kl}$ and linearity of matrix multiplication,
\[
\sum_{k\le l}x_{kl}A_{kl}
=
\Omega\sum_{k\le l}x_{kl}S_{kl}
=
\Omega H.
\]

Consider the quadratic form associated with $M$:
\begin{align}
x^\top M x
&=
\sum_{k\le l}\sum_{r\le s}
x_{kl}x_{rs}
\Bigg[
\frac{\alpha^2}{4(1+\alpha)}
\operatorname{tr}(A_{kl})\operatorname{tr}(A_{rs})
+
\frac{1}{2(1+\alpha)}
\operatorname{tr}(A_{kl}A_{rs})
\Bigg].
\end{align}

We treat the two terms separately. By bilinearity,
\begin{align}
\sum_{k\le l}\sum_{r\le s}
x_{kl}x_{rs}\operatorname{tr}(A_{kl})\operatorname{tr}(A_{rs})
% &
=\left(\sum_{k\le l}x_{kl}\operatorname{tr}(A_{kl})\right)^2 
% \notag\\
% &
=\left(
\operatorname{tr}\Bigl(\Omega\sum_{k\le l}x_{kl}S_{kl}\Bigr)
\right)^2 
% \notag\\
% &
=\left[\operatorname{tr}(\Omega H)\right]^2.
\end{align}

Similarly, using linearity of the trace,
\begin{align}
\sum_{k\le l}\sum_{r\le s}
x_{kl}x_{rs}\operatorname{tr}(A_{kl}A_{rs})
% &
=
\operatorname{tr}\Bigl(
\Omega\Bigl(\sum_{k\le l}x_{kl}S_{kl}\Bigr)
\Omega\Bigl(\sum_{r\le s}x_{rs}S_{rs}\Bigr)
\Bigr) 
% \notag\\
% &
=
\operatorname{tr}(\Omega H\Omega H).
\end{align}

Combining the two parts yields
\begin{equation}
\label{eq:quad_form}
x^\top M x
=
\frac{\alpha^2}{4(1+\alpha)}\left[\operatorname{tr}(\Omega H)\right]^2
+
\frac{1}{2(1+\alpha)}\operatorname{tr}(\Omega H\Omega H).
\end{equation}

Since $\Omega\succ0$, there exists $\Omega^{1/2}$ such that
\[
\operatorname{tr}(\Omega H\Omega H)
=
\operatorname{tr}\left((\Omega^{1/2}H\Omega^{1/2})^2\right)
=
\|\Omega^{1/2}H\Omega^{1/2}\|_F^2.
\]
Hence $\operatorname{tr}(\Omega H\Omega H)\ge0$, with equality if and only if
$H=0$. As $\alpha>-1$, both coefficients in~\eqref{eq:quad_form} are strictly
positive, implying
\[
x^\top M x>0
\quad\text{for all }x\neq0.
\]

Therefore $M$ is symmetric positive definite and, in particular, invertible.
\end{proof}

%\exp\fst{max_i|K_{i,0}|}

\setcounter{section}{0}

% --- leave appendix-style numbering ---
\renewcommand{\thesection}{S.\arabic{section}}
\renewcommand{\thesubsection}{S.\arabic{section}.\arabic{subsection}}
\renewcommand{\theequation}{S\arabic{section}.\arabic{subsection}.\arabic{equation}}

\setcounter{equation}{0}

\section*{Supplementary}
\addcontentsline{toc}{section}{Supplementary}

\section{Calculation of Derivatives and Expectations}\label{App:Deriv_Expect}
\subsection{Matrix notations and Expectations}

Throughout, $\Sigma$ denotes a symmetric positive definite $d\times d$ matrix.
For $1\le r,s\le d$, let $e_r$ denote the $r$th canonical basis vector in
$\mathbb R^d$ and define
\[
E_{rs}:=e_r e_s^\top .
\]
To account for the symmetry of $\Sigma$, we introduce the matrices
\[
S_{rs}
:=
\begin{cases}
E_{rr}, & r=s,\\[0.2em]
E_{rs}+E_{sr}, & r\neq s.
\end{cases}
\]

With this notation, the derivative of the determinant term is given by
\begin{equation}\label{eq:det-derivative}
\frac{\partial |\Sigma|^{-1}}{\partial \sigma_{rs}}
=
-\,|\Sigma|^{-1}
\operatorname{tr}\!\left(
\Sigma^{-1} S_{rs}
\right).
\end{equation}

Similarly, for $x\in\mathbb R^d$, the derivative of the quadratic form satisfies
\begin{equation}\label{eq:quadform-derivative}
\frac{\partial}{\partial \sigma_{rs}}
\left(
x^\top \Sigma^{-1} x
\right)
=
-\,x^\top
\Sigma^{-1} S_{rs} \Sigma^{-1}
x .
\end{equation}

For further details on matrix differential calculus for symmetric matrices, see,
for example, \citet{Anderson2003}, \citet{MardiaKentBibby1979} or \citet{PetersenPedersen2012}.

\subsection*{Expectation calculations}

Let $Z\sim\mathcal N_d(0,I_d)$ and let $M,B,C\in\mathbb R^{d\times d}$ be symmetric
p.d. matrices such that
\[
Q:=I_d+\alpha M
\]
is symmetric positive definite. Throughout this subsection, expectations are taken
with respect to the law of $Z$.

\medskip

\noindent\textbf{Quadratic form.}
A direct Gaussian calculation yields
\begin{equation}\label{eq:exp_quad_identity}
\mathbb E\!\left[
Z^\top B Z\,
\exp\!\left(-\frac{\alpha}{2}Z^\top M Z\right)
\right]
=
|Q|^{-1/2}
\operatorname{tr}\!\left(
B Q^{-1}
\right).
\end{equation}
Indeed, writing
\[
\mathbb E\!\left[
Z^\top B Z\,
\exp\fst{-\frac{\alpha}{2}Z^\top M Z}
\right]
=
(2\pi)^{-d/2}
\int_{\mathbb R^d}
z^\top B z\,
\exp\fst{-\frac12 z^\top Q z}\,dz,
\]
and using $z^\top B z=\operatorname{tr}(zz^\top B)$ together with the standard
identity
\[
\int_{\mathbb R^d}
zz^\top
\exp\fst{-\frac12 z^\top Q z}\,dz
=
(2\pi)^{d/2}|Q|^{-1/2}Q^{-1},
\]
gives \eqref{eq:exp_quad_identity}.

\medskip

\noindent\textbf{Product of quadratic forms.}
Similarly, for symmetric matrices $B$ and $C$,
\begin{equation}\label{eq:exp_quad_quad_identity}
\begin{aligned}
\mathbb E\!\left[
(Z^\top B Z)(Z^\top C Z)\,
\exp\!\left(-\frac{\alpha}{2}Z^\top M Z\right)
\right]
&=
|Q|^{-1/2}
\Big\{
\operatorname{tr}(B Q^{-1})\operatorname{tr}(C Q^{-1})
\\
&\qquad
+2\,\operatorname{tr}(B Q^{-1} C Q^{-1})
\Big\}.
\end{aligned}
\end{equation}
This follows from Isserlis’ formula (see \cite{Isserlis1918}) applied to a centered Gaussian vector with
covariance matrix $Q^{-1}$.

\subsection{Derivatives of \texorpdfstring{$V_{n,i}^\alpha$}{V n i alpha}}

\begin{align*}
&\begin{aligned}
Q_i(\bm\beta, \Sigma) & :=\frac{1}{h_n}R_i(\bm\beta)^{\top} \Sigma^{-1} R_i(\bm\beta) 
% c_\alpha & :=(1+\alpha)^{-d / 2} \\
% f(\Sigma) & :=|\Sigma|^{-\alpha / 2} \\
% g_i(\bm\beta, \Sigma) & :=\exp \left(-\frac{\alpha}{2 h_n} Q_i(\bm\beta, \Sigma)\right)
% \end{aligned}\\
% &\text { Then }\\
% &V_{n, i}^\alpha=c_\alpha f(\Sigma)-\left(1+\frac{1}{\alpha}\right) f(\Sigma) g_i(\bm\beta, \Sigma)
\end{aligned}.
\end{align*}

\begin{equation}\label{eqn:1st_deriv_beta}
\frac{\partial V_{n, i}^\alpha}{\partial \beta_s}=\left(1+\frac{1}{\alpha}\right)|\Sigma|^{-\alpha / 2} \exp \left(-\frac{\alpha}{2 } Q_i\right) \frac{\alpha}{h_n}\left(\frac{\partial R_i}{\partial \beta_s}\right)^{\top} \Sigma^{-1} R_i
\end{equation}

\begin{equation}\label{eqn:2nd_deriv_beta}
\begin{aligned}
\frac{\partial^2 V_{n, i}^\alpha}{\partial \beta_s \partial \beta_r}
&=
\left(1+\frac{1}{\alpha}\right)|\Sigma|^{-\alpha / 2} \exp \left(-\frac{\alpha}{2 } Q_i\right) 
% \\
% & 
\times\left[\frac{\alpha}{h_n}\left(\frac{\partial^2 R_i}{\partial \beta_s \partial \beta_r}\right)^{\top} \Sigma^{-1} R_i+\frac{\alpha}{h_n}\left(\frac{\partial R_i}{\partial \beta_s}\right)^{\top} \Sigma^{-1} \frac{\partial R_i}{\partial \beta_r}\right. \\
& \qquad\left.-\frac{\alpha^2}{h_n^2}\left(\frac{\partial R_i}{\partial \beta_s}\right)^{\top} \Sigma^{-1} R_i\left(\frac{\partial R_i}{\partial \beta_r}\right)^{\top} \Sigma^{-1} R_i\right]
\end{aligned}
\end{equation}

\begin{align}\label{eqn:1st_deriv_sigma}
\begin{aligned}
\frac{\partial V_{n, i}^\alpha}{\partial \sigma_{r s}}= & -\frac{\alpha}{2(1+\alpha)^{d / 2}}|\Sigma|^{-\alpha / 2} \operatorname{tr}\left(\Sigma^{-1} S_{r s}\right) 
% \\
% & 
+\left(1+\frac{1}{\alpha}\right)|\Sigma|^{-\alpha / 2} \exp \left(-\frac{\alpha}{2} Q_i\right) \\
& \times\left[\frac{\alpha}{2} \operatorname{tr}\left(\Sigma^{-1} S_{r s}\right)-\frac{\alpha}{2 h_n} R_i^{\top} \Sigma^{-1} S_{r s} \Sigma^{-1} R_i\right]
\end{aligned}
\end{align}

\begin{align}\label{eqn:1st_deriv_cross}
\begin{aligned}
\frac{\partial^2 V_{n,i}^\alpha}{\partial \beta_u\,\partial \sigma_{rs}}
&=
\left(1+\frac{1}{\alpha}\right)
|\Sigma|^{-\alpha/2}
\exp\!\left(-\frac{\alpha}{2}Q_i\right)
% \\
% &\quad
\times
\Bigg[
-\frac{\alpha}{h_n}\,
\big({\frac{\partial}{\partial\beta_u}} R_i\big)^\top
\Sigma^{-1}R_i
\left\{
\frac{\alpha}{2}\operatorname{tr}(\Sigma^{-1}S_{rs})
-
\frac{\alpha}{2h_n}
R_i^\top\Sigma^{-1}S_{rs}\Sigma^{-1}R_i
\right\}
\\
&\qquad
-\frac{\alpha}{h_n}
\big({\frac{\partial}{\partial\beta_u}} R_i\big)^\top
\Sigma^{-1}S_{rs}\Sigma^{-1}R_i
\Bigg],
\end{aligned}
\end{align}

\begin{align}\label{eqn:2nd_deriv_sigma}
\begin{aligned}
\frac{\partial^2 V_{n, i}^\alpha}{\partial \sigma_{k l} \partial \sigma_{r s}}= & \frac{\alpha}{2(1+\alpha)^{d / 2}}|\Sigma|^{-\alpha / 2}\left[\frac{\alpha}{2} \operatorname{tr}\left(\Sigma^{-1} S_{k l}\right) \operatorname{tr}\left(\Sigma^{-1} S_{r s}\right)+\operatorname{tr}\left(\Sigma^{-1} S_{k l} \Sigma^{-1} S_{r s}\right)\right] 
\\
&
+\mathcal{C}_i\left[\left\{\left(-\frac{\alpha}{2} \operatorname{tr}\left(\Sigma^{-1} S_{k l}\right)+\frac{\alpha}{2h_n} R_i^{\top} \Sigma^{-1} S_{k l} \Sigma^{-1} R_i\right)\right.\right. \\
& \left.\times\left(\frac{\alpha}{2} \operatorname{tr}\left(\Sigma^{-1} S_{r s}\right)-\frac{\alpha}{2 h_n} R_i^{\top} \Sigma^{-1} S_{r s} \Sigma^{-1} R_i\right)\right\} 
% \\
% & 
\left.+\left(-\frac{\alpha}{2} \operatorname{tr}\left(\Sigma^{-1} S_{k l} \Sigma^{-1} S_{r s}\right)+\frac{\alpha}{2 h_n} R_i^{\top} B_{rs,kl} R_i\right)\right]%\\
\end{aligned}
\end{align}
where, \[\mathcal{C}_i  =  \left(1+\frac{1}{\alpha}\right)
|\Sigma|^{-\alpha/2}
\exp\!\left(-\frac{\alpha}{2}Q_i\right), \quad B_{rs,kl} = \Sigma^{-1}\left(S_{k l} \Sigma^{-1} S_{r s}+S_{r s} \Sigma^{-1} S_{k l}\right) \Sigma^{-1},\]

% \quad A_{kl} = 
%  \Sigma^{-1} S_{kl} 
\subsection{Derivatives of \texorpdfstring{$U_{n,i}^\alpha$}{U n i alpha}}

$U_{n,i}^\alpha$ is defined as
\begin{align*}
U_{n,i}^\alpha
=&\,
\frac{1}{(1+\alpha)^{d/2}}|\Sigma|^{-\alpha/2}
\\
&\quad
-
\left(1+\frac{1}{\alpha}\right)
|\Sigma|^{-\alpha/2}
\exp\!\left[
-\frac{\alpha}{2h_n}
\left(
D_i(\bm\beta)+\sqrt{h_n}\Lambda_0Z_{n,i}
\right)^\top
\Sigma^{-1}
\left(
D_i(\bm\beta)+\sqrt{h_n}\Lambda_0Z_{n,i}
\right)
\right].
\end{align*}

and \[
Y_i := D_i(\bm\beta)+\sqrt{h_n}\,\Lambda_0 Z_{n,i},
\qquad
S_i := Y_i^\top \Sigma^{-1} Y_i.
\]

\begin{equation}\label{eqn:1st_deriv_u_beta}
\frac{\partial U_{n, i}^\alpha}{\partial \beta_u}=-\alpha\left(1+\frac{1}{\alpha}\right)|\Sigma|^{-\alpha / 2} \exp \left(-\frac{\alpha}{2 h_n} S_i\right)\left({\frac{\partial}{\partial\beta_u}} a_{i-1}(\bm\beta)\right)^{\top} \Sigma^{-1} Y_i
\end{equation}

\begin{equation}\label{eqn:1st_deriv_true_u_beta}
\left.\frac{\partial U_{n, i}^\alpha}{\partial \beta_u}\right|_{\left({\bm\beta_0}, \Sigma_0\right)}=-(1+\alpha) \sqrt{h_n}\left|\Sigma_0\right|^{-\alpha / 2} \exp \left(-\frac{\alpha}{2} Z_{n,i}^{\top} Z_{n,i}\right)\left({\frac{\partial}{\partial\beta_u}} a_{i-1}\left({\bm\beta_0}\right)\right)^{\top} \Lambda_0^{-\top} Z_{n,i}
\end{equation}

\begin{equation}\label{eqn:1st_deriv_true_u_beta_sqr}
\begin{aligned}
\frac{\partial U_{n, i}^\alpha}{\partial \beta_u} \big|_{({\bm\beta_0},\Sigma_0)}\frac{\partial U_{n, i}^\alpha}{\partial \beta_v}\big|_{({\bm\beta_0},\Sigma_0)}= & (1+\alpha)^2 h_n\left|\Sigma_0\right|^{-\alpha} \exp \left(-\alpha Z_{n,i}^{\top} Z_{n,i}\right) \\
& \times\left({\frac{\partial}{\partial\beta_u}} a_{i-1}\left({\bm\beta_0}\right)\right)^{\top} \Lambda_0^{-\top} Z_{n,i} Z_{n,i}^{\top} \Lambda_0^{-1}\left({\frac{\partial}{\partial\beta_v}} a_{i-1}\left({\bm\beta_0}\right)\right)
\end{aligned}
\end{equation}

\begin{align}\label{eqn:1st_deriv_u_sigma}
    \begin{aligned}
\frac{\partial U_{n,i}^\alpha}{\partial \sigma_{rs}}
=&\;
-\frac{\alpha}{2(1+\alpha)^{d/2}}
|\Sigma|^{-\alpha/2} \operatorname{tr}(\Sigma^{-1}S_{rs})
\\
&\qquad
+
\left(1+\frac{1}{\alpha}\right)
|\Sigma|^{-\alpha/2}
\exp\!\left(-\frac{\alpha}{2h_n} S_i \right)
% \\
% &\quad
\times
\left[
\frac{\alpha}{2}\operatorname{tr}(\Sigma^{-1}S_{rs})
-
\frac{\alpha}{2h_n} Y_i^\top \Sigma^{-1} S_{rs} \Sigma^{-1} Y_i
\right].
\end{aligned}
\end{align}

\begin{align}\label{eqn:1st_deriv_true_u_sigma}
    \begin{aligned}
\frac{\partial U_{n,i}^\alpha}{\partial {\sigma_{kl}}}  \big|_{({\bm\beta_0},\Sigma_0)}
=&\;
-\frac{\alpha}{2(1+\alpha)^{d/2}}
|\Sigma_0|^{-\alpha/2}
\operatorname{tr}(\Sigma_0^{-1}S_{kl})
\\
&\quad
+
\left(1+\frac{1}{\alpha}\right)
|\Sigma_0|^{-\alpha/2}
\exp\!\left(-\frac{\alpha}{2} Z_{n,i}^\top Z_{n,i}\right)
\left[
\frac{\alpha}{2}\operatorname{tr}(\Sigma^{-1}S_{kl})
-
\frac{\alpha}{2} Z_{n,i}^\top \Lambda_0^{-\top} S_{kl} \Lambda_0^{-1} Z_{n,i}
\right].
\end{aligned}
\end{align}

\begin{align}\label{eqn:1st_deriv_true_u_cross}
\begin{aligned}
    &\left.\frac{\partial U_{n,i}^\alpha}{\partial \beta_u}\right|_{({\bm\beta_0},\Sigma_0)}
\left.\frac{\partial U_{n,i}^\alpha}{\partial \sigma_{kl}}\right|_{({\bm\beta_0},\Sigma_0)}
\\
&=
\sqrt{h_n}\,
(1+\alpha)\,
|\Sigma_0|^{-\alpha}
\exp\!\left(-\frac{\alpha}{2} Z_{n,i}^\top Z_{n,i}\right)
\left({\frac{\partial}{\partial\beta_u}} a_{i-1}({\bm\beta_0})\right)^\top
\Lambda_0^{-\top} Z_{n,i}
\\
&\quad\times
\Bigg[
\frac{\alpha}{2(1+\alpha)^{d/2}}
\operatorname{tr}\!\left(\Sigma_0^{-1} S_{kl}\right)
% \\
% &\qquad
-
\left(1+\frac{1}{\alpha}\right)
\exp\!\left(-\frac{\alpha}{2} Z_{n,i}^\top Z_{n,i}\right)
\left\{
\frac{\alpha}{2}
\operatorname{tr}\!\left(\Sigma_0^{-1} S_{kl}\right)
-
\frac{\alpha}{2}
Z_{n,i}^\top \Lambda_0^{-\top} S_{kl} \Lambda_0^{-1} Z_{n,i}
\right\}
\Bigg].
\end{aligned}
\end{align}

\begin{align}\label{eqn:1st_deriv_true_u_sigma_sqr}
\begin{aligned}
    &\left.\frac{\partial U_{n,i}^\alpha}{\partial \sigma_{kl}}\right|_{({\bm\beta_0},\Sigma_0)}
\left.\frac{\partial U_{n,i}^\alpha}{\partial \sigma_{rs}}\right|_{({\bm\beta_0},\Sigma_0)}\\=&\quad
\frac{\alpha^2}{4(1+\alpha)^d}
|\Sigma_0|^{-\alpha}
\operatorname{tr}\!\left(\Sigma_0^{-1} S_{kl}\right)
\operatorname{tr}\!\left(\Sigma_0^{-1} S_{rs}\right)
\\
&\quad
-
\frac{\alpha}{2(1+\alpha)^{d/2}}
|\Sigma_0|^{-\alpha}
\fst{\frac{1+\alpha}{\alpha}}\exp\!\left(-\frac{\alpha}{2} Z_{n,i}^\top Z_{n,i}\right)
\operatorname{tr}\!\left(\Sigma_0^{-1} S_{kl}\right)
% \\
% &\qquad
\times
\left\{
\frac{\alpha}{2}
\operatorname{tr}\!\left(\Sigma_0^{-1} S_{rs}\right)
-
\frac{\alpha}{2}
Z_{n,i}^\top \Lambda_0^{-\top} S_{rs} \Lambda_0^{-1} Z_{n,i}
\right\}
\\
&\quad
-
\frac{\alpha}{2(1+\alpha)^{d/2}}
|\Sigma_0|^{-\alpha}\fst{\frac{1+\alpha}{\alpha}}
\exp\!\left(-\frac{\alpha}{2} Z_{n,i}^\top Z_{n,i}\right)
\operatorname{tr}\!\left(\Sigma_0^{-1} S_{rs}\right)
% \\
% &\qquad
\times
\left\{
\frac{\alpha}{2}
\operatorname{tr}\!\left(\Sigma_0^{-1} S_{kl}\right)
-
\frac{\alpha}{2}
Z_{n,i}^\top \Lambda_0^{-\top} S_{kl} \Lambda_0^{-1} Z_{n,i}
\right\}
\\
&\quad
+
\left(1+\frac{1}{\alpha}\right)^2
|\Sigma_0|^{-\alpha}
\exp\!\left(-\alpha Z_{n,i}^\top Z_{n,i}\right)
% \\
% &\qquad
\times
\left\{
\frac{\alpha}{2}
\operatorname{tr}\!\left(\Sigma_0^{-1} S_{kl}\right)
-
\frac{\alpha}{2}
Z_{n,i}^\top \Lambda_0^{-\top} S_{kl} \Lambda_0^{-1} Z_{n,i}
\right\}
\\
&\qquad\qquad\qquad\times
\left\{
\frac{\alpha}{2}
\operatorname{tr}\!\left(\Sigma_0^{-1} S_{rs}\right)
-
\frac{\alpha}{2}
Z_{n,i}^\top \Lambda_0^{-\top} S_{rs} \Lambda_0^{-1} Z_{n,i}
\right\}.
\end{aligned}
\end{align}

\section{Some Related Equations}
\renewcommand{\theequation}{S\arabic{section}.\arabic{equation}}
\begin{align}
    J_{0,i} =-\alpha \sqrt{h_n}\left(\frac{\partial^2 a_{i-1}}{\partial \beta_r\partial\beta_s}\right)^{\top} \Sigma^{-1} \Lambda_0 Z_{n,i}\phantom{\alpha \sqrt{h_n}\left(\partial_{\beta_r \beta_s}^2\right.}\label{eqn:J_0}
\end{align}
\begin{align}
\begin{aligned}
J_{1,i}= & -\alpha h_n\left(\frac{\partial^2 a_{i-1}}{\partial \beta_r\partial\beta_s}\right)^{\top} \Sigma^{-1} D_i(\bm\beta) 
% \\
% & 
+\alpha h_n\left(\frac{\partial a_{i-1}}{\partial \beta _r}\right)^{\top} \Sigma^{-1}\left(\frac{\partial a_{i-1}}{\partial \beta _s}\right) \\
& -\alpha^2 h_n\left(\frac{\partial a_{i-1}}{\partial \beta _r}\right)^{\top} \Sigma^{-1} \Lambda_0 Z_{n,i}\left(\frac{\partial a_{i-1}}{\partial \beta _s}\right)^{\top} \Sigma^{-1} \Lambda_0 Z_{n,i} \\
\end{aligned}\label{eqn:J_1}
\end{align}

\begin{align}
\begin{aligned}
J_{2,i}= & -\alpha\left(\frac{\partial^2 a_{i-1}}{\partial \beta_r\partial\beta_s}\right)^{\top} \Sigma^{-1} \Delta_{n, i} 
% \\
% &
-\alpha^2 h_n^{3 / 2}\left[\left(\frac{\partial a_{i-1}}{\partial \beta _r}\right)^{\top} \Sigma^{-1} \Lambda_0 Z_{n,i}\left(\frac{\partial a_{i-1}}{\partial \beta _s}\right)^{\top} \Sigma^{-1} D_i(\bm\beta)\right.
\\
& 
\left.+\left(\frac{\partial a_{i-1}}{\partial \beta _r}\right)^{\top} \Sigma^{-1} D_i(\bm\beta)\left(\frac{\partial a_{i-1}}{\partial \beta _s}\right)^{\top} \Sigma^{-1} \Lambda_0 Z_{n,i}\right] 
% \\
% & 
-\alpha^2 h_n^2\left(\frac{\partial a_{i-1}}{\partial \beta _r}\right)^{\top} \Sigma^{-1} D_i(\bm\beta)\left(\frac{\partial a_{i-1}}{\partial \beta _s}\right)^{\top} \Sigma^{-1} D_i(\bm\beta) \\
& -\alpha^2 h_n\left[\left(\frac{\partial a_{i-1}}{\partial \beta _r}\right)^{\top} \Sigma^{-1} \Lambda_0 Z_{n,i}\left(\frac{\partial a_{i-1}}{\partial \beta _s}\right)^{\top} \Sigma^{-1} \Delta_{n, i}\right. 
% \\
% & 
\left.+\left(\frac{\partial a_{i-1}}{\partial \beta _r}\right)^{\top} \Sigma^{-1} \Delta_{n, i}\left(\frac{\partial a_{i-1}}{\partial \beta _s}\right)^{\top} \Sigma^{-1} \Lambda_0 Z_{n,i}\right] \\
& -\alpha^2 h_n^2\left[\left(\frac{\partial a_{i-1}}{\partial \beta _r}\right)^{\top} \Sigma^{-1} D_i(\bm\beta)\left(\frac{\partial a_{i-1}}{\partial \beta _s}\right)^{\top} \Sigma^{-1} \Delta_{n, i}\right. 
% \\
% & 
\left.+\left(\frac{\partial a_{i-1}}{\partial \beta _r}\right)^{\top} \Sigma^{-1} \Delta_{n, i}\left(\frac{\partial a_{i-1}}{\partial \beta _s}\right)^{\top} \Sigma^{-1} D_i(\bm\beta)\right] \\
& -\alpha^2\left(\frac{\partial a_{i-1}}{\partial \beta _r}\right)^{\top} \Sigma^{-1} \Delta_{n, i}\left(\frac{\partial a_{i-1}}{\partial \beta _s}\right)^{\top} \Sigma^{-1} \Delta_{n, i} .
\end{aligned}\label{eqn:J_2}
\end{align}

% \begin{align*}
%     \operatorname{tr}\left(A_{k l}\right)-\frac{1}{h_n} R_i^{\top} \Sigma^{-1} S_{k l} \Sigma^{-1} R_i=\operatorname{tr}\left(A_{k l}\right)-Z_{n, i}^{\top} \Lambda_0^{\top} \Sigma^{-1} S_{k l} \Sigma^{-1} \Lambda_0 Z_{n, i} \\ -2 \sqrt{h_n} Z_{n, i}^{\top} \Lambda_0^{\top} \Sigma^{-1} S_{k l} \Sigma^{-1} D_i \\ -2 h_n^{-1 / 2} Z_{n, i}^{\top} \Lambda_0^{\top} \Sigma^{-1} S_{k l} \Sigma^{-1} \Delta_{n, i} \\ -h_n D_i^{\top} \Sigma^{-1} S_{k l} \Sigma^{-1} D_i \\ -2 D_i^{\top} \Sigma^{-1} S_{k l} \Sigma^{-1} \Delta_{n, i} \\ -h_n^{-1} \Delta_{n, i}^{\top} \Sigma^{-1} S_{k l} \Sigma^{-1} \Delta_{n, i}
% \end{align*}

%%%%%%%%%%%%%%%%%%%%%%%%%%%%%%%%%%%%%%%%%%%%%%%%%%%%%%%%%%%%
% Decomposition convention
%%%%%%%%%%%%%%%%%%%%%%%%%%%%%%%%%%%%%%%%%%%%%%%%%%%%%%%%%%%%
We define
\begin{equation}
\mathcal{J}'_{m,i}
=
\mathcal{J}^{(0)}_{m,i}
-\mathcal{C}_i \frac{\alpha^2}{4}\,\mathcal{J}^{(1)}_{m,i}
+\mathcal{C}_i \frac{\alpha}{2}\,\mathcal{J}^{(2)}_{m,i},
\qquad m = 0,1,2,3,4,
\label{eqn:all_J_sigma}
\end{equation}
and \[\mathcal{C}_i  =  \left(1+\frac{1}{\alpha}\right)
|\Sigma|^{-\alpha/2}
\exp\!\left(-\frac{\alpha}{2}Q_i\right),\]where
\begin{itemize}
\item $\mathcal{J}^{(0)}_{m,i}$ denotes the term independent of $\mathcal{C}_i$;
\item $\mathcal{J}^{(1)}_{m,i}$ denotes the component multiplied by $-\mathcal{C}_i \alpha^{2}/4$;
\item $\mathcal{J}^{(2)}_{m,i}$ denotes the component multiplied by $+\mathcal{C}_i \alpha/2$.
\end{itemize}

%%%%%%%%%%%%%%%%%%%%%%%%%%%%%%%%%%%%%%%%%%%%%%%%%%%%%%%%%%%%
% J0
%%%%%%%%%%%%%%%%%%%%%%%%%%%%%%%%%%%%%%%%%%%%%%%%%%%%%%%%%%%%
\begin{align}
\mathcal{J}^{(0)}_{0,i}
&=
\frac{\alpha}{2(1+\alpha)^{d/2}}|\Sigma|^{-\alpha/2}
\left[
\frac{\alpha}{2}\operatorname{tr}(A_{kl})\operatorname{tr}(A_{rs})
+
\operatorname{tr}(A_{kl}A_{rs})
\right],
\label{eqn:J0i-0}
\\[0.5em]
\mathcal{J}^{(1)}_{0,i}
&=\Bigg\{
\operatorname{tr}(A_{kl})\operatorname{tr}(A_{rs})
-\operatorname{tr}(A_{kl})
Z_{n,i}^\top\Lambda_0^\top\Sigma^{-1}S_{rs}\Sigma^{-1}\Lambda_0 Z_{n,i}
\nonumber
% \\
% &
% \qquad
-\operatorname{tr}(A_{rs})
Z_{n,i}^\top\Lambda_0^\top\Sigma^{-1}S_{kl}\Sigma^{-1}\Lambda_0 Z_{n,i}
\nonumber\\
&\qquad
+
\big(Z_{n,i}^\top\Lambda_0^\top\Sigma^{-1}S_{kl}\Sigma^{-1}\Lambda_0 Z_{n,i}\big)
\big(Z_{n,i}^\top\Lambda_0^\top\Sigma^{-1}S_{rs}\Sigma^{-1}\Lambda_0 Z_{n,i}\big)
\Bigg\},
\label{eqn:J0i-1}
\\[0.5em]
\mathcal{J}^{(2)}_{0,i}
&=
% \left\{
-\operatorname{tr}(A_{kl}A_{rs})
+
Z_{n,i}^\top\Lambda_0^\top B_{rs,kl}\Lambda_0 Z_{n,i}
% \right\}.
\label{eqn:J0i-2}
\end{align}

%%%%%%%%%%%%%%%%%%%%%%%%%%%%%%%%%%%%%%%%%%%%%%%%%%%%%%%%%%%%
% J1
%%%%%%%%%%%%%%%%%%%%%%%%%%%%%%%%%%%%%%%%%%%%%%%%%%%%%%%%%%%%
\begin{align}
\mathcal{J}^{(1)}_{1,i}
&=
\Bigg\{
-2\sqrt{h_n}\operatorname{tr}(A_{kl})
Z_{n,i}^\top\Lambda_0^\top\Sigma^{-1}S_{rs}\Sigma^{-1}D_i
\nonumber
% \\
% &\qquad
-2\sqrt{h_n}\operatorname{tr}(A_{rs})
Z_{n,i}^\top\Lambda_0^\top\Sigma^{-1}S_{kl}\Sigma^{-1}D_i
\nonumber\\
&\qquad
+2\sqrt{h_n}
\big(Z_{n,i}^\top\Lambda_0^\top\Sigma^{-1}S_{kl}\Sigma^{-1}\Lambda_0 Z_{n,i}\big)
\big(Z_{n,i}^\top\Lambda_0^\top\Sigma^{-1}S_{rs}\Sigma^{-1}D_i\big)
\nonumber\\
&\qquad
+2\sqrt{h_n}
\big(Z_{n,i}^\top\Lambda_0^\top\Sigma^{-1}S_{rs}\Sigma^{-1}\Lambda_0 Z_{n,i}\big)
\big(Z_{n,i}^\top\Lambda_0^\top\Sigma^{-1}S_{kl}\Sigma^{-1}D_i\big)
\Bigg\},
\label{eqn:J1i-1}
\\[0.5em]
\mathcal{J}^{(2)}_{1,i}
&=
% \left\{
2\sqrt{h_n}\,
Z_{n,i}^\top\Lambda_0^\top B_{rs,kl}D_i.
% \right\}.
\label{eqn:J1i-2}
\end{align}

%%%%%%%%%%%%%%%%%%%%%%%%%%%%%%%%%%%%%%%%%%%%%%%%%%%%%%%%%%%%
% J2
%%%%%%%%%%%%%%%%%%%%%%%%%%%%%%%%%%%%%%%%%%%%%%%%%%%%%%%%%%%%
\begin{align}
\mathcal{J}^{(1)}_{2,i}
&=
\Bigg\{
-2h_n^{-1/2}\operatorname{tr}(A_{kl})
Z_{n,i}^\top\Lambda_0^\top\Sigma^{-1}S_{rs}\Sigma^{-1}\Delta_{n,i}
\nonumber
% \\
% &\qquad
-2h_n^{-1/2}\operatorname{tr}(A_{rs})
Z_{n,i}^\top\Lambda_0^\top\Sigma^{-1}S_{kl}\Sigma^{-1}\Delta_{n,i}
\nonumber\\
&\qquad
+2h_n^{-1/2}
\big(Z_{n,i}^\top\Lambda_0^\top\Sigma^{-1}S_{kl}\Sigma^{-1}\Lambda_0 Z_{n,i}\big)
\big(Z_{n,i}^\top\Lambda_0^\top\Sigma^{-1}S_{rs}\Sigma^{-1}\Delta_{n,i}\big)
\nonumber\\
&\qquad
+2h_n^{-1/2}
\big(Z_{n,i}^\top\Lambda_0^\top\Sigma^{-1}S_{rs}\Sigma^{-1}\Lambda_0 Z_{n,i}\big)
\big(Z_{n,i}^\top\Lambda_0^\top\Sigma^{-1}S_{kl}\Sigma^{-1}\Delta_{n,i}\big)
\Bigg\},
\label{eqn:J2i-1}
\\[0.5em]
\mathcal{J}^{(2)}_{2,i}
&=
% \left\{
2h_n^{-1/2}
Z_{n,i}^\top\Lambda_0^\top B_{rs,kl}\Delta_{n,i}.
% \right\}.
\label{eqn:J2i-2}
\end{align}

%%%%%%%%%%%%%%%%%%%%%%%%%%%%%%%%%%%%%%%%%%%%%%%%%%%%%%%%%%%%
% J3
%%%%%%%%%%%%%%%%%%%%%%%%%%%%%%%%%%%%%%%%%%%%%%%%%%%%%%%%%%%%
\begin{align}
\mathcal{J}^{(1)}_{3,i}
&=
\Bigg\{
-\operatorname{tr}(A_{kl})h_n
D_i^\top\Sigma^{-1}S_{rs}\Sigma^{-1}D_i
-\operatorname{tr}(A_{rs})h_n
D_i^\top\Sigma^{-1}S_{kl}\Sigma^{-1}D_i
\nonumber\\
&\qquad
+h_n
\big(Z_{n,i}^\top\Lambda_0^\top\Sigma^{-1}S_{kl}\Sigma^{-1}\Lambda_0 Z_{n,i}\big)
\big(D_i^\top\Sigma^{-1}S_{rs}\Sigma^{-1}D_i\big)
\nonumber\\
&\qquad
+h_n
\big(Z_{n,i}^\top\Lambda_0^\top\Sigma^{-1}S_{rs}\Sigma^{-1}\Lambda_0 Z_{n,i}\big)
\big(D_i^\top\Sigma^{-1}S_{kl}\Sigma^{-1}D_i\big)
\Bigg\},
\label{eqn:J3i-1}
\\[0.5em]
\mathcal{J}^{(2)}_{3,i}
&=
% \left\{
h_n D_i^\top B_{rs,kl}D_i.
% \right\}.
\label{eqn:J3i-2}
\end{align}

%%%%%%%%%%%%%%%%%%%%%%%%%%%%%%%%%%%%%%%%%%%%%%%%%%%%%%%%%%%%
% J4
%%%%%%%%%%%%%%%%%%%%%%%%%%%%%%%%%%%%%%%%%%%%%%%%%%%%%%%%%%%%
\begin{align}
\mathcal{J}^{(1)}_{4,i}
&=
\Bigg\{
-\operatorname{tr}(A_{kl})h_n^{-1}
\Delta_{n,i}^\top\Sigma^{-1}S_{rs}\Sigma^{-1}\Delta_{n,i}
-\operatorname{tr}(A_{rs})h_n^{-1}
\Delta_{n,i}^\top\Sigma^{-1}S_{kl}\Sigma^{-1}\Delta_{n,i}
\nonumber\\
&\qquad
+h_n^{-1}
\big(Z_{n,i}^\top\Lambda_0^\top\Sigma^{-1}S_{kl}\Sigma^{-1}\Lambda_0 Z_{n,i}\big)
\big(\Delta_{n,i}^\top\Sigma^{-1}S_{rs}\Sigma^{-1}\Delta_{n,i}\big)
\nonumber\\
&\qquad
+h_n^{-1}
\big(Z_{n,i}^\top\Lambda_0^\top\Sigma^{-1}S_{rs}\Sigma^{-1}\Lambda_0 Z_{n,i}\big)
\big(\Delta_{n,i}^\top\Sigma^{-1}S_{kl}\Sigma^{-1}\Delta_{n,i}\big)
\Bigg\},
\label{eqn:J4i-1}
\\[0.5em]
\mathcal{J}^{(2)}_{4,i}
&=
% \left\{
2D_i^\top B_{rs,kl}\Delta_{n,i}
+h_n^{-1}\Delta_{n,i}^\top B_{rs,kl}\Delta_{n,i}.
% \right\}.
\label{eqn:J4i-2}
\end{align}

%%%%%%%%%%%%%%%%%%%%%%%%%%%%%%%%%%%%%%%%%%%%%%%%%%%%%%%%%%%%
% Remainder of \Sigma
%%%%%%%%%%%%%%%%%%%%%%%%%%%%%%%%%%%%%%%%%%%%%%%%%%%%%%%%%%%%

\begin{align}
\begin{aligned}
    \mathscr R_{\ref{thm:asymp-sigma},i}=& -\frac{\alpha^2}{4}\left(1+\frac{1}{\alpha}\right)
|\Sigma|^{-\alpha/2}\fst{- K_i({\bm\theta}) \exp\fst{\upsilon_i}}\mathcal{J}_{0,i}^{(1)} +\frac{\alpha}{2}\left(1+\frac{1}{\alpha}\right)
|\Sigma|^{-\alpha/2}\fst{- K_i({\bm\theta}) \exp\fst{\upsilon_i}}\mathcal{J}_{0,i}^{(2)}\\
&\qquad\qquad +\mathcal{J}_{1,i}'+ \mathcal{J}_{2,i}'+ \mathcal{J}_{3,i}' +\mathcal{J}_{4,i}'
\end{aligned}\label{eqn:remainder_thm_sigma_asymp}
\end{align}

\begin{align}
\begin{aligned}\label{eqn:expansion_cross_deriv}
\frac{\partial^2 V_{n,i}^\alpha}{\partial \beta_u\,\partial \sigma_{rs}}
&=
\left(1+\frac{1}{\alpha}\right)
|\Sigma|^{-\alpha/2}
\exp\!\left(
-\frac{\alpha}{2}
Q_i
\right)
\\
&\quad\times
\Bigg[
\alpha
\Bigg(
\sqrt{h_n}\,
\fst{\frac{\partial a_{i-1}}{\partial\beta_u}}^\top
\Sigma^{-1}\Lambda_0 Z_{n,i}
+
h_n\,
\fst{\frac{\partial a_{i-1}}{\partial\beta_u}}^\top
\Sigma^{-1}D_i(\bm\beta)
+
\fst{\frac{\partial a_{i-1}}{\partial\beta_u}}^\top
\Sigma^{-1}\Delta_{n,i}
\Bigg)
\\
&\qquad\times
\Bigg\{
\frac{\alpha}{2}
\operatorname{tr}(\Sigma^{-1}S_{rs})
-
\frac{\alpha}{2h_n}
\Big(
h_n\,
Z_{n,i}^\top\Lambda_0^\top
\Sigma^{-1}S_{rs}\Sigma^{-1}\Lambda_0 Z_{n,i}
\\
&\qquad\qquad
+2\sqrt{h_n}\,
Z_{n,i}^\top\Lambda_0^\top
\Sigma^{-1}S_{rs}\Sigma^{-1}\Delta_{n,i}
+2h_n^{3/2}\,
Z_{n,i}^\top\Lambda_0^\top
\Sigma^{-1}S_{rs}\Sigma^{-1}D_i(\bm\beta)
\\
&\qquad\qquad
+\Delta_{n,i}^\top
\Sigma^{-1}S_{rs}\Sigma^{-1}\Delta_{n,i}
+2h_n\,
\Delta_{n,i}^\top
\Sigma^{-1}S_{rs}\Sigma^{-1}D_i(\bm\beta)
\\
&\qquad\qquad
+h_n^2\,
D_i(\bm\beta)^\top
\Sigma^{-1}S_{rs}\Sigma^{-1}D_i(\bm\beta)
\Big)
\Bigg\}
\\
&\qquad
+\alpha
\Bigg(
\sqrt{h_n}\,
\fst{\frac{\partial a_{i-1}}{\partial\beta_u}}^\top
\Sigma^{-1}S_{rs}\Sigma^{-1}\Lambda_0 Z_{n,i}
+
h_n\,
\fst{\frac{\partial a_{i-1}}{\partial\beta_u}}^\top
\Sigma^{-1}S_{rs}\Sigma^{-1}D_i(\bm\beta)
\\
&\qquad\qquad
+
\fst{\frac{\partial a_{i-1}}{\partial\beta_u}}^\top
\Sigma^{-1}S_{rs}\Sigma^{-1}\Delta_{n,i}
\Bigg)
\Bigg].
\end{aligned}
\end{align}

\begin{align}\begin{aligned}
   \mathscr R_{\ref{thm:asymp-indep},i}
&=
\left(1+\frac{1}{\alpha}\right)
|\Sigma|^{-\alpha/2}
\exp\!\left(
-\frac{\alpha}{2h_n}
R_i^\top\Sigma^{-1}R_i
\right)
\\
&\quad\times
\alpha
\Bigg[
h_n\,
\fst{\frac{\partial a_{i-1}}{\partial\beta_u}}^\top
\Sigma^{-1}D_i(\bm\beta)
\left\{
\frac{\alpha}{2}
\operatorname{tr}(\Sigma^{-1}S_{rs})
-
\frac{\alpha}{2}
Z_{n,i}^\top\Lambda_0^\top
\Sigma^{-1}S_{rs}\Sigma^{-1}\Lambda_0 Z_{n,i}
\right\}
\\
&\qquad
+\frac{\alpha}{2}\,
\fst{\frac{\partial a_{i-1}}{\partial\beta_u}}^\top
\Sigma^{-1}\Lambda_0 Z_{n,i}
\Big(
-2\sqrt{h_n}\,
Z_{n,i}^\top\Lambda_0^\top
\Sigma^{-1}S_{rs}\Sigma^{-1}D_i(\bm\beta)
\Big)
% \\
% &\qquad
+
h_n\,
\fst{\frac{\partial a_{i-1}}{\partial\beta_u}}^\top
\Sigma^{-1}S_{rs}\Sigma^{-1}D_i(\bm\beta)
\\
&\qquad
+
\fst{\frac{\partial a_{i-1}}{\partial\beta_u}}^\top
\Sigma^{-1}\Delta_{n,i}
\left\{
\frac{\alpha}{2}
\operatorname{tr}(\Sigma^{-1}S_{rs})
-
\frac{\alpha}{2}
Z_{n,i}^\top\Lambda_0^\top
\Sigma^{-1}S_{rs}\Sigma^{-1}\Lambda_0 Z_{n,i}
\right\}
\\
&\qquad
+
\fst{\frac{\partial a_{i-1}}{\partial\beta_u}}^\top
\Sigma^{-1}S_{rs}\Sigma^{-1}\Delta_{n,i}
\\
&\qquad
-
\frac{\alpha}{2h_n}\,
\fst{\frac{\partial a_{i-1}}{\partial\beta_u}}^\top
\Sigma^{-1}\Lambda_0 Z_{n,i}
\Big(
\Delta_{n,i}^\top
\Sigma^{-1}S_{rs}\Sigma^{-1}\Delta_{n,i}
\\
&\qquad\qquad
+2h_n\,
\Delta_{n,i}^\top
\Sigma^{-1}S_{rs}\Sigma^{-1}D_i(\bm\beta)
+h_n^2\,
D_i(\bm\beta)^\top
\Sigma^{-1}S_{rs}\Sigma^{-1}D_i(\bm\beta)
\Big)
\Bigg]\\
&+\left(1+\frac{1}{\alpha}\right)
|\Sigma|^{-\alpha/2}
\exp\!\left(
-\varrho_i
\right)(-K_i({\bm\theta}))
\\
&\quad\times
\alpha\,\sqrt{h_n}
\Bigg[
\fst{\frac{\partial a_{i-1}}{\partial\beta_u}}^\top
\Sigma^{-1}\Lambda_0 Z_{n,i}
\left\{
\frac{\alpha}{2}
\operatorname{tr}(\Sigma^{-1}S_{rs})
-
\frac{\alpha}{2}
Z_{n,i}^\top\Lambda_0^\top
\Sigma^{-1}S_{rs}\Sigma^{-1}\Lambda_0 Z_{n,i}
\right\}
\\
&\qquad
+
\fst{\frac{\partial a_{i-1}}{\partial\beta_u}}^\top
\Sigma^{-1}S_{rs}\Sigma^{-1}\Lambda_0 Z_{n,i}
\Bigg].
\end{aligned}\label{eqn:asymp_indep_rem}
\end{align}

\end{document}